\documentclass[11pt,a4paper]{article}
\RequirePackage{tikz}

 \usepackage{geometry}
\newgeometry{ hmargin={15mm,15mm}}

\usepackage{algorithm}
\usepackage{amsthm}
\usepackage{amsmath}
\usepackage{amssymb}
\usepackage{bm}
\usepackage{booktabs}
\usepackage{caption,subcaption}

\usepackage[numbers]{natbib}
\usepackage{graphicx}

\usepackage{tikz}
\usetikzlibrary{arrows.meta,
                chains,
                positioning,
                shapes.geometric
                }

\DeclareMathOperator*{\argmin}{argmin}
\DeclareMathOperator*{\minimize}{minimize}
\DeclareMathOperator*{\Null}{Null}
\DeclareMathOperator{\diag}{diag}

\DeclareMathOperator{\Var}{Var}

\newtheorem{cor}{Corollary}

\newtheorem{prop}{Proposition}

\DeclareMathOperator{\bA}{\bf A}
\DeclareMathOperator{\ba}{\bf a}
\DeclareMathOperator{\blb}{\bf b}
\DeclareMathOperator{\bd}{\bf d}
\DeclareMathOperator{\bdA}{\bf dA}
\DeclareMathOperator{\bdb}{\bf db}
\DeclareMathOperator{\bdG}{\bf dG}
\DeclareMathOperator{\bdh}{\bf dh}
\DeclareMathOperator{\bdy}{\bf dy}
\DeclareMathOperator{\bdV}{\bf dV}
\DeclareMathOperator{\bdz}{\bf dz}
\DeclareMathOperator{\be}{\bf e}
\DeclareMathOperator{\bG}{\bf G}
\DeclareMathOperator{\bg}{\bf g}
\DeclareMathOperator{\bH}{\bf H}
\DeclareMathOperator{\bh}{\bf h}
\DeclareMathOperator{\bI}{\bf I}
\DeclareMathOperator{\bF}{\bf F}
\DeclareMathOperator{\bM}{\bf M}
\DeclareMathOperator{\bP}{\bf P}
\DeclareMathOperator{\bu}{\bf u}
\DeclareMathOperator{\bV}{\bf V}
\DeclareMathOperator{\bw}{\bf w}
\DeclareMathOperator{\bX}{\bf X}
\DeclareMathOperator{\bx}{\bf x}
\DeclareMathOperator{\bY}{\bf Y}
\DeclareMathOperator{\by}{\bf y}
\DeclareMathOperator{\bz}{\bf z}

\DeclareMathOperator{\bone}{ \boldsymbol{1} }
\DeclareMathOperator{\bepsilon}{ \boldsymbol{\epsilon} }
\DeclareMathOperator{\bdlambda}{ \bd \boldsymbol{\lambda} }
\DeclareMathOperator{\bdnu}{ \bd \boldsymbol{\nu}}
\DeclareMathOperator{\blambda}{ \boldsymbol{\lambda} }
\DeclareMathOperator{\bnu}{ \boldsymbol{\nu}}
\DeclareMathOperator{\bSigma}{ \boldsymbol{\Sigma} }
\DeclareMathOperator{\btheta}{ \boldsymbol{\theta} }
\DeclareMathOperator{\bTheta}{\Theta}

\providecommand{\keywords}[1]{\textbf{\textit{Keywords:}} #1}

\newcommand{\toi}[2][i]{%
  \mathop{
    \mathrm{#2}^{( #1 )}
  }
}

\newcommand{\toit}[2][i]{%
  \mathop{
    \mathrm{#2}^{{T}^{( #1 )}}
  }
}

\newcommand{\toinv}[2][i]{%
  \mathop{
    \mathrm{#2}^{{(#1)}^{-1 }}
  }
}

\begin{document}

\title{Integrating prediction in mean-variance portfolio optimization}




\author{Andrew Butler and Roy H. Kwon \\ University of Toronto\\Department of Mechanical and Industrial Engineering}

\maketitle

\begin{abstract}
Prediction models are traditionally optimized independently from their use in the asset allocation decision-making process. We address this shortcoming and present a framework for integrating regression prediction models in a mean-variance optimization (MVO) setting. Closed-form analytical solutions are provided for the unconstrained and equality constrained MVO case. For the general inequality constrained case, we make use of recent advances in neural-network architecture for efficient optimization of batch quadratic-programs. To our knowledge, this is the first rigorous study of integrating prediction in a mean-variance portfolio optimization setting. We present several historical simulations using both synthetic and global futures data to demonstrate the benefits of the integrated approach.

\end{abstract}


\section{Introduction} \label{sec:intro}
Many problems in quantitative finance can be characterized by the following elements:
\begin{enumerate}
\item A sample data set $ \bY = \{   \by^{(1)},...,   \by^{(m)} \} = \{  \toi{\by} \}_{i=1}^m$ of uncertain quantities of interest, $  \toi{\by} \in \mathbb{R}^{d_y}$, such as asset returns.
\item A decision, $ \bz \in \mathbb{R}^{d_z}$, often constrained to some feasible region $\mathbb{S} \subseteq \mathbb{R}^{d_z}$.
\item An \textit{objective (cost) function}, $c \colon \mathbb{R}^{d_z} \times \mathbb{R}^{d_y} \to \mathbb{R}$,  to be minimized over decision variable $ \bz \in \mathbb{S}$ in the context of the observed realization $ \toi{\by}$.
\end{enumerate}

For example, in portfolio management we are often presented with the following problem:  for a particular observation of asset returns, $ \toi{\by}$, the objective is to construct a vector of assets weights, $ \bz^*( \toi{\by})$, that minimizes the  cost, $c( \bz, \toi{\by})$ and  adheres to the constraint set $\mathbb{S}$. A common choice for  cost is the Markowitz mean-variance quadratic objective \citep{Markowitz1952}, with typical constraints being that the weights be non-negative and sum to one. Of course, the realization of asset returns, $\{  \toi{\by}\}_{i=1}^m$, are not directly observable at decision time and instead must be estimated through associated feature data $ \bX = \{   \bx^{(1)},...,  \bx^{(m)} \}$, of covariates of $ \bY$. Let $f\colon \mathbb{R}^{d_x} \times \mathbb{R}^{d_\theta} \to \mathbb{R}^{d_y}$ denote the $ \btheta$-parameterized prediction model for estimating $\hat{ \by}$. In this paper, we consider regression prediction models of the form:
$$ \toi{\hat{ \by}} = f( \toi{\bx}, \btheta) =   \btheta^T   \toi{\bx},$$
with regression coefficient matrix $ \btheta \in \mathbb{R}^{d_x \times d_y}$.

In most applications, estimating $ \toi{\hat{ \by}}$ requires solving an independent \textit{prediction optimization problem} over the prediction model parameter $\btheta$.  Continuing with the example above; in order to generate mean-variance efficient portfolios we must supply, at a minimum, an estimate of expected asset returns and covariances. A prototypical framework would first estimate the conditional expectations of asset returns and covariances by ordinary least-squares (OLS) regression and then `plug-in' those estimates to a mean-variance quadratic program (see for example \citet{Goldfarb2003}, \citet{Clarke2005} \citet{Chen2015}).

As exemplified above, prediction and decision-based optimization are often decoupled processes; first predict, then optimize. Indeed a perfect prediction model $(\toi{\hat{ \by}} = \toi{\by} )$ would invariably lead to optimal decision-making. In reality, however, prediction models rarely have perfect accuracy and as such an  inefficiency exists in the `predict, then optimize' paradigm; prediction models are estimated  in order to produce `optimal' predictions, not  `optimal' decisions.  

In this paper, we follow the work of \citet{Donti2017}, \citet{Elma2020} and others, and propose the use of an integrated prediction and optimization (IPO) framework with direct applications to mean-variance portfolio optimization.  Specifically,  we estimate $ \btheta$ such that the resulting optimal decisions, $\{ \bz^*(  \toi{ \hat{ \by} } ) \}_{i=1}^m$, minimizes the expected realized  decision cost:
\begin{equation} \label{eq:ed1}
\begin{split}
\minimize_{  \btheta \in  \bTheta} \quad & L(  \btheta) = \mathbb{E}[ c(  \bz^*(\hat{ \by}),   \by) ]  \\
\text{subject to }\quad &    \bz^*( \hat{ \by})   = \argmin_{  \bz \in \mathbb{S}}  c(  \bz,  \hat{ \by} ),
\end{split}
\end{equation}

Solving Program $\eqref{eq:ed1}$ challenging for several reasons. First, even in the case where the decision program is convex, the resulting integrated program is likely not convex in $ \btheta$ and therefore we have no guarantee that a particular local solution is globally optimal. Secondly,  as outlined by  \citet{Donti2017}, in the case where $L(  \btheta)$ is differentiable, computing the gradient, $\nabla_{  \btheta} L$, remains difficult as it requires differentiation through the $\argmin$ operator. Moreover, solving program $\eqref{eq:ed1}$ through iterative descent methods can be computationally demanding as at each iteration we must solve several instances of $ \bz^*( \toi{\hat{ \by} }).$ 

In this paper we address the aforementioned challenges and provide an efficient framework for integrating linear regression predictions into a mean-variance portfolio optimization.  The remainder of the paper is outlined as follows. We first review the growing body of literature in the field of integrated methods and summarize our primary contributions. In Section \ref{sec:method} we present the mean-variance portfolio optimization problem and provide the IPO formulation. We review the current state-of-the-art approach for locally solving Program $\eqref{eq:ed1}$ in the presence of lower-level inequality constraints. We then consider several special instances of the IPO mean-variance optimization problem. In particular, we demonstrate that when the  MVO program is either unconstrained or contains only linear equality constraints then the IPO problem can be recast as a convex quadratic program and solved analytically. We discuss the sampling distribution properties of the optimal IPO regression coefficients and demonstrate that the IPO solution explicitly minimizes the tracking-error to ex-post optimal mean-variance portfolios. 

In Section \ref{sec:sims_1} we perform several simulation studies, using synthetically generated data, and compare the IPO approach to a traditional `predict, then optimize' framework with prediction models estimated by OLS regression. In Sections \ref{sec:sims_2} we discuss the computational challenges of the state-of-the art iterative solution. We demonstrate the computational advantage of the closed-form IPO solution, which guarantees optimality and is approximately an order of magnitude more computationally efficient. In Section \ref{sec:sims_3} we present a simulation that demonstrates that a heuristic analytical IPO solution, with inequality constraints removed, can provide competitive out-of-sample performance and lower variance over a wide range of problem parameterizations.  We conclude in Section \ref{sec:results}  with a historical analysis using global futures data and demonstrate that the IPO framework can provide lower realized costs and improved economic outcomes in comparison to the `predict, then optimize' alternative.

\subsection{Existing Literature}
In recent years there has been a growing body of research on methods for integrating prediction models with downstream decision-making processes.  For example, \citet{Ban2019} present a direct empirical risk minimization approach using nonparametric kernel regression as the core prediction method. They consider a data-driven newsvendor problem and demonstrate that their approach outperforms the `best-practice benchmark' when evaluated out-of-sample. More recently, \citet{Kannan2020} present three frameworks for integrating machine learning prediction models within a stochastic optimization setting. Their primary contribution is in using the out-of-sample residuals from leave-one-out prediction models to generate scenarios which are then optimized in the context of a sample average approximation program. Their frameworks are flexible and accommodate parametric and nonparametric prediction models, for which they derive convergence rates and finite sample guarantees.

\citet{Bert2020} present a general framework for optimizing a conditional stochastic approximation program whereby the conditional density is estimated through a variety of parametric and nonparametric machine learning methods. They generate locally optimal decision policies within the context of the decision optimization problem and consider the setting where the decision policy affects subsequent realizations of the uncertainty variable. They also consider an empirical risk minimization framework for generating predictive prescriptions and discuss the relative trade-offs of such an approach

Recently, \citet{Elma2020} proposed replacing the prediction-based loss function with a convex surrogate loss function that optimizes prediction variables based on the decision error induced by the prediction. They demonstrate that their `smart predict, then optimize' (SPO) loss function attains Fisher consistency with the least-squares loss function and show through example that optimizing predictions in the context of decision objectives and constraints can lead to improved overall decision error. The SPO loss function however is limited to linear objective functions, and despite convexity can be computationally demanding due to repeatedly solving the decision program.

Our approach is most similar to, and is largely inspired by, the work of \citet{Amos2017} and \citet{Donti2017}. Recall that computing the Jacobian, $\partial  \bz^*/\partial \btheta$, is complicated by the bi-level structure of Program $\eqref{eq:ed1}$. \citet{Amos2017} present an efficient framework for embedding quadratic programs as differentiable layers in a neural network. The author's demonstrate  that for linearly constrained quadratic programs, implicit differentiation of the KKT optimality conditions provides the necessary ingredients for computing the desired gradient, $\partial L/\partial  \btheta$. \citet{Donti2017} present the first direct application of the aforementioned work and propose an end-to-end stochastic programming approach for estimating the parameters of probability density functions in the context of their final task-based loss function. They consider applications from power scheduling and battery storage and demonstrate that their `task-based end-to-end' approach can result in lower out-of-sample decision costs in comparison to traditional maximum likelihood estimation and a black-box neural network.

\subsection{Main Contributions}
While our methodology follows closely to that of \citet{Donti2017} and \citet{Elma2020}, in this paper we provide several notable differences and extensions. 

\begin{enumerate}
\item We consider linear regression prediction models with a downstream quadratic MVO objective function. We demonstrate that when the MVO program is either unconstrained or contains only linear equality constraints then the integrated program can be recast as quadratic program. We discuss the necessary conditions for convexity and provide analytical solutions for the optimal IPO  coefficients, $\btheta^*$. We present conditions for which $\btheta^*$ is an unbiased estimator of $\btheta$ and derive the analytical expression for the variance.  We demonstrate that the IPO coefficients explicitly minimize the tracking error to the unconstrained ex-post optimal MVO portfolio and provide the equivalent minimum-tracking error optimization program.   

\item We conduct three simulation studies based on synthetically generated data. The first simulation compares the out-of-sample performance of the IPO and OLS models under varying degrees of estimation error. We demonstrate that for unconstrained and equality constrained cases, the IPO model can produce consistently lower out-of-sample decision costs. The second simulation demonstrates the computational advantage of the analytical IPO solution over a wide range of asset universe sizes. The third simulation considers linear inequality constrained  MVO program under varying degrees of model misspecification. We propose approximating the non-convex problem with the analytical IPO solution whereby the inequality constraints are ignored.  We demonstrate the computational and performance advantage of the analytical IPO solution, which is on average $100$x - $1000$x times faster than the current state-of-the-art method and produces solutions with lower out-of-sample variance and, in some instances, improved MVO costs.

\item	We perform several historical simulations using global futures data, considering both unconstrained and constrained MVO programs and univariate and multivariate regression models. Out-of-sample numerical results demonstrate that the IPO model can provide lower realized cost and superior economic performance in comparison to the traditional  OLS `predict then optimize' approach.


\end{enumerate}

Finally we note that in this paper, asset mean returns are estimated through linear regression, whereas the asset covariance matrices are estimated by a traditional weighted moving average approach \citep{Bauwens2003, Zum2006}. This is supported by the observation that asset mean returns are both nonstationary and heterogeneous and are therefore likely to be dependent on feature data \citep{Engle1982,Hsu1974,Officer1971}, whereas variance and covariances are typically much more stable and exhibit strong autocorrelation effects \citep{Bollerslev1986,Drees2002,Starica2005}. Moreover, \citet{Chopra1993} and \citet{Best1991} report that MVO portfolio weights are an order of magnitude more sensitive to the estimate of asset mean returns compared to estimates of asset covariances. The choice of linear regression model is deliberate and motivated by the long established history of regression forecasting in the financial literature (see for example, \citep{Fama1992,Fama1993,Fama2015}).  Indeed, asset returns are often characterized as time-varying and reactive, and typically exhibit extremely low signal-to-noise ratios (SNRs) \citep{Israel2020}. As a result, low variance models, like simple linear regression, tend to generalize out-of-sample and are often preferred over models of higher complexity. 

:

\section{Methodology}\label{sec:method}

\subsection{IPO: Mean-Variance Optimization}\label{sec:mvo}
We begin with a brief introduction to mean-variance portfolio optimization. We denote the matrix of (excess) return observations as $\bY=[\by^{(1)},\by^{(2)},...,\by^{(m)} ] \in \mathbb{R}^{m\times d_z}$ and denote the portfolio at time $i$ as $\toi{\bz} \in \mathbb{R}^{d_z}$ . Let $\toi{\bV}  \in \mathbb{R}^{d_z \times d_z}$ denote the time-varying symmetric positive definite covariance matrix of asset returns. The mean variance cost function at time $i$ is given by:
\begin{equation}\label{eq:mvo_cost}
c(\bz, \toi{\by} ) =    -\bz^T \toi{\by} + \frac{\delta}{2} \bz^T \toi{\bV} \bz
\end{equation}
with risk-aversion parameter $\delta \in \mathbb{R}_+$ and denote the optimal portfolio weights as:
\begin{equation}\label{eq:mvo}
\begin{split}
\bz^*( \toi{\by} ) & = \argmin_{\bz \in \mathbb{S}}  -\bz^T \toi{\by} + \frac{\delta}{2} \bz^T \toi{\bV} \bz.
\end{split}
\end{equation}

In reality, we do not know the values $\toi{\by}$ or $\toi{\bV}$ at decision time. In this paper we model the time-varying covariance matrix using a weighted moving average approach and denote the covariance estimate as $\toi{\hat{\bV}}$.  Asset returns are modelled according to the following linear model:
\begin{equation}\label{eq:y_i_general}
\toi{\by} = \bP \diag(\toi{\bx}) \btheta + \toi{\bepsilon}
\end{equation}
with residuals $\toi{\bepsilon} \sim \mathcal{N}(\bm 0, \bSigma )$. Here $\diag(\cdot)$ denotes the usual diagonal operator and $\bP \in \mathbb{R}^{d_y \times d_x}$ controls the regression design with each element $\bP_{jk} \in \{0,1\}$. In particular, we assume that each asset has its own, perhaps unique, set of feature variables. For example, if the feature variables represent price-to-earnings (P/E) and debt-to-equity (D/E) ratios for each asset under consideration, then it would be unrealistic to model a particular assets return as a function of all available P/E and D/E ratios. Indeed, doing so would almost certainly lead to model overfit. Instead, we choose to model asset $j$'s return as a linear function of the P/E and D/E ratios relevant to asset $j$, specifically:
\begin{equation}\label{eq:y_j_general}
\toi{\hat{\by}_j}  = \btheta^T_{\ba(j)} \toi{\bx_{\ba(j)}},
\end{equation}
where $\ba(j)$ denotes the indices of the feature variables relevant to asset $j$. Therefore,
\begin{equation}\label{P}
    \bP_{jk} =
    \begin{cases}
      1, & \text{if}\ k \in \ba(j) \\
      0, & \text{otherwise}
    \end{cases}
  \end{equation}
and for observation $i$, the regression estimate of asset expected returns is given by:
\begin{equation}\label{eq:y_hat_i_general}
\toi{\hat{\by}} = f(\toi{\bx},\btheta) =  \bP \diag(\toi{\bx}) \btheta.
\end{equation}

Given that $\toi{\by}$ and $\toi{\bV}$ are unobservable, it follows that in practice portfolio managers solve the MVO program under the estimation hypothesis:
\begin{equation}\label{eq:mvo_cost_hat}
\begin{split}
\minimize_{ \bz \in \mathbb{S}} c(\bz, \toi{\hat{\by}} )&  = -\bz^T \toi{\hat{\by}} + \frac{\delta}{2} \bz^T \toi{\hat{\bV}} \bz =  -\bz^T \bP \diag(\toi{\bx}) \btheta + \frac{\delta}{2} \bz^T \toi{\hat{\bV}} \bz
\end{split}
\end{equation}

In a `predict, then optimize' parameter estimation, $ \btheta$ would be chosen in order to minimize a prediction loss function $\ell \colon \mathbb{R}^{d_y} \times \mathbb{R}^{d_y} \to \mathbb{R}$, such as least-squares. We denote $\mathbb{E}_D$ as the expectation operator with respect to the training set $ D = \{(  \toi{\bx},  \toi{\by})\}_{i=1}^m$ and choose $\hat{\btheta}$ such that:
\begin{equation} \label{theta_hat}
\begin{split}
    \hat{ \btheta} = \argmin_{  \btheta \in  \bTheta}  \mathbb{E}_D[\ell (f( \toi{\bx},  \btheta),  \toi{\by})],
\end{split}
\end{equation}
A `predict, then optimize' framework, would simply `plug-in' the estimate, $\toi{\hat{\by}} = \bP \diag(\toi{\bx})    \hat{ \btheta} $, into program $\eqref{eq:mvo_cost_hat}$ in order to generate the optimal decisions $ \bz^*(\toi{\hat{\by}})$.

Conversely, in the IPO framework, the objective is to choose $\btheta$ in order to minimize the average MVO cost induced by the optimal decisions $\{\bz^*( \toi{\hat{\by}} )\}_{i=1}^m$. Specifically, we solve the bi-level optimization program $\eqref{eq:ed1}$, presented in discrete form in program $\eqref{eq:mvo_full}$:

\begin{equation} \label{eq:mvo_full}
\begin{split}
\minimize_{ \btheta \in \bTheta} \quad & L(\btheta) = \frac{1}{m} \sum_{i = 1}^m \Big(   -\bz^*(\toi{\hat{\by}}) ^T \toi{\by} + \frac{\delta}{2} \bz^*(\toi{\hat{\by}}) ^T \toi{\bV} \bz^*(\toi{\hat{\by}})  \Big)  \\
\text{subject to }\quad  &     \bz^*(\toi{\hat{\by}})   = \argmin_{\bz \in \mathbb{S}}  -\bz^T \bP \diag(\toi{\bx}) \btheta + \frac{\delta}{2} \bz^T \toi{\hat{\bV}} \bz \quad \forall i=1,...,m.
\end{split}
\end{equation}

Note at this point we have not described the feasible region, $\mathbb{S}$, of the MVO program. In the following subsections we briefly discuss the general case where $\mathbb{S}$ describes a set of linear equality and inequality constraints and formalize the current state-of-the-art neural-network framework. We then discuss two special cases in which an analytical solution to the  MVO problem is possible and derive the relevant theory.

\subsection{Current state-of-the-art methodology} \label{sec:ineq}
We begin  with the general case whereby the feasible region of the MVO program is defined by both linear equality and inequality constraints. Specifically we consider the following MVO program:

\begin{equation}\label{eq:mvo_ineq}
\begin{split}
\minimize_{\bz} \quad &  -\bz^T \toi{\hat{\by}} + \frac{\delta}{2} \bz^T \toi{\hat{\bV}} \bz \\
\text{subject to} \quad &   \bA \bz = \blb\\
\quad & \bG \bz \leq \bh
\end{split}
\end{equation}
where $\bA \in \mathbb{R}^{d_{\text{eq}} \times d_{\bz} }$, $\blb \in \mathbb{R}^{d_{\text{eq}}}$ and $\bG \in \mathbb{R}^{d_{\text{iq}} \times d_{\bz}}$, $\bh \in \mathbb{R}^{d_{\text{iq}}}$ describe the linear equality and inequality constraints, respectively.

 In general, there is no known analytical solution to Program $\eqref{eq:mvo_ineq}$ and instead the solution, $\bz^*$ is obtained through iterative optimization methods. Moreover, because of the inequality constraints, the IPO objective, $L(\btheta)$, is generally not a convex function of $\btheta$. Therefore, in the general case we follow  \citet{Amos2017} and \citet{Donti2017} and compute locally optimal solutions, $\btheta^*$, by restructuring Program  $\eqref{eq:mvo_full}$ as an end-to-end neural network and apply (stochastic) gradient descent. The IPO equivalent neural-network structure is depicted in Figure \ref{fig:network}. In the forward pass, the input layer takes the feature variables $\toi{\bx}$ and passes them to a simple linear layer to produce the estimates, $\toi{\hat{\by}}$. The predictions are then passed to a differentiable quadratic programming layer which, for a given input $\toi{\hat{\by}}$, solves the decision program and returns the optimal portfolio weights $\bz^*(\toi{\hat{\by}}).$  Finally, the quality of the portfolio decisions are evaluated by the MVO cost function, $c(\bz^*(\toi{\hat{\by}}), \toi{\by})$ in the context of the true return values  $\toi{\by}$. We refer the reader to Appendix \ref{sec:app_proof} for more comprehensive implementation details.

\begin{figure}[H]
\centering
\begin{tikzpicture}[
    node distance = 5mm and 7mm,
      start chain = going right,
  alg/.style = {draw, align=center, font=\linespread{0.8}\selectfont, minimum width={width("Backward Passing:  ")+7pt}}
                    ]
\begin{scope}[every node/.append style={on chain, join=by -Stealth}]
\node (n1) [alg] {$\quad \toi{\bx} \quad $};
\node (n2) [alg]  {$\toi{\hat{\by}} = \bP \diag(\toi{\bx})\btheta $};
\node (n3) [alg]  {$\min_{\bz} c( \bz,  \toi{\hat{\by}}  )$};
\node (n4) [alg] { $c(\bz^*, \toi{\by})$  };
\end{scope}

\node[above=of n1]  {Input layer};
\node[above=of n2]  {Linear layer};
\node[above=of n3]  {QP layer};
\node[above=of n4]  {Loss Function};
\end{tikzpicture}

\begin{tikzpicture}[
    node distance = 5mm and 7mm,
      start chain = going left,
  alg/.style = {draw, align=center, font=\linespread{0.8}\selectfont, minimum width={width("Backward Passing:  ")+7pt} }
                    ]
\begin{scope}[every node/.append style={on chain, join=by -Stealth}]
\node (n4) [alg] { $\partial c / \partial \bz^*$   };
\node (n3) [alg]  {$\partial \bz^* / \partial \toi{\hat{\by}} $};
\node (n2) [alg]  {$\partial \toi{\hat{\by}} / \partial \btheta $};
\node (n1) [alg] {$\btheta \leftarrow \btheta - \bg_{\btheta}   $};

\end{scope}

\end{tikzpicture}
\caption{IPO program represented as an end-to-end neural-network with predictive linear layer, differentiable quadratic programming layer and realized cost loss function. }
\label{fig:network}
\end{figure}
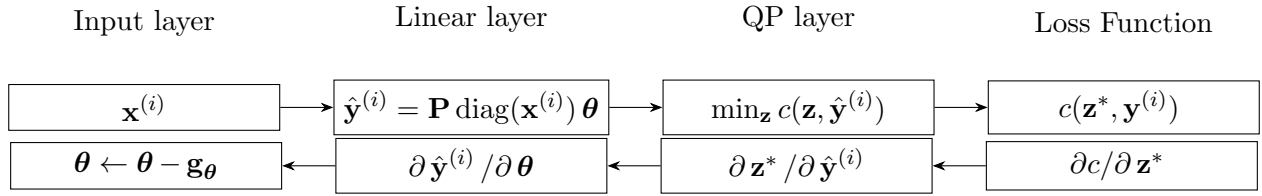

In the general case we compute a locally optimal solution, $\btheta^*$,  by applying (stochastic) gradient descent. Prediction model parameters are updated by backpropogation, with descent direction, $\bg_{\theta}$, estimated over a randomly drawn sample batch, $B$:
$$
\bg_{\theta} = \sum_{i\in B}\Big( \frac{\partial c   }{\partial \btheta} \Big)_{\mid (\bz^*(\toi{\hat{\by}}),\toi{\by} )} \approx \nabla_\theta L.
$$
Note that each iteration of gradient descent therefore requires forward solving and backward differentiating through, at most, $m$ mean-variance optimization programs, $\{ \bz^*(\toi{\hat{\by}}) \}_{i=1}^m$, which in some applications can be computationally expensive to compute. 

\subsection{Special case 1:  $\mathbb{S} = \mathbb{R}^{d_z}$}\label{sec:method_uncon}

We are motivated by \citet{Gould2016} who demonstrate that under special constraint cases, an analytical solution for the gradient and Hessian of a bi-level optimization problem exists. We first consider the case where the  MVO program is unconstrained  and therefore an analytical solution is given by Equation $\eqref{eq:z_star_uncon}$.

\begin{equation}\label{eq:z_star_uncon}
\begin{split}
\bz^*(\toi{\hat{\by}})  & = \frac{1}{\delta}\toinv{\hat{\bV}} \toi{\hat{\by}} =  \frac{1}{\delta} \toinv{\hat{\bV}} \bP \diag(\toi{\bx}) \btheta.
\end{split}
\end{equation}

\begin{prop}\label{prop:l_uncon_uni}
Let $\mathbb{S} = \mathbb{R}^{d_z}$ and $\bTheta = \mathbb{R}^{d_\theta}$. We define

\begin{equation}\label{eq:d_uncon}
\bd(\bx,\by)  = \frac{1}{m\delta} \sum_{i = 1}^m  \Big( \diag(\toi{\bx}) \bP^T \toinv{\hat{\bV}}  \toi{\by} \Big)
\end{equation}
and

\begin{equation}\label{eq:h_uncon}
\bH(\bx)  =  \frac{1}{m \delta} \sum_{i = 1}^m  \Big( \diag(\toi{\bx}) \bP^T \toinv{\hat{\bV}}  \toi{\bV}  \toinv{\hat{\bV}} \bP \diag(\toi{\bx})  \Big).
\end{equation}
Then the IPO program $\eqref{eq:mvo_full}$ is an unconstrained quadratic program (QP) given by:

\begin{equation} \label{eq:l_uncon_uni}
\minimize_{\btheta \in \bTheta} \frac{1}{2}\btheta^T \bH(\bx)  \btheta -  \btheta^T \bd(\bx,\by).
\end{equation}
Furthermore, if there exists an $\toi{\bx}$ such that $\toi{\bx_j} \neq 0 \quad \forall j \in 1,...,d_x$ then $\bH(\bx)  \succ 0$ and therefore program $\eqref{eq:l_uncon_uni}$ is an unconstrained convex quadratic program with unique minimum:
\begin{equation} \label{eq:theta_star}
\btheta^* = \bH(\bx)^{-1}  \bd(\bx,\by).
\end{equation}

\end{prop}

All proofs are provided in Appendix \ref{sec:app_proof}. We make a few important observations. The first, is that for the realistic case where there exists an $\toi{\bx}$ such that each $\toi{\bx_j}$ are not exactly zero, then the optimal IPO regression coefficients, $\btheta^*$, are unique. Furthermore, we observe that the solution is independent of the risk-aversion parameter. This is intuitive, as when the MVO program is unconstrained, then the risk-aversion parameter simply controls the scale of the resulting solutions $\bz^*(\toi{\hat{\by}})$. 

We note that the solution presented in Equation $\eqref{eq:theta_star}$ requires the action of the inverse of the Hessian: $\bH(\bx)$. In many applications of machine learning, such as computer vision or statistical meta-modelling, it is difficult, if not impossible, to solve the inverse problem without customized algorithms or prior knowledge of the data (see for example \citet{Jones1994}, \citet{Ranjan2011}, \citet{Ongie2020}). In many cases, the dimension of the relevant Hessian is either too large for both forward-mapping and inversion in reasonable compute time or is computationally unstable due to near-singularity. In our IPO framework, we fortunately do not encounter these technical difficulties surrounding the action of the inverse. In most practical settings, the dimension of the Hessian matrix, is on the order of $10$ or $100$, whereas the number of observations, $m$, is on the order of $1000$ or $10000$. The Hessian is therefore likely to be computationally stable and the action of the inverse is computationally tractable. This is validated numerically in Section $\ref{sec:sims_2}$ and we demonstrate the computational advantage of the analytical solution over the iterative descent method. 

Furthermore, while outside of the scope of the current paper, we note that under the QP formulation  $\eqref{eq:l_uncon_uni}$, it is trivial to incorporate both regularization and constraints on $\btheta$. This is demonstrated by Program $\eqref{eq:l_uncon_uni_reg}$:

\begin{equation} \label{eq:l_uncon_uni_reg}
\begin{split}
\minimize_{\bz} \quad &  \frac{1}{2}\btheta^T \bH(\bx)  \btheta -  \btheta^T \bd(\bx,\by) + \Omega(\mid \btheta \mid) \\
\text{subject to} \quad &   \bA_{\btheta} \btheta = \blb_{\btheta}\\
\quad & \bG_{\btheta} \btheta \leq \bh_{\btheta}
\end{split}
\end{equation}
where $\Omega \colon \mathbb{R}^{d_\theta} \to \mathbb{R}$ is a convex regularization function. In most cases, Program $\eqref{eq:l_uncon_uni_reg}$ can be solved to optimality by standard quadratic programming techniques, whereas the incorporation of regularization and constraints in the current state-of-the-art solution is structurally more challenging.

We now discuss the properties of the sampling distribution of the IPO parameter estimate, $\btheta^*$, and derive an estimate of the variance, $\Var{(\btheta^*)}$. Recall, from Equation $\eqref{eq:y_i_general}$ we have: $\toi{\by} \sim \mathcal{N}(\bP \diag(\toi{\bx}) \btheta , \bSigma )$.

\begin{prop}\label{prop:uncon_bias_general}
Let
\begin{equation}\label{eq:du}
 \bd_{\bu}(\bx) = \frac{1}{m\delta} \sum_{i = 1}^m  \Big( \diag(\toi{\bx}) \bP^T \toinv{\hat{\bV}} \bP \diag(\toi{\bx}) \Big),
 \end{equation}
then the optimal IPO estimate, $\btheta^*$, is a biased estimate of $\btheta$ with bias $\bH(\bx)^{-1}\bd_{\bu}(\bx)$.
\end{prop}


\begin{cor}\label{cor:uncon_unbias_general}
Let  $\btheta_{\bu}^* = \bd_{\bu}(\bx)^{-1}\bH(\bx)\btheta^*$. Then $\btheta_{\bu}^*$ is an unbiased estimator of $\btheta$.
\end{cor}

\begin{cor}\label{cor:uncon_unbias}
Let $\toi{\hat{\bV}} = \toi{\bV} \forall i \in \{1, ..., m\}$. Then $\btheta^*$ is an unbiased estimator of $\btheta$.
\end{cor}


We observe from Proposition \ref{prop:uncon_bias_general} that differences, or estimation errors, between $\toi{\hat{\bV}}$ and $\toi{\bV}$, make $\btheta^*$ a biased estimator in $\btheta$. In particular, the bias can be corrected by left multiplication of $\btheta^*$ by $\bd_{\bu}(\bx)^{-1} \bH(\bx)$. This observation leads to Corollary \ref{cor:uncon_unbias}, which shows that when the estimation error in the covariance is zero then $\btheta^*$ is an unbiased estimator of $\btheta$. Moreover, unlike the OLS estimate, $\hat{\btheta}$, the IPO estimate, $\btheta^*$, incorporates estimation error in the sample covariance in the (likely) event that the estimation error is nonzero. This is discussed in more detail in Section \ref{sec:sims_1}.


\begin{prop}\label{prop:uncon_var}
Let $\{\toi{\by}\}_{i=1}^m$ be independent random variables with  $\toi{\by} \sim \mathcal{N}(\bP \diag(\toi{\bx}) \btheta , \bSigma )$. Let $\hat{\bSigma}$ be an unbiased estimate of the sample covariance of residuals, given by:
\begin{equation}\label{eq:sigma_hat}
\hat{\bSigma} = \frac{1}{m-1} \sum_{i = 1}^m \Big( \toi{\by} - \bP \diag(\toi{\bx})\btheta \Big)^2.
\end{equation}
Let
\begin{equation}\label{eq:m_uncon}
\bM = \frac{1}{\delta^2 m^2} \Big( \diag(\toi{\bx}) \bP^T \toinv{\hat{\bV}} \hat{\bSigma} \toinv{\hat{\bV}} \bP \diag(\toi{\bx}) \Big),
\end{equation}
then the variance, $\Var(\btheta^*)$, is given by:
\begin{equation}\label{eq:uncon_var}
\Var(\btheta^*) = \bH(\bx)^{-1} \bM \bH(\bx)^{-1}
\end{equation}
\end{prop}
We conclude this section by providing an alternative, and perhaps more intuitive expression of the optimal IPO coefficients derived from portfolio tracking-error optimization. Let $\lVert \cdot \rVert_{\bV}$ denote the elliptic norm with respect to the symmetric positive definite matrix $\bV$, defined as:

\begin{equation}\label{eq:a_norm}
\lVert \bw \rVert_{\bV} =\sqrt{\bw^T \bV \bw}.
\end{equation}
More specifically, $\lVert \toi[1]{\bz} - \toi[2]{\bz} \rVert^2_{\bV}$ measures the tracking-error between two portfolio weights with respect to the covariance $\bV$.

\begin{prop}\label{prop:uncon_tracking_error}
Let $\bz^*(\toi{\by})$ and $\bz^*(\toi{\hat{\by}})$ be as defined in Equation $\eqref{eq:mvo}$ and Equation $\eqref{eq:z_star_uncon}$, respectively. Then the optimal IPO coefficients, $\btheta^*$, minimizes the average tracking error between $\bz^*(\toi{\by})$ and $\bz^*(\toi{\hat{\by}})$ with respect to the realized covariance $\toi{\bV}$:
\begin{equation} \label{eq:uncon_tracking_error}
\btheta^* = \argmin_{\btheta \in \bTheta} \frac{1}{2m} \sum_{i = 1}^m \lVert \bz^*(\toi{\hat{\by}}) - \bz^*(\toi{\by})  \rVert^2_{\toi{\bV}}
\end{equation}
\end{prop}

Indeed, Proposition \ref{prop:uncon_tracking_error}  states that the IPO coefficients $\btheta^*$ minimizes the average tracking error between the estimated optimal weights, $\bz^*(\toi{\hat{\by}})$ and the ex-post optimal weight $\bz^*(\toi{\by})$.



\subsection{Special Case 2: $\mathbb{S} =\{\bA \bz = \blb\}$}\label{sec:method_eqcon}
We now consider the case where the MVO program is constrained by a set of linear equality constraints:
$$
\mathbb{S} =\{\bA \bz = \blb\},
$$
where $\bA \in \mathbb{R}^{d_{\text{eq}} \times d_{z} }$ and $\blb \in \mathbb{R}^{d_{\text{eq}}}$. We assume the non-trivial case where $\bA$ is not full rank.  Let the columns of $\bF$ form a basis for the nullspace of $\bA$ defined as:
$$\Null(\bA) = \{\bz \in \mathbb{R}^{d_z}\mid\bA \bz =0 \}.$$
Let $\bz_0$ be a particular element of $\mathbb{S}$. It follows that $\forall \ \bw \in \mathbb{R}^{d_{\bz}-d_n}$ then $\bz = \bF \bw + \bz_0$ is also an element of $\mathbb{S}$, where $d_n = \text{nullity}(\bA)$. We follow \citet{Boyd2004} and recast the MVO program as an unconstrained convex quadratic program:

\begin{equation}\label{eq:a_eq_recast_2}
\min_{\bw} c(\bF \bw + \bz_0, \toi{\hat{\by}} ),
\end{equation}
with unique global minimum:
\begin{equation}\label{eq:a_w_star_2}
\begin{split}
\bw^*(\toi{\hat{\by}}) & = \frac{1}{\delta}( \bF^T \toi{\hat{\bV}} \bF )^{-1} \bF^T \Big( \toi{\hat{\by}}  - \delta \toi{\hat{\bV}} \bz_0 \Big)
\end{split}
\end{equation}
The solution to the MVO Program $\eqref{eq:mvo}$ is then given by:
\begin{equation}\label{eq:z_star_eqcon_general}
\begin{split}
\bz^*(\toi{\hat{\by}}) & = \frac{1}{\delta} \bF ( \bF^T \toi{\hat{\bV}} \bF )^{-1} \bF^T ( \bP \diag(\toi{\bx})\btheta - \delta \toi{\hat{\bV}} \bz_0) + \bz_0\\
& = \frac{1}{\delta} \bF ( \bF^T \toi{\hat{\bV}} \bF )^{-1} \bF^T  \bP \diag(\toi{\bx})\btheta +  (\bI - \bF ( \bF^T \toi{\hat{\bV}} \bF )^{-1} \bF^T \toi{\hat{\bV}} )\bz_0
\end{split}
\end{equation}

\begin{prop}\label{prop:l_eqcon_uni}
Let $\mathbb{S} =\{\bA \bz = \blb\}$ and $\bTheta = \mathbb{R}^{d_\theta}$. Define

\begin{equation}\label{eq:d_eqcon}
\bd_{\text{eq}}(\bx,\by)  = \frac{1}{m\delta} \sum_{i = 1}^m  \Big( \diag(\toi{\bx}) \bP^T \bF ( \bF^T \toi{\hat{\bV}} \bF )^{-1} \bF^T (\toi{\by} - \toi{\bV} (\bI - \bF ( \bF^T \toi{\hat{\bV}} \bF )^{-1} \bF^T \toi{\hat{\bV}} )\bz_0)  \Big)
\end{equation}
and
\begin{equation}\label{eq:h_eqcon}
\bH_{\text{eq}}(\bx) =  \frac{1}{m \delta} \sum_{i = 1}^m  \Big( \diag(\toi{\bx}) \bP^T \bF ( \bF^T \toi{\hat{\bV}} \bF )^{-1} \bF^T \toi{\bV} \bF ( \bF^T \toi{\hat{\bV}} \bF )^{-1} \bF^T  \bP \diag(\toi{\bx})  \Big).
\end{equation}
Then the IPO program $\eqref{eq:mvo_full}$ is an unconstrained quadratic program given by:

\begin{equation} \label{eq:l_eqcon_uni}
\minimize_{\btheta \in \bTheta} \frac{1}{2}\btheta^T \bH_{\text{eq}}(\bx)  \btheta -  \btheta^T \bd_{\text{eq}}(\bx,\by).
\end{equation}
Furthermore, if there exists an $\toi{\bx}$ such that $\toi{\bx_j} \neq 0 \quad \forall j \in 1,...,d_x$ then $\bH_{\text{eq}}(\bx)  \succ 0$ and therefore program $\eqref{eq:l_eqcon_uni}$ is an unconstrained convex quadratic program with unique minimum:
\begin{equation} \label{eq:theta_star_eqcon}
\btheta_{\text{eq}}^* = \bH_{\text{eq}}(\bx)^{-1}  \bd_{\text{eq}}(\bx,\by).
\end{equation}

\end{prop}

As before we briefly discuss the properties of the sampling distribution of the equality constrained IPO parameter estimate, $\btheta_{\text{eq}}^*$, and derive an estimate of the variance, $\Var{(\btheta_{\text{eq}}^*)}$. 

\begin{prop}\label{prop:eqcon_bias_general}
Let
\begin{equation}\label{eq:de}
\bd_{\be}(\bx) = \frac{1}{m\delta} \sum_{i = 1}^m  \Big( \diag(\toi{\bx}) \bP^T \bF ( \bF^T \toi{\hat{\bV}} \bF )^{-1} \bF^T \bP \diag(\toi{\bx})  \Big),
\end{equation}
then the optimal IPO estimate, $\btheta_{\text{eq}}^*$, is a biased estimate of $\btheta$ with bias $\bH_{\text{eq}}(\bx)^{-1}\bd_{\be}(\bx)$.
\end{prop}


\begin{cor}\label{cor:eqcon_unbias}
Let $\toi{\hat{\bV}} = \toi{\bV} \forall i \in \{1, ..., m\}$. Then $\btheta_{\text{eq}}^*$ is an unbiased estimator of $\btheta$.
\end{cor}


Again we observe from Proposition \ref{prop:eqcon_bias_general} that in general $\btheta_{\text{eq}}^*$ a biased estimator of $\btheta$. In particular, the bias can be corrected by left multiplication of $\btheta_{\text{eq}}^*$ by $\bd_{\be}(\bx)^{-1} \bH_{\text{eq}}(\bx)$. Furthermore when the estimation error in the covariance is zero then $\btheta_{\text{eq}}^*$ is an unbiased estimator of $\btheta$.

\begin{prop}\label{prop:eqcon_var}
 Let $\{\toi{\by}\}_{i=1}^m$ be independent random variables with  $\toi{\by} \sim \mathcal{N}(\bP \diag(\toi{\bx}) \btheta , \bSigma )$. Let
 \begin{equation}\label{eq:m_eq}
 \bM_{\text{eq}} = \frac{1}{\delta^2m^2} \sum_{i = 1}^m  \Big( \diag(\toi{\bx}) \bP^T \bF ( \bF^T \toi{\hat{\bV}} \bF )^{-1} \bF^T \hat{\bSigma} \bF ( \bF^T \toi{\hat{\bV}} \bF )^{-1} \bF^T  \bP \diag(\toi{\bx}) \Big).
 \end{equation}
 then the variance, $\Var(\btheta_{\text{eq}}^*)$, is given by:
 \begin{equation}\label{eq:eqcon_var}
 \Var(\btheta_{\text{eq}}^*) = \bH_{\text{eq}}(\bx)^{-1} \bM_{\text{eq}} \bH_{\text{eq}}(\bx)^{-1}
 \end{equation}
\end{prop}

 As before, we conclude this subsection with the following proposition that states that  the IPO coefficients $\btheta^*$ minimizes the average tracking error between the estimated optimal weights, $\bz^*(\toi{\hat{\by}})$ and the ex-post optimal weight $\bz^*(\toi{\by})$. 

\begin{prop}\label{prop:eqcon_tracking_error}
Let $\bz^*(\toi{\by})$ and $\bz^*(\toi{\hat{\by}})$ be as defined in Equation $\eqref{eq:mvo}$ and Equation $\eqref{eq:z_star_uncon}$, respectively. Then the optimal IPO coefficients, $\btheta_{\text{eq}}^*$, minimizes the average tracking error between $\bz^*(\toi{\by})$ and $\bz^*(\toi{\hat{\by}})$ with respect to the realized covariance $\toi{\bV}$:
\end{prop}

\begin{equation}\label{eq:eqcon_tracking_error}
\begin{split}
\btheta_{\text{eq}}^* & = \argmin_{\btheta} \quad   \frac{1}{2m} \sum_{i = 1}^m \lVert \bz^*(\toi{\hat{\by}}) - \bz^*(\toi{\by})  \rVert^2_{\toi{\bV}} \\
& \text{subject to} \quad   \bA \bz^*(\toi{\hat{\by}}) = \blb \\
& \text{subject to} \quad   \bA \bz^*(\toi{\by}) = \blb
\end{split}
\end{equation}

\section{Simulated experiments}\label{sec:sims}
\subsection{Simulation 1: estimation error in $\hat{\bV}$}\label{sec:sims_1}
\citet{Elma2020} consider the integration of predictive forecasting with downstream optimization problems that have linear cost functions. Their simulated experiments demonstrate that the benefit of the `smart predict, then optimize' (SPO) framework increases as the amount of model misspecification increases. Specifically, model misspecification is introduced by synthetically generating cost vectors that are polynomial functions of the simulated feature data and modelling the relationship as though it is linear. In particular, they demonstrate that a linear forecasting model trained with SPO can outperform traditional prediction models, such as OLS and random forest, and the amount of outperformance increases as the degree of nonlinearity in the ground truth increases.

Here, we demonstrate that, for a mean-variance  decision program, the IPO model can provide lower out-of-sample MVO costs in comparison to a traditional OLS-based `predict, then optimize' model, even when the underlying ground truth is \textit{linear} in the feature variables. In particular, we demonstrate that the OLS model is vulnerable to estimation error in $\toi{\hat{\bV}}$, resulting in sub-optimal decision making and increasing out-of-sample MVO costs as estimation error in $\toi{\hat{\bV}}$ increases. The IPO model, on the other hand, incorporates the impact of estimation in the covariance matrix. The simulated experiment below demonstrates that the IPO model consistently outperforms the OLS model in terms of minimizing the out-of-sample MVO cost. Moreover, the outperformance is shown to be consistent over a wide range of signal-to-noise ratios (SNRs) and asset correlation assumptions commonly observed in financial forecasting. In general we observe that the benefit of the IPO model increases as the estimation error in $\toi{\hat{\bV}}$ increases, even when the underlying ground truth is linear in the feature variables.

We follow an experimental design similar to \citet{Hastie2017}. Asset returns are assumed to be normally distributed, $\toi{\by} \sim \mathcal{N}(\diag(\toi{\bx}) \btheta_0, \bV )$ where $\bV \in \mathbb{R}^{d_z \times d_z}$ has entry $(j, k)$ equal to $\sigma^2 \rho^{|j-k|}$, and $\sigma = 0.0125$ ($20\%$ annualized). Asset mean returns are modelled according to univariate model of the form:
$$
\toi{\by} = \diag(\toi{\bx}) \btheta_0 + \tau \toi{\bepsilon},
$$
where feature data, $\toi{\bx} \sim \mathcal{N}(\bm 0, \bI_{d_x})$, and residuals, $\toi{\bepsilon} \sim \mathcal{N}(\bm 0, \bV)$. The scalar value $\tau$ controls the SNR level, where $\text{SNR} = \Var(f(\bx,\btheta_0))/\Var(\bepsilon)$. We consider asset correlation values in the range of:
$
\rho \in \{0,0.25,0.5,0.75\},
$
and SNR values:
$
\text{SNR} \in \{0.001,0.002,0.003,0.004,0.005,0.01,0.05,0.10\}.
$
Note that it may appear that these SNR values are extremely low. However, in most applications of asset return forecasting, the SNRs are typically found to be much less than $1\%$. Indeed, a moderate sized universe (25 assets) with each asset having SNRs of $1\%$ can generate annualized Sharpe ratios in the low double digits - which is extremely rare - and SNRs of $10\%$ are extremely unlikely at a daily trading frequency.

We introduce estimation error in $\toi{\hat{\bV}}$ by varying the sample size, $s = \text{res}*d_z$, used  for estimation. We set the number of assets, $d_z = 10$, and consider resolutions, $\text{res} \in \{5, 10, 20\}$, thus giving covariance sample sizes of $s \in \{50, 100, 200\}$. In all experiments we set the risk aversion parameter $\delta = 1$.

The simulation process can be described as follows:
\begin{enumerate}
\item Generate the ground truth coefficients: $\btheta_0 \sim \mathcal{N}(\bm 0, \bI_{d_\theta} )$.
\item Generate feature variables: $\toi{\bx} \sim \mathcal{N}(\bm 0, \bI_{d_x})$.
\item Generate $2000$ return observations: $\toi{\by} \sim \mathcal{N}(\diag(\toi{\bx}) \btheta,\tau \toi{\bepsilon})$, where $\tau$ is chosen to meet the desired SNR.
\item Divide the sample data into two equally sized disjoint data sets: in-sample and out-of-sample.
\item Generate estimates $\toi{\hat{\bV}}$ using the chosen sample size, $s$.
\item Estimate the optimal OLS and IPO coefficients on the in-sample data.
\item Generate the optimal out-of-sample MVO decisions, $\bz^*(\hat{\by})$, using the covariance estimates $\toi{\hat{\bV}}$ and corresponding optimal regression coefficients for predicting $\toi{\hat{\by}}$.
\item Evaluate several performance metrics (described below) on the out-of-sample data.
\item Repeat steps $1$-$8$ a total of $100$ times and average the results.
\end{enumerate}

\textbf{Performance metrics:}
Let $\btheta_0$ be the ground truth and $\btheta$ denote an estimated (OLS or IPO) regression coefficient. Let $\bV$ denote the true asset covariance and let $\{\toi{\by}\}_{i=1}^m$ denote the realized return observations.
 \begin{itemize}
 \item \textbf{MVO Cost}: Let $\bz^*(\hat{\by})$ be as defined in Equation $\eqref{eq:mvo}$. The out-of-sample MVO cost is then given by:
 \begin{equation} \label{eq:cost_sim_1}
 c(\bz^*(\hat{\by}), \toi{\by} ) =   -\bz^*(\hat{\by})^T \toi{\by} + \frac{\delta}{2} \bz^*(\hat{\by})^T \bV \bz^*(\hat{\by})
 \end{equation}
 \item \textbf{Proportion of variance explained:} a measure of return prediction accuracy on defined as:
 $$
 \text{PVE}(\btheta) = 1 - \mathbb{E}[(\toi{\by} - \diag(\toi{\bx}) \btheta)^2]/\Var{(\toi{\by})}.
 $$
\end{itemize}

We consider the case where the  MVO program contains equality constraints. In particular, we enforce that the sum of the weights must  be equal to one: $\mathbb{S} =\{\bz^T\bone  = 1\}$. Figures \ref{fig:sims_1_mvo_200} and \ref{fig:sims_1_pve_200} report the average and $95\%$-ile range of the out-of-sample MVO costs and PVE values, respectively, as a function of the SNR. Here, the covariance resolution is set to $20$ and therefore the expected estimation error in $\toi{\hat{\bV}}$ is relatively low. As a result, we observe that the difference in both out-of-sample MVO cost and PVE is negligible, with the IPO model producing marginally lower MVO costs and the OLS model producing marginally higher PVE, as expected. Observe that even in the most optimistic case where estimation error in $\toi{\hat{\bV}}$ is low and the ground truth relationship is linear, there is no adverse repercussions in using the IPO model. Furthermore, in order to effectively eliminate estimation error we require a covariance resolution on the order of $20$; which in practical terms implies that for a $100$ asset portfolio we require a sample size of $2000$ return observations. In many forecasting applications a sample size of this magnitude would be impractical and would potentially interfere with the observed time-varying dependency of asset volatilities and correlations  \citep{Engle1982,Bollerslev1986,Bollerslev1988,Bollerslev1990, Bollersev1994}. Furthermore, as estimation error increases, we observe that for the majority of relevant SNRs, the IPO model produces a lower realized out-of-sample MVO cost. In particular Figures \ref{fig:sims_1_mvo_100} and \ref{fig:sims_1_mvo_50} demonstrate a statistically significant reduction in out-of-sample MVO costs as the covariance resolution decreases to $10$ and $5$, respectively. Interestingly, Figures \ref{fig:sims_1_pve_100} and \ref{fig:sims_1_pve_50} demonstrate that the IPO model produces lower realized MVO costs, despite providing lower average prediction accuracy, as measured by PVE. Note that this finding is consistent with the results presented in \citep{Elma2020}. Finally, we observe in Figures \ref{fig:sims_1_mvo_100} and \ref{fig:sims_1_mvo_50} that the benefit of the IPO model is greatest when the ground truth correlation values, $\rho$, are closest to zero. Indeed this is intuitive as the extent of covariance estimation error in both magnitude and sign is largest when correlation values approach zero.

\begin{figure}[H]
  \centering
  \begin{subfigure}[b]{0.35\linewidth}
    \includegraphics[width=\linewidth , trim={00mm 0cm 0cm 0cm},clip]{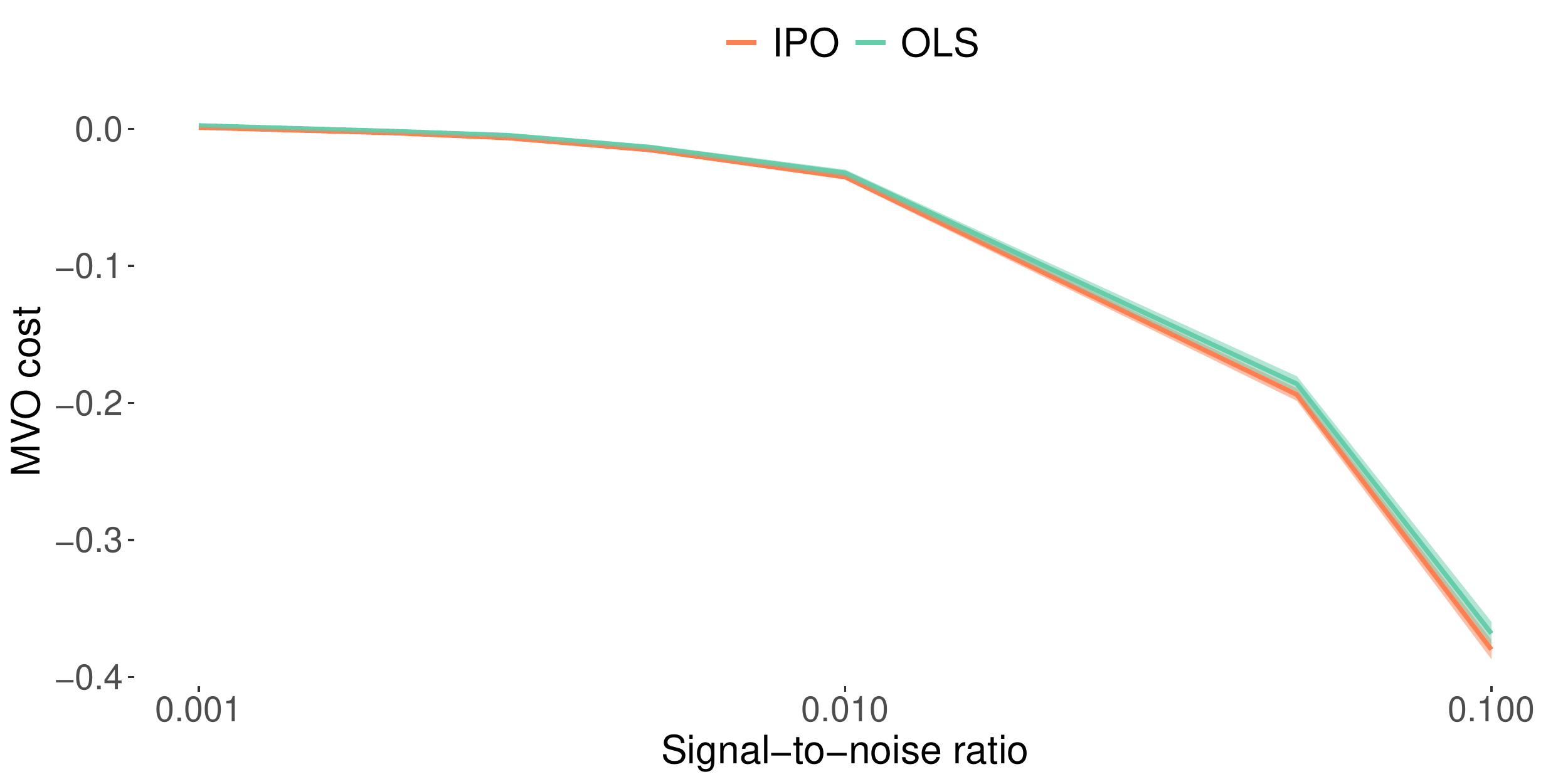}
    \caption{$\rho = 0, \text{res} = 20$.}
  \end{subfigure}
 \begin{subfigure}[b]{0.35\linewidth}
    \includegraphics[width=\linewidth , trim={00mm 0cm 0cm 0cm},clip]{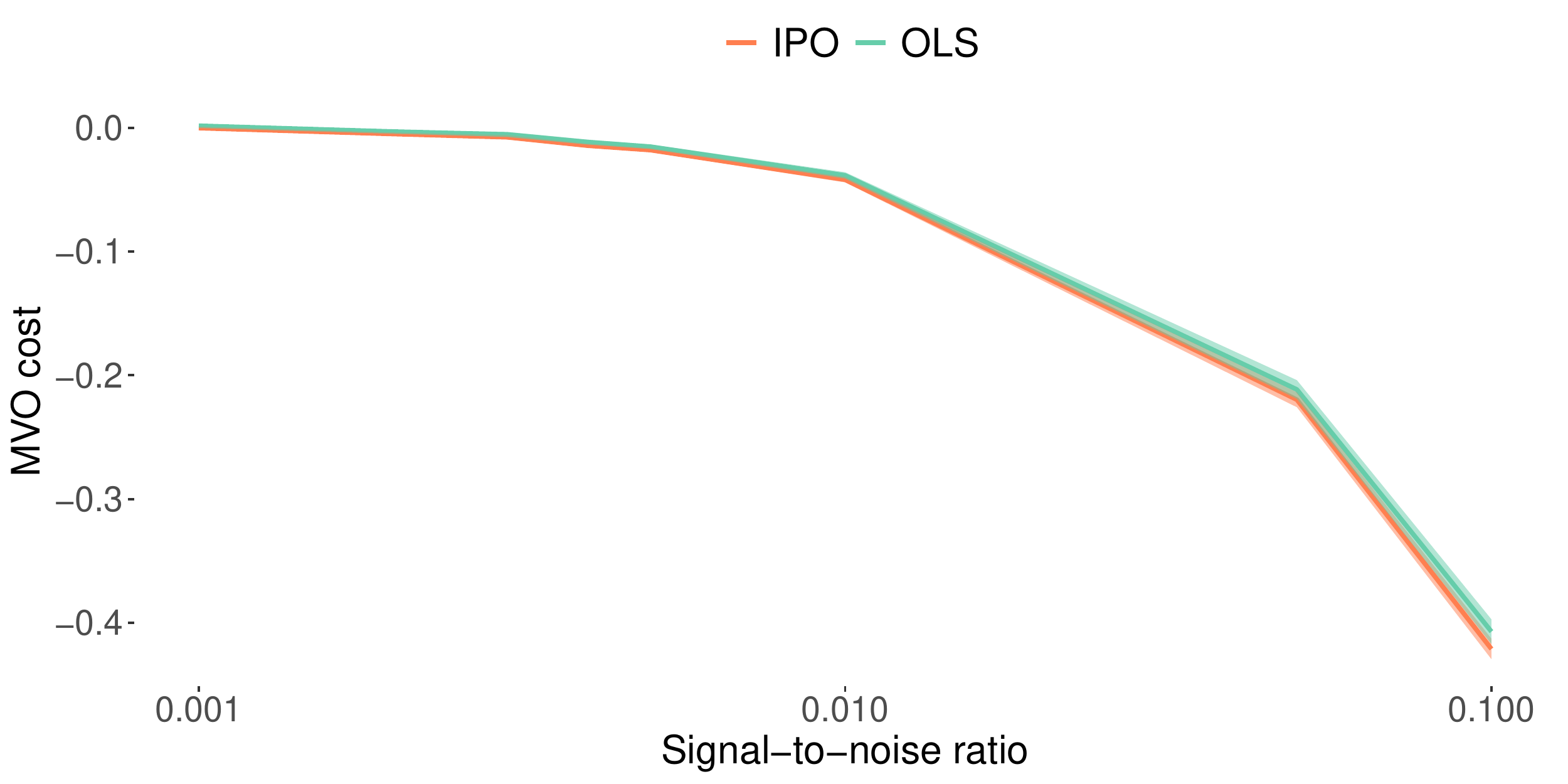}
    \caption{$\rho = 0.25, \text{res} = 20$.}
  \end{subfigure}
  \begin{subfigure}[b]{0.35\linewidth}
    \includegraphics[width=\linewidth , trim={00mm 0cm 0cm 0cm},clip]{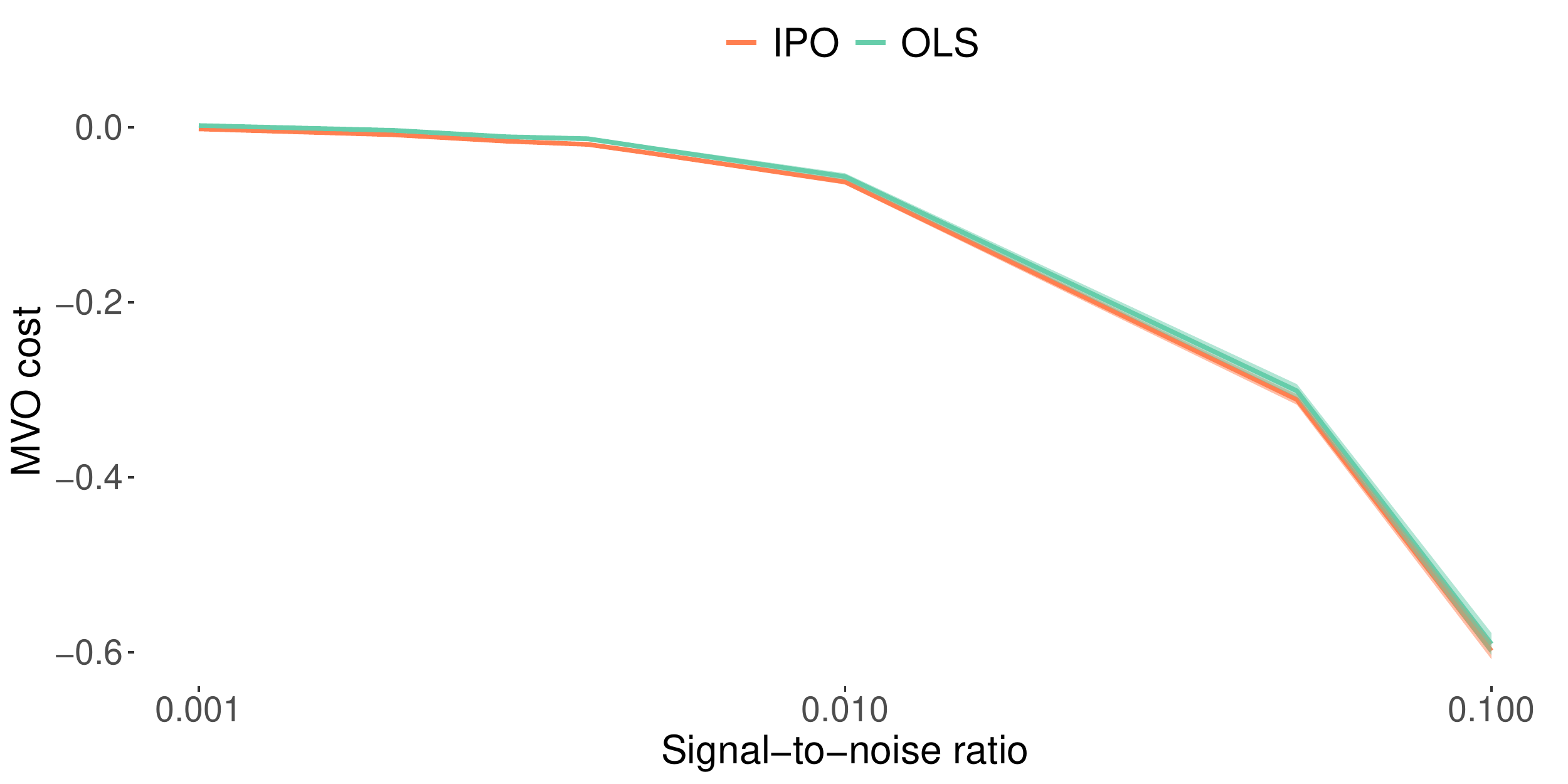}
    \caption{$\rho = 0.50, \text{res} = 20$.}
  \end{subfigure}
  \begin{subfigure}[b]{0.35\linewidth}
    \includegraphics[width=\linewidth , trim={00mm 0cm 0cm 0cm},clip]{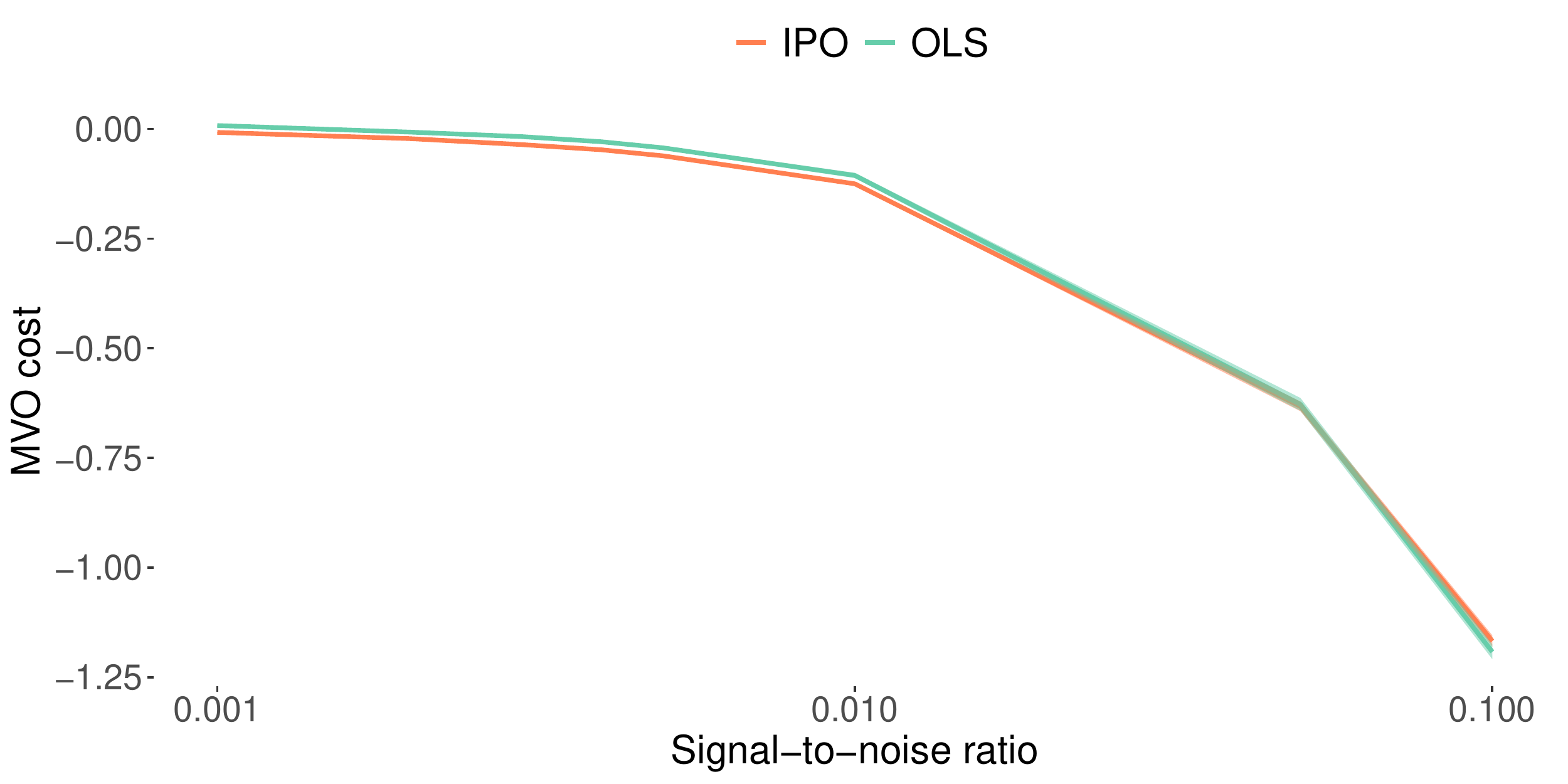}
    \caption{$\rho = 0.75, \text{res} = 20$.}
  \end{subfigure}
  \caption{Out-of-sample MVO cost for IPO and OLS as of function of return signal-to-noise ratios.}
  \label{fig:sims_1_mvo_200}
\end{figure}

\begin{figure}[H]
  \centering
  \begin{subfigure}[b]{0.35\linewidth}
    \includegraphics[width=\linewidth , trim={00mm 0cm 0cm 0cm},clip]{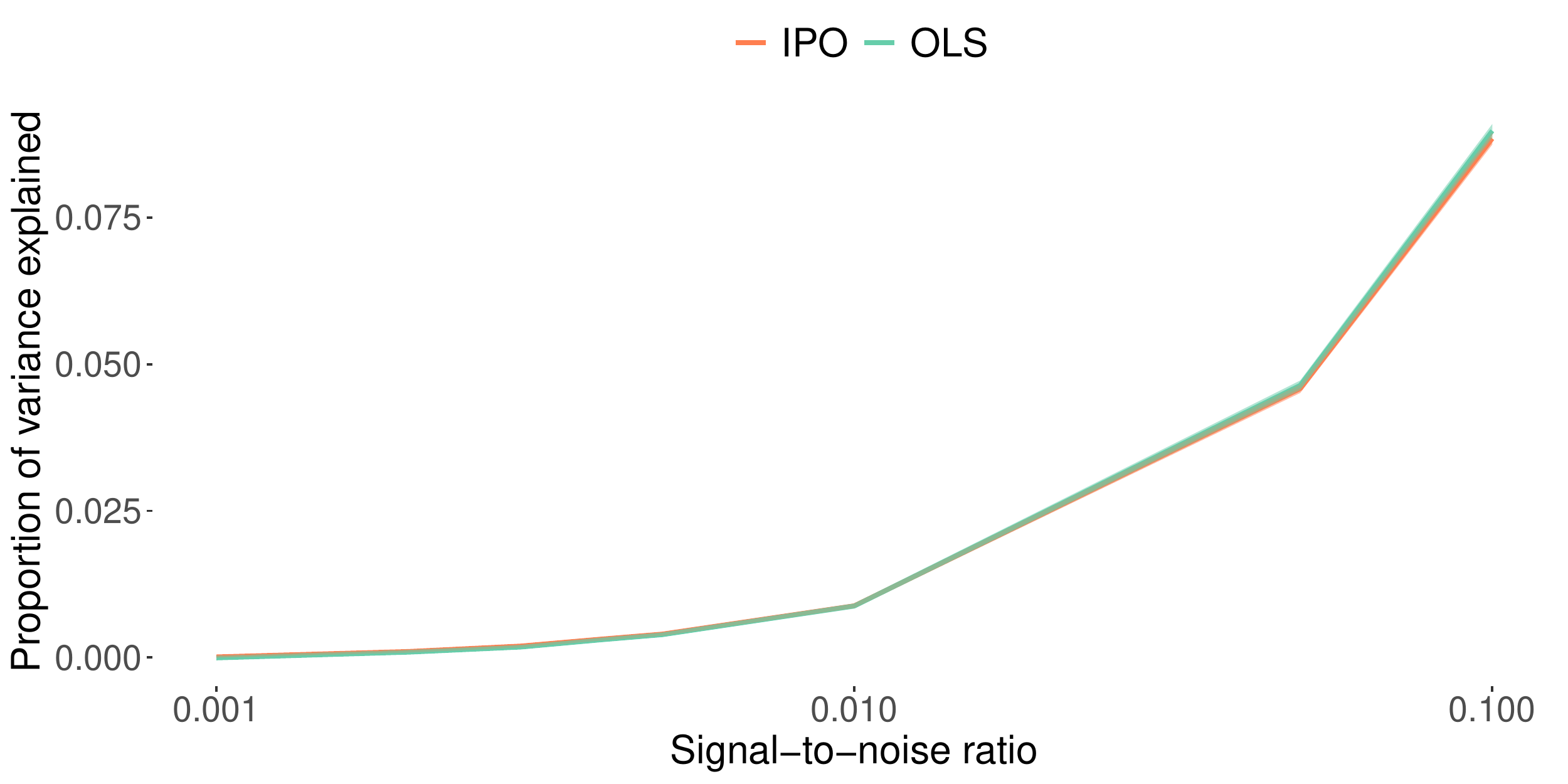}
    \caption{$\rho = 0, \text{res} = 20$.}
  \end{subfigure}
 \begin{subfigure}[b]{0.35\linewidth}
    \includegraphics[width=\linewidth , trim={00mm 0cm 0cm 0cm},clip]{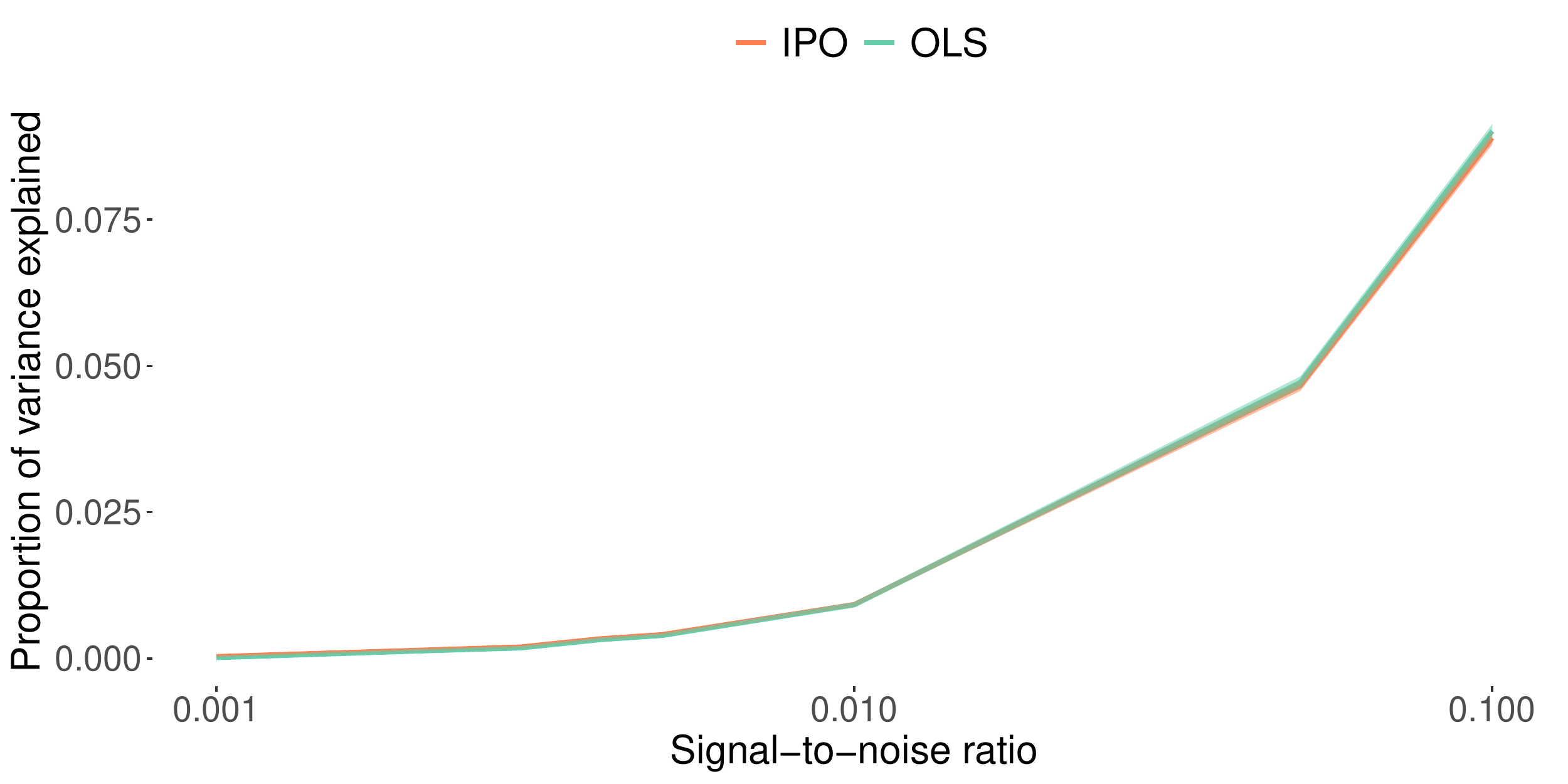}
    \caption{$\rho = 0.25, \text{res} = 20$.}
  \end{subfigure}
  \begin{subfigure}[b]{0.35\linewidth}
    \includegraphics[width=\linewidth , trim={00mm 0cm 0cm 0cm},clip]{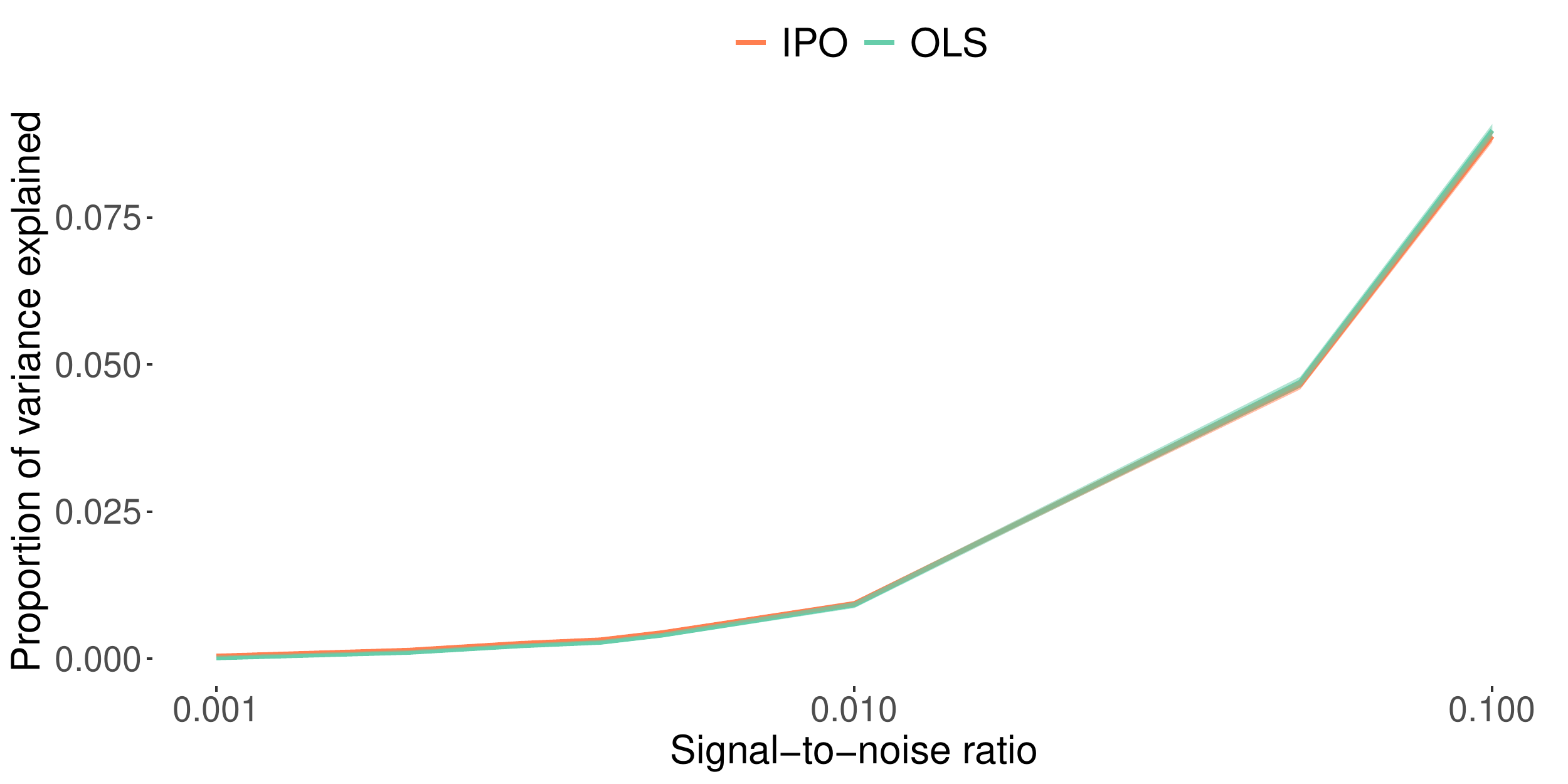}
    \caption{$\rho = 0.50, \text{res} = 20$.}
  \end{subfigure}
  \begin{subfigure}[b]{0.35\linewidth}
    \includegraphics[width=\linewidth , trim={00mm 0cm 0cm 0cm},clip]{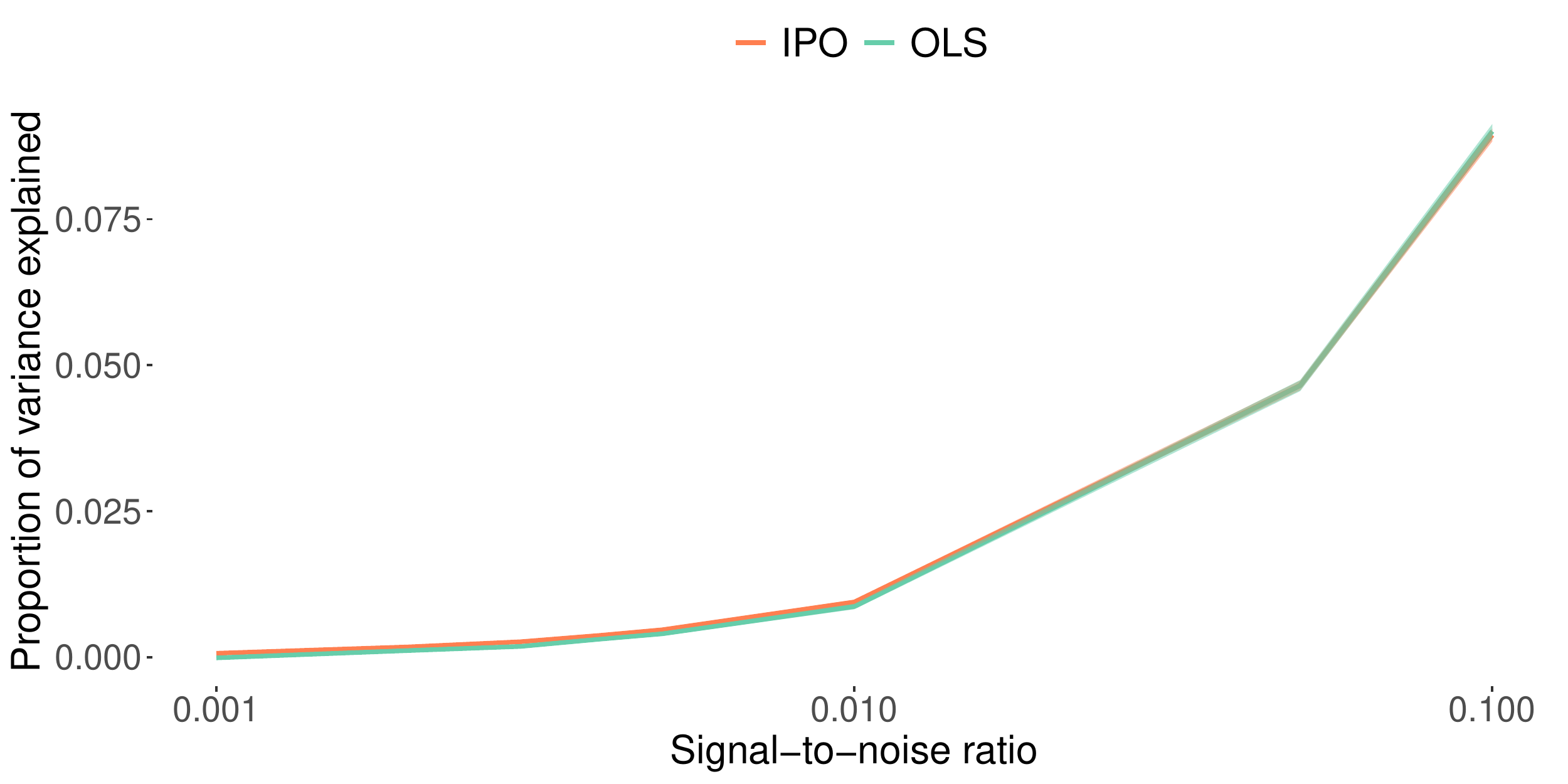}
    \caption{$\rho = 0.75, \text{res} = 20$.}
  \end{subfigure}
  \caption{Out-of-sample PVE for IPO and OLS as of function of return signal-to-noise ratios.}
  \label{fig:sims_1_pve_200}
\end{figure}

\begin{figure}[H]
  \centering
  \begin{subfigure}[b]{0.35\linewidth}
    \includegraphics[width=\linewidth , trim={00mm 0cm 0cm 0cm},clip]{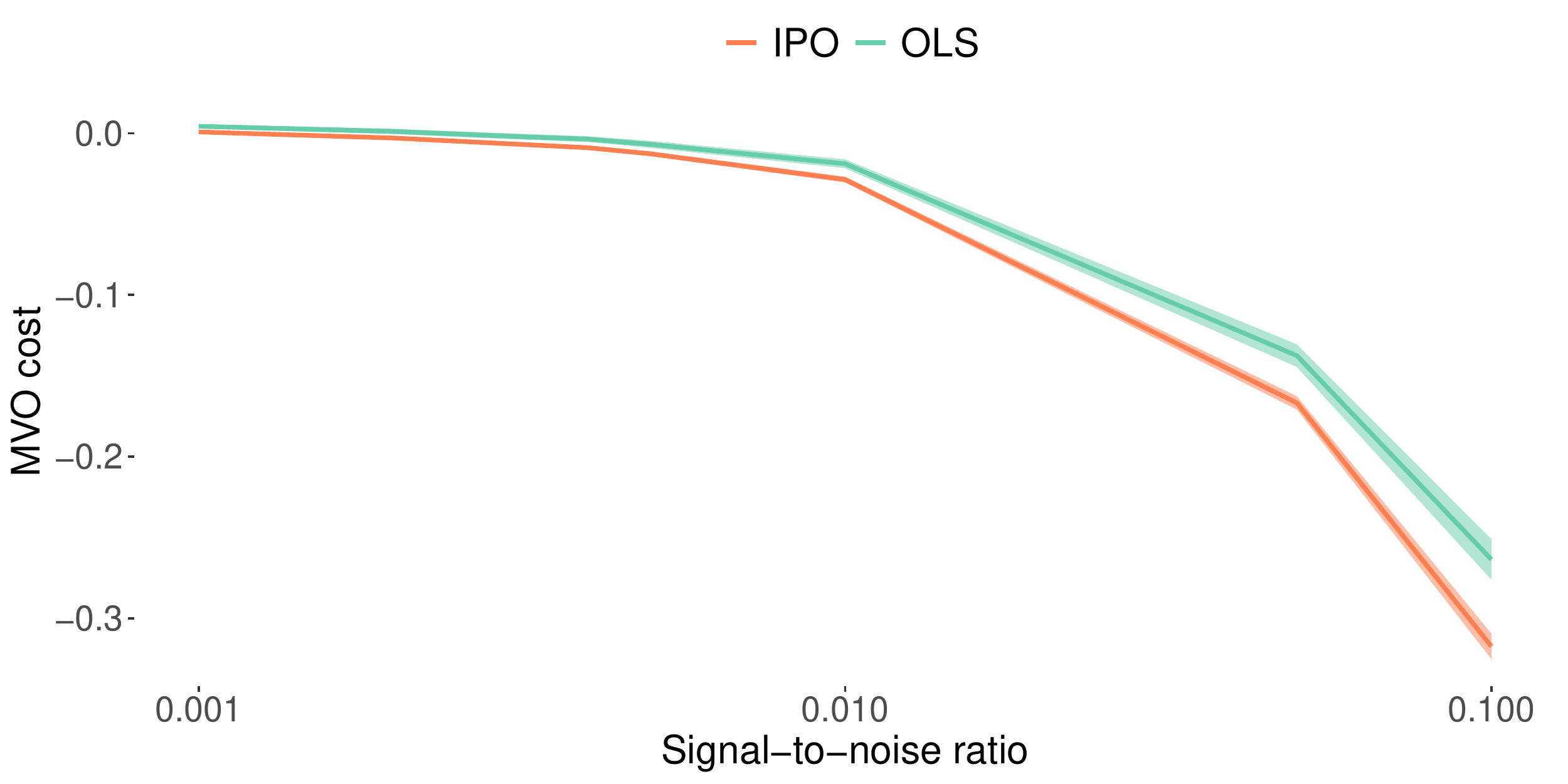}
    \caption{$\rho = 0, \text{res} = 10$.}
  \end{subfigure}
 \begin{subfigure}[b]{0.35\linewidth}
    \includegraphics[width=\linewidth , trim={00mm 0cm 0cm 0cm},clip]{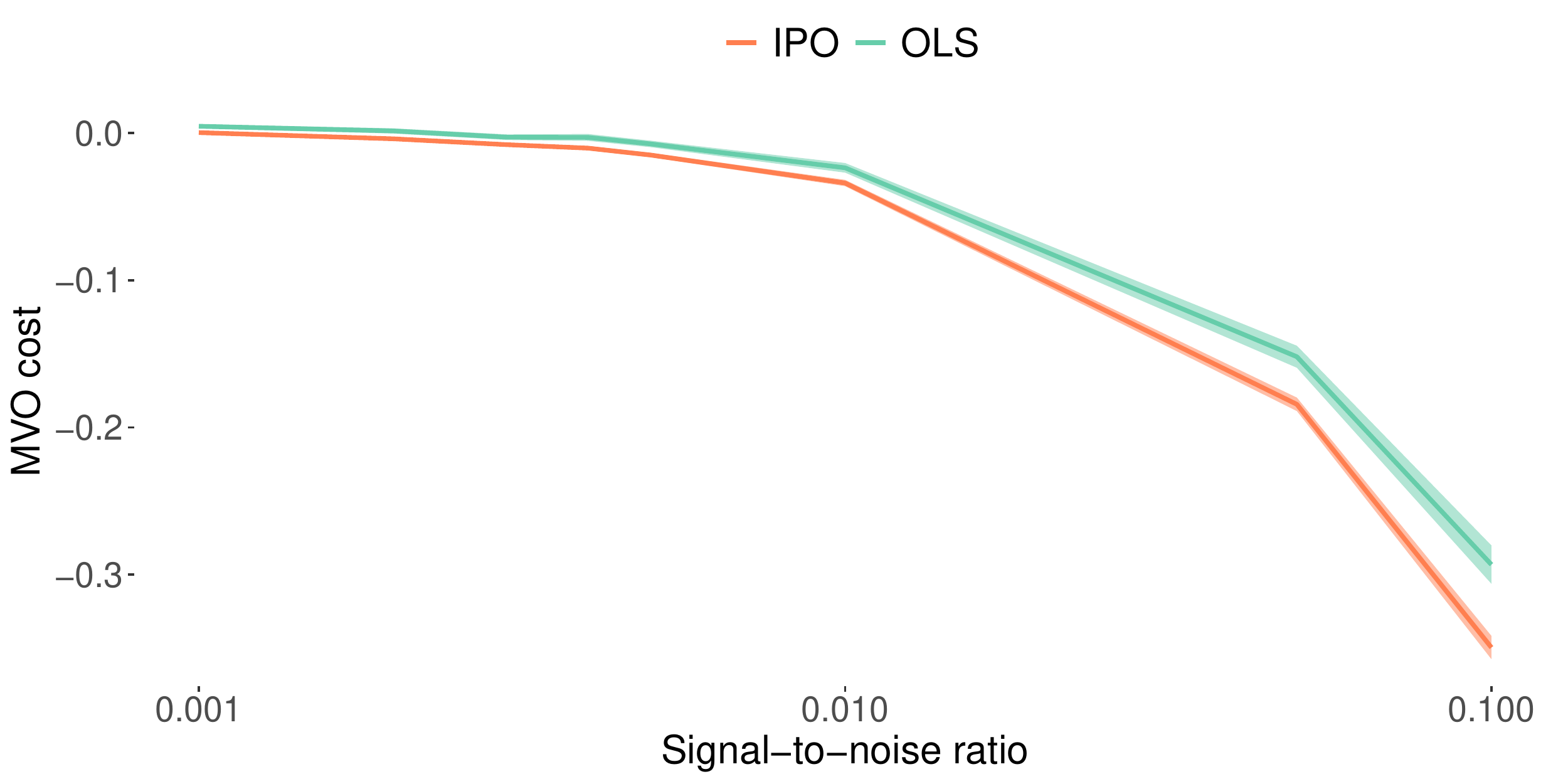}
    \caption{$\rho = 0.25, \text{res} = 10$.}
  \end{subfigure}
  \begin{subfigure}[b]{0.35\linewidth}
    \includegraphics[width=\linewidth , trim={00mm 0cm 0cm 0cm},clip]{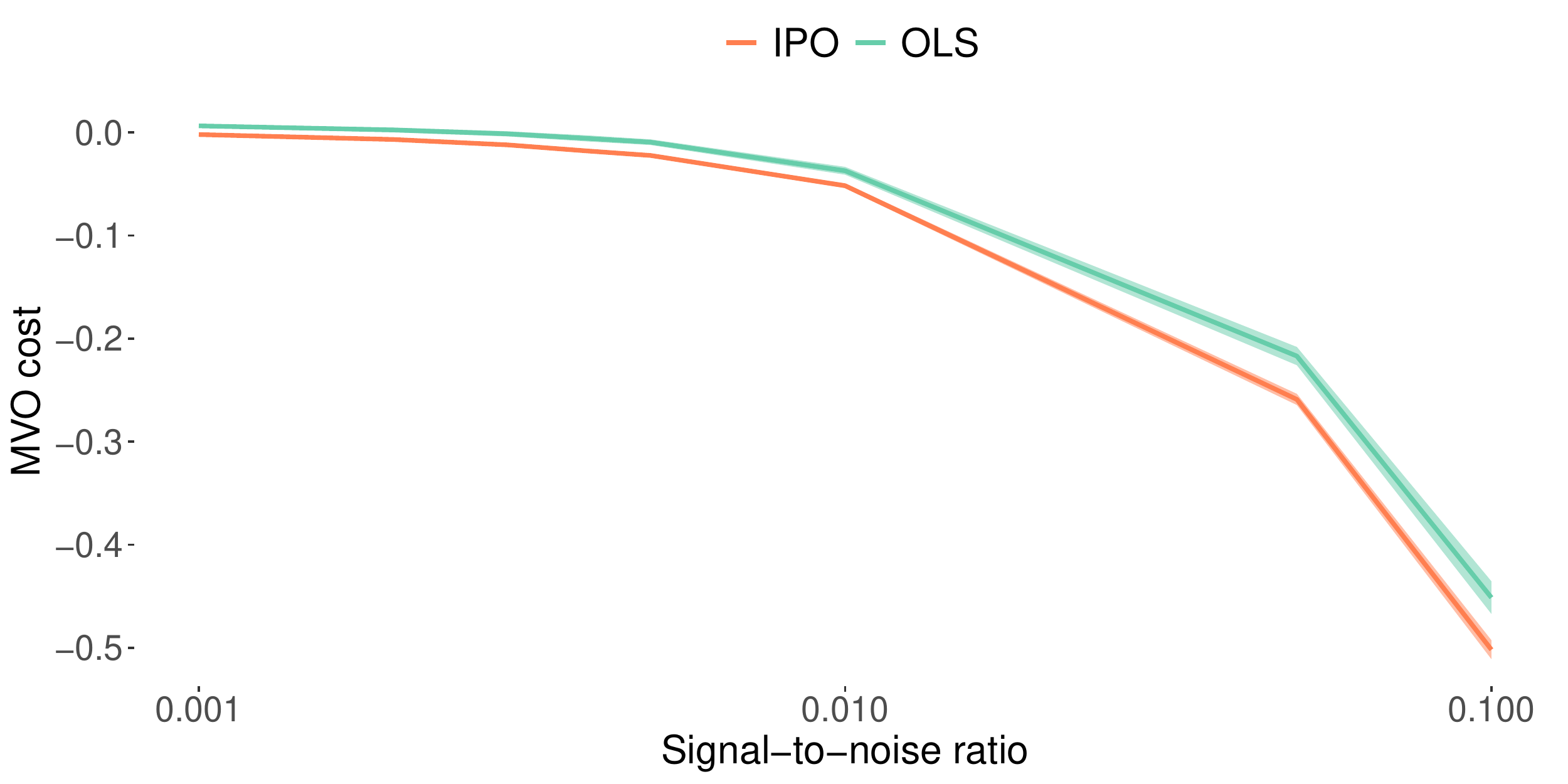}
    \caption{$\rho = 0.50, \text{res} = 10$.}
  \end{subfigure}
  \begin{subfigure}[b]{0.35\linewidth}
    \includegraphics[width=\linewidth , trim={00mm 0cm 0cm 0cm},clip]{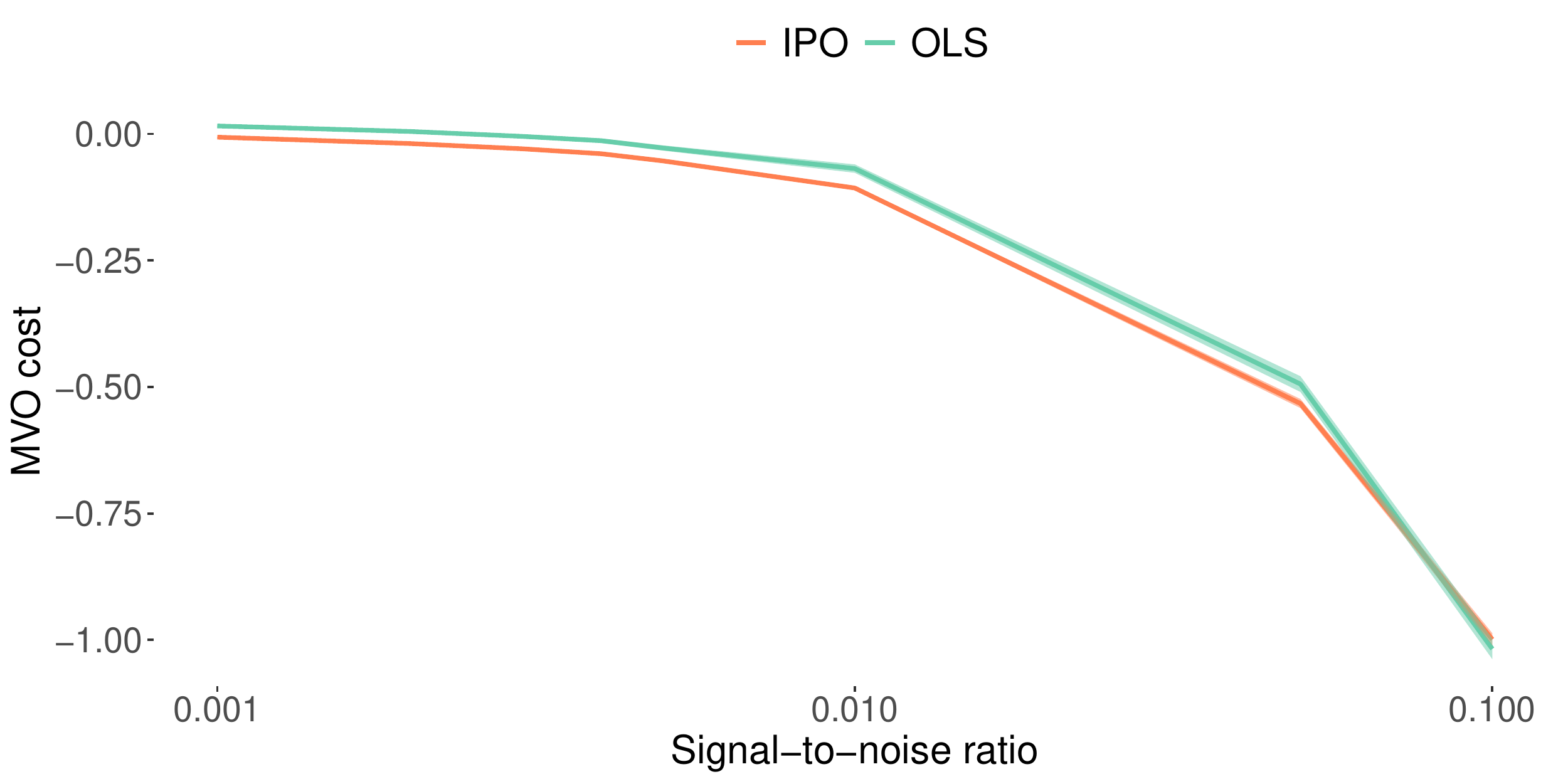}
    \caption{$\rho = 0.75, \text{res} = 10$.}
  \end{subfigure}
  \caption{Out-of-sample MVO cost for IPO and OLS as of function of return signal-to-noise ratios.}
  \label{fig:sims_1_mvo_100}
\end{figure}

\begin{figure}[H]
  \centering
  \begin{subfigure}[b]{0.35\linewidth}
    \includegraphics[width=\linewidth , trim={00mm 0cm 0cm 0cm},clip]{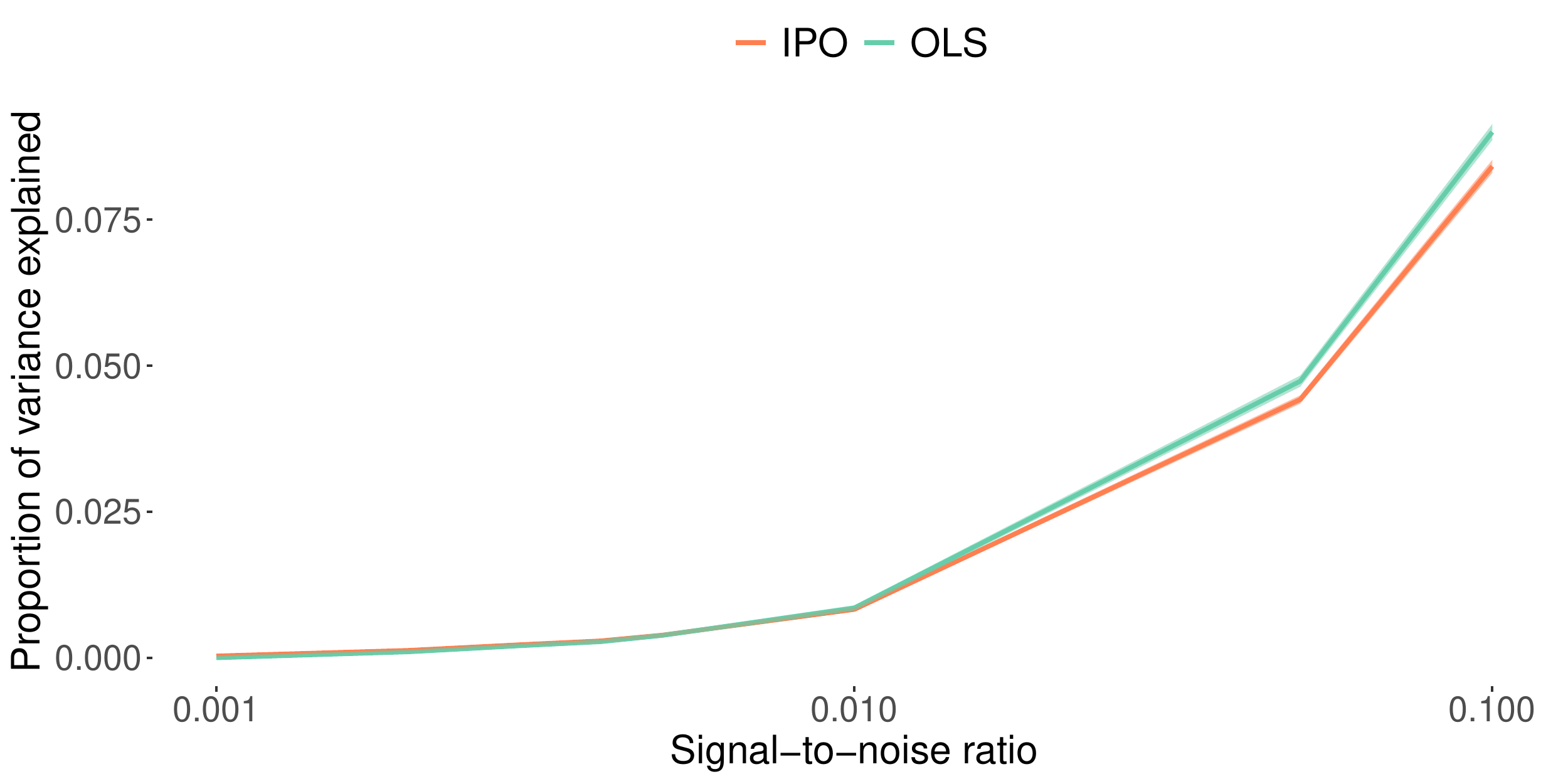}
    \caption{$\rho = 0, \text{res} = 10$.}
  \end{subfigure}
 \begin{subfigure}[b]{0.35\linewidth}
    \includegraphics[width=\linewidth , trim={00mm 0cm 0cm 0cm},clip]{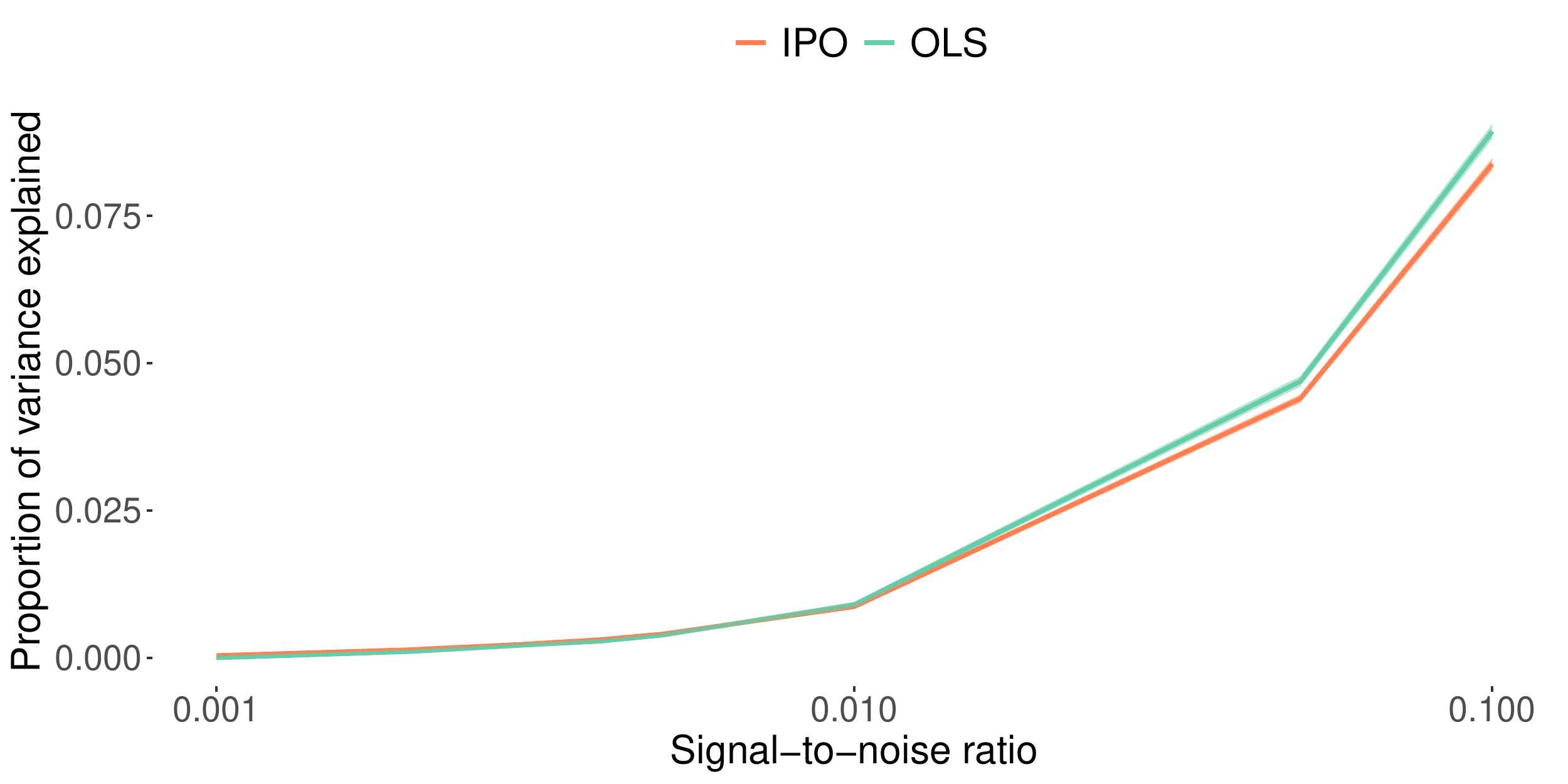}
    \caption{$\rho = 0.25, \text{res} = 10$.}
  \end{subfigure}
  \begin{subfigure}[b]{0.35\linewidth}
    \includegraphics[width=\linewidth , trim={00mm 0cm 0cm 0cm},clip]{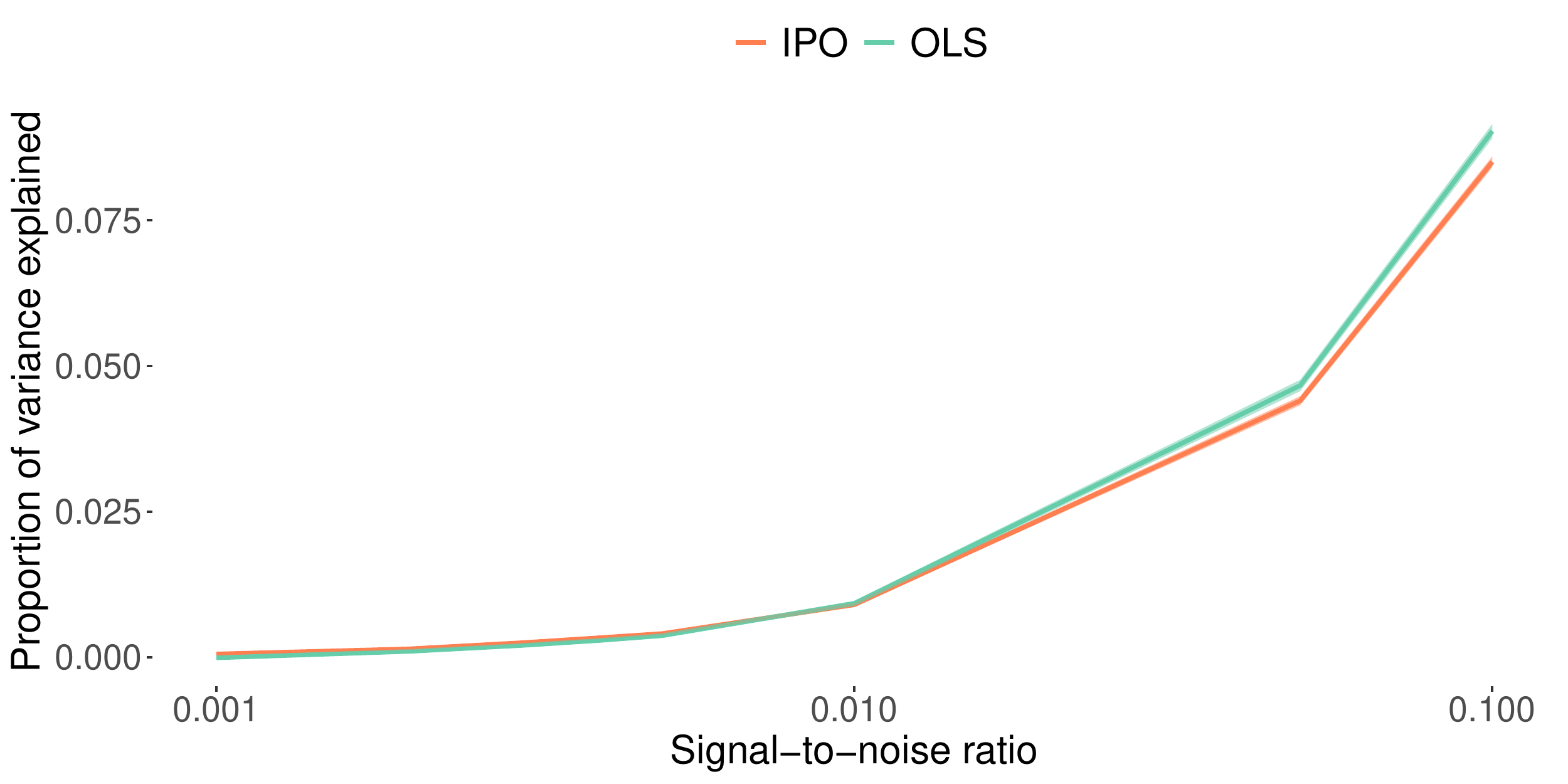}
    \caption{$\rho = 0.50, \text{res} = 10$.}
  \end{subfigure}
  \begin{subfigure}[b]{0.35\linewidth}
    \includegraphics[width=\linewidth , trim={00mm 0cm 0cm 0cm},clip]{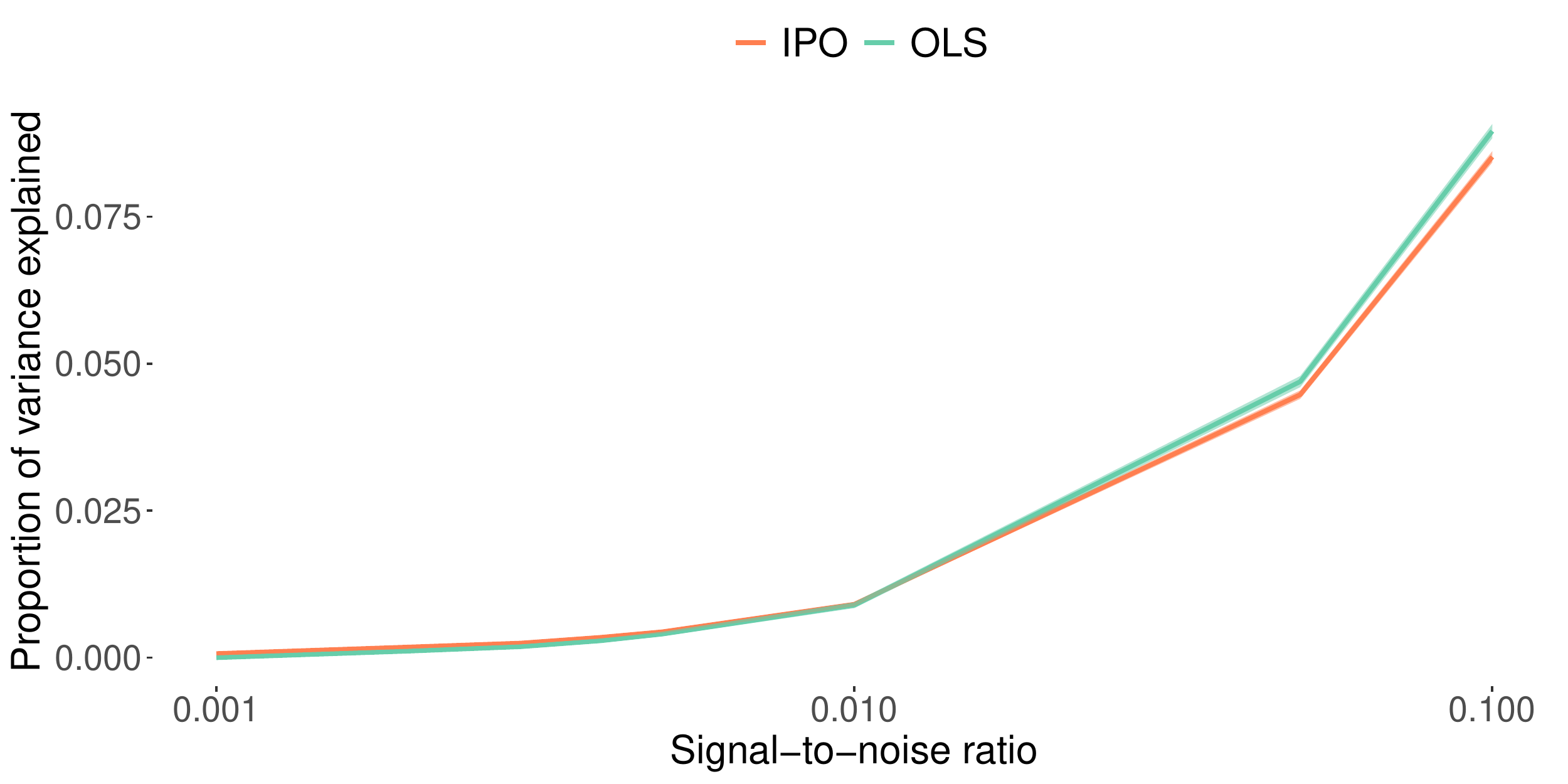}
    \caption{$\rho = 0.75, \text{res} = 10$.}
  \end{subfigure}
  \caption{Out-of-sample PVE for IPO and OLS as of function of return signal-to-noise ratios.}
  \label{fig:sims_1_pve_100}
\end{figure}

\begin{figure}[H]
  \centering
  \begin{subfigure}[b]{0.35\linewidth}
    \includegraphics[width=\linewidth , trim={00mm 0cm 0cm 0cm},clip]{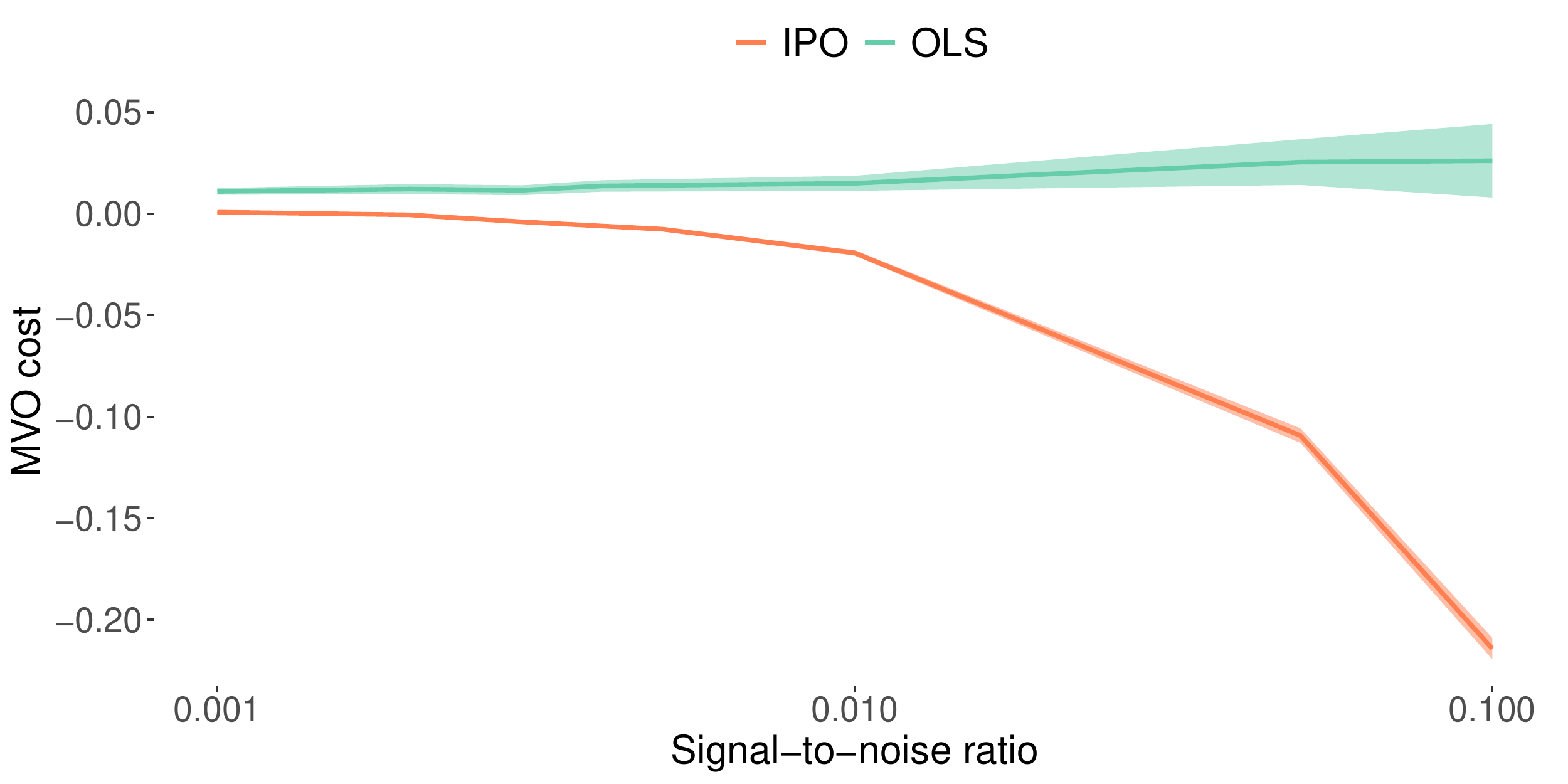}
    \caption{$\rho = 0, \text{res} = 5$.}
  \end{subfigure}
 \begin{subfigure}[b]{0.35\linewidth}
    \includegraphics[width=\linewidth , trim={00mm 0cm 0cm 0cm},clip]{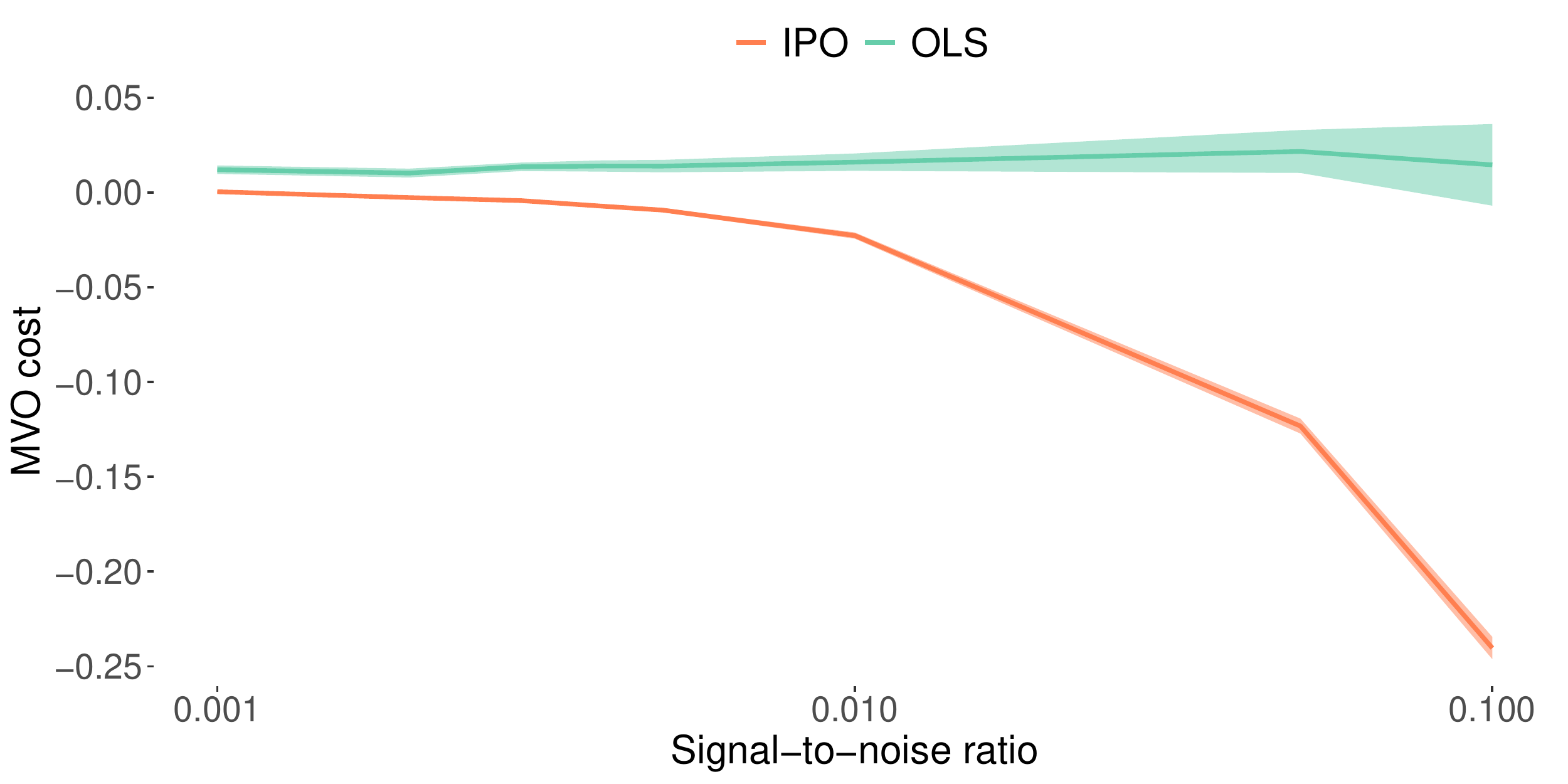}
    \caption{$\rho = 0.25, \text{res} = 5$.}
  \end{subfigure}
  \begin{subfigure}[b]{0.35\linewidth}
    \includegraphics[width=\linewidth , trim={00mm 0cm 0cm 0cm},clip]{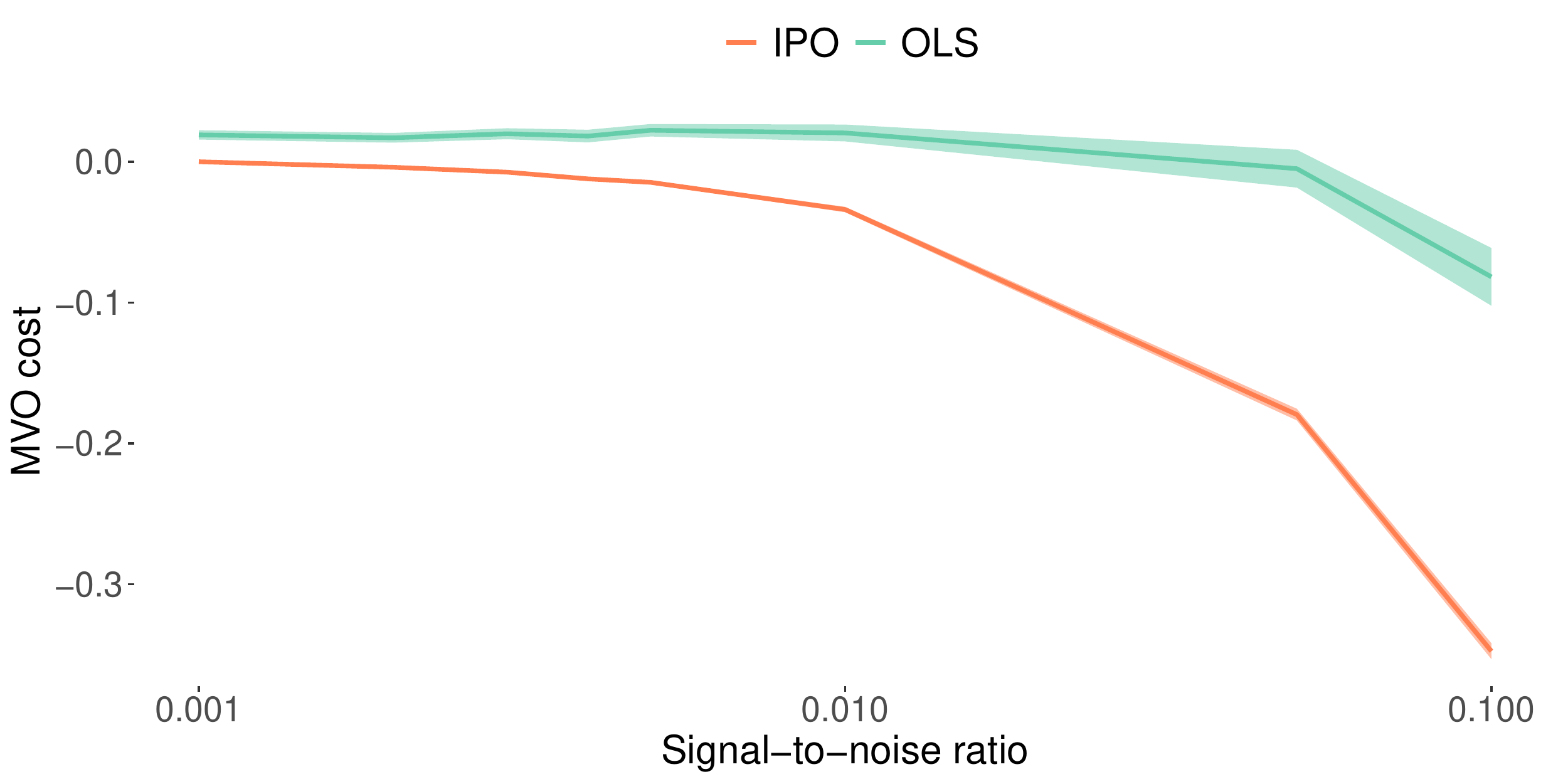}
    \caption{$\rho = 0.50, \text{res} = 5$.}
  \end{subfigure}
  \begin{subfigure}[b]{0.35\linewidth}
    \includegraphics[width=\linewidth , trim={00mm 0cm 0cm 0cm},clip]{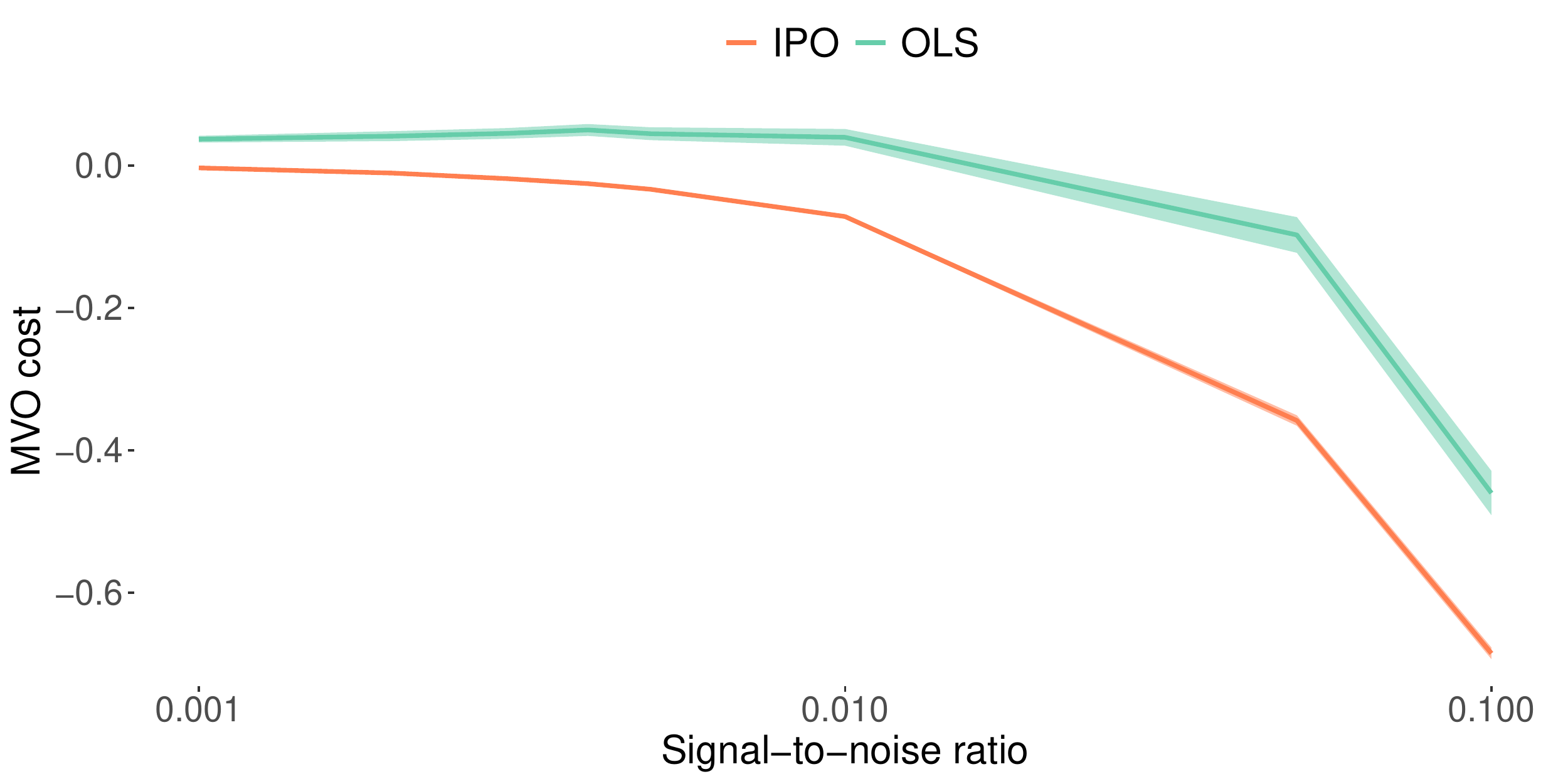}
    \caption{$\rho = 0.75, \text{res} = 5$.}
  \end{subfigure}
  \caption{Out-of-sample MVO cost for IPO and OLS as of function of return signal-to-noise ratios.}
  \label{fig:sims_1_mvo_50}
\end{figure}

\begin{figure}[H]
  \centering
  \begin{subfigure}[b]{0.35\linewidth}
    \includegraphics[width=\linewidth , trim={00mm 0cm 0cm 0cm},clip]{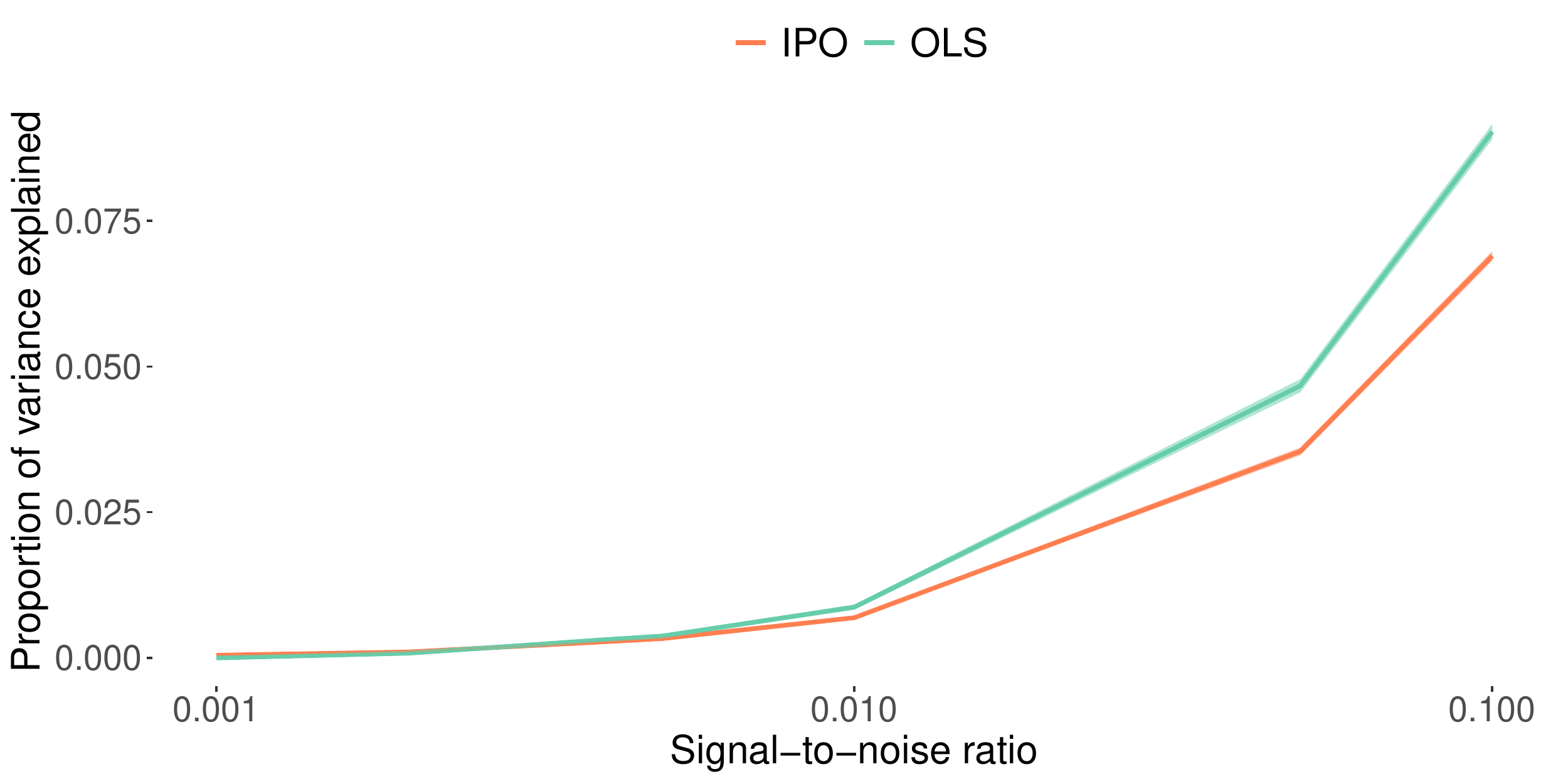}
    \caption{$\rho = 0, \text{res} = 5$.}
  \end{subfigure}
 \begin{subfigure}[b]{0.35\linewidth}
    \includegraphics[width=\linewidth , trim={00mm 0cm 0cm 0cm},clip]{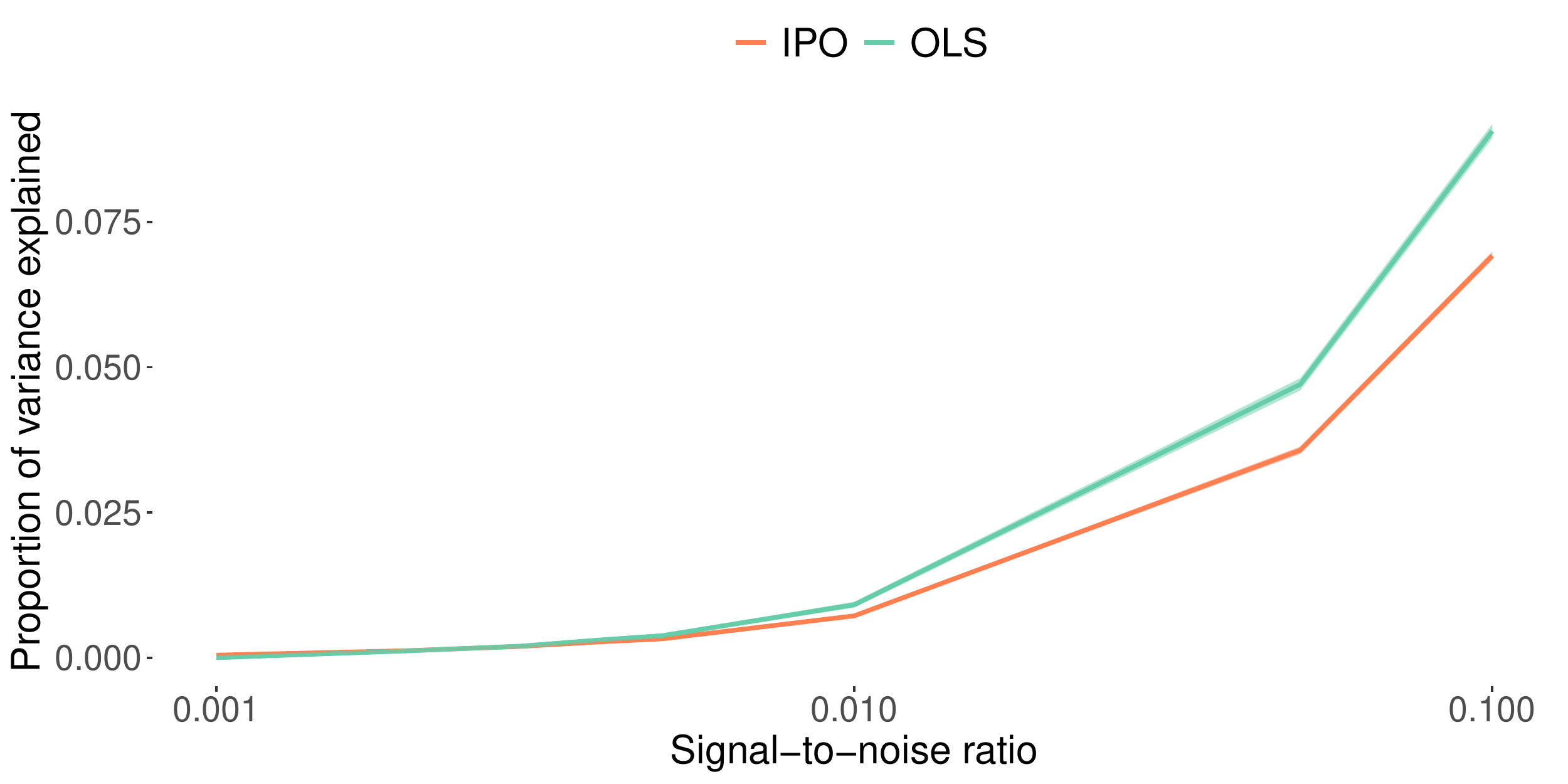}
    \caption{$\rho = 0.25, \text{res} = 5$.}
  \end{subfigure}
  \begin{subfigure}[b]{0.35\linewidth}
    \includegraphics[width=\linewidth , trim={00mm 0cm 0cm 0cm},clip]{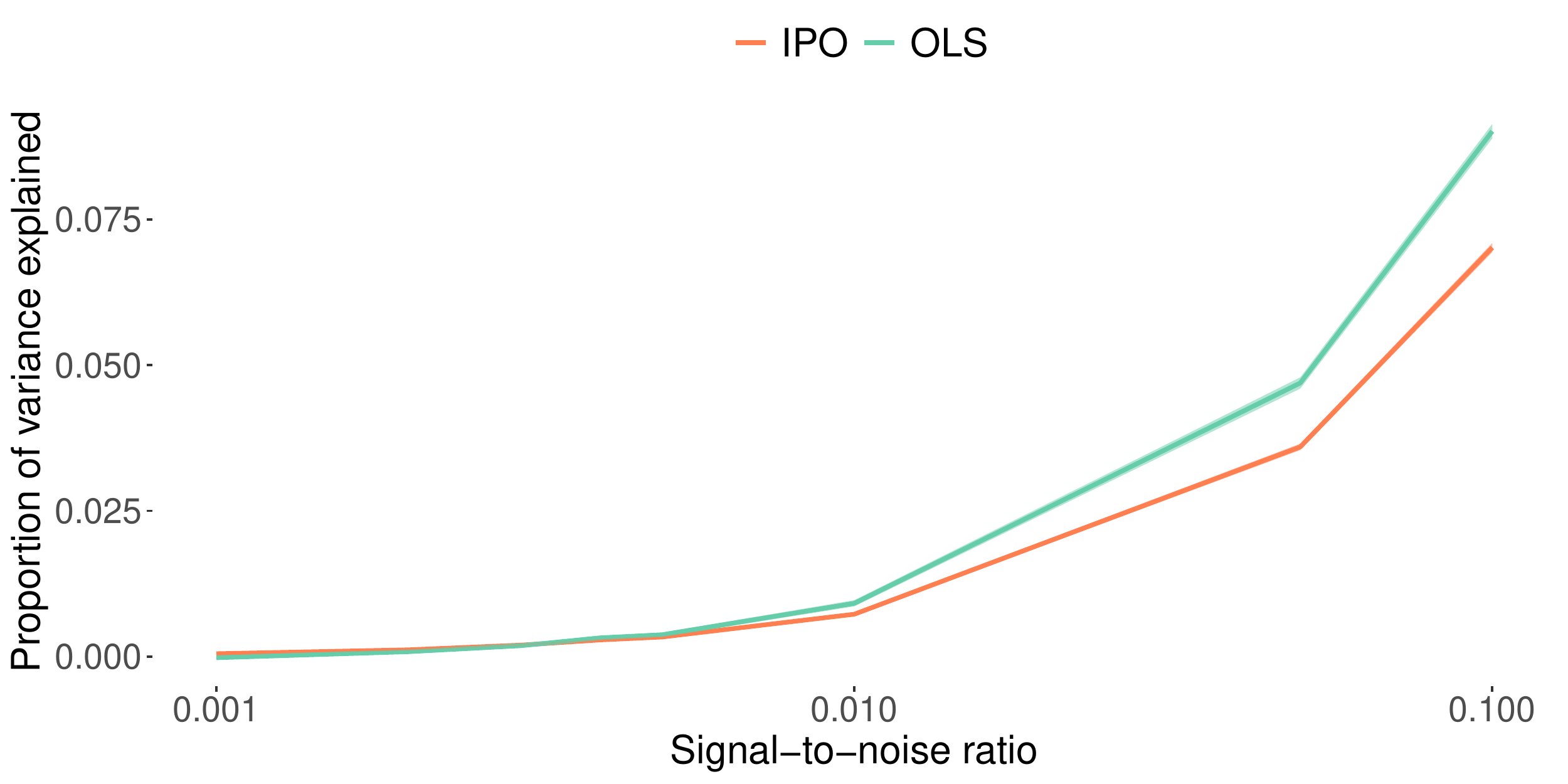}
    \caption{$\rho = 0.50, \text{res} = 5$.}
  \end{subfigure}
  \begin{subfigure}[b]{0.35\linewidth}
    \includegraphics[width=\linewidth , trim={00mm 0cm 0cm 0cm},clip]{z_eqcon_cov_10_rho_0_75_pv.pdf}
    \caption{$\rho = 0.75, \text{res} = 5$.}
  \end{subfigure}
  \caption{Out-of-sample PVE cost for IPO and OLS as of function of return signal-to-noise ratios.}
  \label{fig:sims_1_pve_50}
\end{figure}


\subsection{Simulation 2: computational efficiency}\label{sec:sims_2}

Here, we compare the computational efficiency of the analytical IPO solution with the current state-of-the-art method based on implicit differentiation and iterative gradient descent, from here on denoted as IPO-GRAD. Note that the IPO-GRAD implementation is optimized such that the matrix factorization (Equation $\eqref{eq:diff}$), required to compute gradient, is performed once at the initialization of the algorithm. The IPO-GRAD coefficients are initialized by drawing from the standard normal distribution and the algorithm terminates when $\lVert \partial L /\partial \btheta \rVert < 10^{-6}$.

We generate synthetic asset returns, following the procedure outlined in Section \ref{sec:sims_1}, with $\rho = 0$, ${\text{SNR} = 0.005}$ and varying the number of assets, $d_z \in \{25, 50, 100, 250\}$. Each asset is assumed to have $3$ unique features, and therefore $d_\theta = 3d_z$.  Tables \ref{table:sims_2_uncon} and \ref{table:sims_2_eqcon}, report the time, in seconds, taken by each method to compute the optimal regression coefficients for the unconstrained and equality constrained cases, respectively. For the IPO-GRAD method we also report the number of iterations of gradient descent. For each portfolio size, we report the average and $95\%$-ile range over $100$ instances of simulated data. Observe that for problems with $100$ or fewer assets, the computation time required to compute the optimal IPO coefficients analytically is comparable to the computation time required to compute the optimal OLS coefficients. In contrast, the IPO-GRAD method typically requires over $100$ iterations of gradient descent and is anywhere from $10$x - $1000$x slower than the corresponding IPO method. We note that for problems of larger scale, the analytical IPO solution remains tractable and is on average $6$x faster than the IPO-GRAD method.


\begin{table}[h]
\centering
\begin{tabular}{l l l l l}
\hline
No. Assets  & OLS & IPO & IPO-GRAD & Iterations\\
\hline
25 & 0.029 & 0.071 & 4.333 & 178\\
 & (0.028,0.032) & (0.07,0.08) & (3.966,5.076) & (164,210)\\
\hline
50 & 0.247 & 0.429 & 6.557 & 186\\
 & (0.209,0.253) & (0.342,0.447) & (6.032,7.278) & (173,207)\\
\hline
100 & 0.545 & 1.7 & 17.642 & 200\\
 & (0.491,0.638) & (1.495,1.837) & (16.03,21.301) & (183,247)\\
\hline
250 & 2.89 & 17.961 & 123.975 & 208.5\\
 & (2.75,3.335) & (17.546,18.092) & (114.008,165.094) & (193,279)\\
\hline
\end{tabular}
\caption{Time in seconds for computing the optimal OLS, IPO and IPO-GRAD coefficients for an unconstrained MVO problem. Results are averaged over $100$ instances of simulated data.}
\label{table:sims_2_uncon}
\end{table}

\begin{table}[h]
\centering
\begin{tabular}{l l l l l}
\hline
No. Assets  & OLS & IPO & IPO-GRAD & Iterations\\
\hline
25 & 0.029 & 0.088 & 4.664 & 176\\
 & (0.028,0.032) & (0.085,0.094) & (4.333,5.587) & (163,211)\\
\hline
50 & 0.247 & 0.473 & 7.389 & 188\\
 & (0.171,0.259) & (0.383,0.543) & (6.696,8.227) & (172,208)\\
\hline
100 & 0.549 & 2.025 & 19.711 & 200\\
 & (0.492,0.669) & (1.855,2.161) & (18.034,23.449) & (183,241)\\
\hline
250 & 2.815 & 22.378 & 129.348 & 208\\
 & (2.71,3.315) & (21.8,22.511) & (119.684,174.607) & (193,280)\\
\hline
\end{tabular}
\caption{Time in seconds for computing the optimal OLS, IPO and IPO-GRAD coefficients for an equality constrained MVO problem. Results are averaged over $100$ instances of simulated data.}
\label{table:sims_2_eqcon}
\end{table}

\subsection{Simulation 3: inequality constrained IPO}\label{sec:sims_3}
We now consider the more general case whereby the feasible region of the MVO program is defined by inequality constraints. In general, an analytical solution to the IPO Program $\eqref{eq:mvo_full}$ in the presence of lower-level inequality constraints is not possible. Furthermore, Program $\eqref{eq:mvo_full}$ is not convex in $\btheta$. As a result, the current state-of-the-art approach (IPO-GRAD), described in Section \ref{sec:ineq}, is recommended in order to obtain locally optimal solutions.

The IPO-GRAD solution, however, is challenging for several reasons. First, in contrast to the traditional OLS approach, estimating the IPO coefficients by iterative methods can be computationally expensive. Specifically, at each iteration of gradient descent the IPO-GRAD method must solve at most $m$ constrained quadratic programs, where $m$ is the total number of training observations. Computation time therefore scales linearly with the number of training examples, $m$, and the number of iterations of gradient descent, $n$.  Moreover, when solved by interior-point methods, convex quadratic programs have worst-case time complexity on the order of $\mathcal{O}(d_z^3)$ \citep{Goldfarb1991}, and  therefore the worst-case complexity for IPO-GRAD is on the order of $\mathcal{O}(mnd_z^3)$. Fortunately, in most practical settings,  quadratic programs are solved in substantially fewer iterations than their worst-case bound \citep{Boyd2004}. Nonetheless, most real-world portfolio optimization problems involve portfolio sizes on the order of $10$ or $100$  and are trained using thousands of training observations. Estimating prediction model parameters by IPO-GRAD can therefore be computationally expensive. Secondly, because the inequality constrained IPO problem is not convex, we have no guarantee that any particular local solution is globally optimal. Moreover, if is difficult to estimate $\Var(\btheta)$ and compute confidence intervals by standard parametric methods, and instead expensive nonparametric bootstrap methods are required.

As a heuristic, we are interested in determining the out-of-sample efficacy of the analytical IPO solutions, presented in Sections \ref{sec:method_uncon} - \ref{sec:method_eqcon}, applied to the inequality constrained problem. Specifically, we compute the IPO optimal coefficients analytically by dropping the inequality constraints in the lower-level MVO problem. The realized policy, $\bz^*(\toi{\hat{\by}})$, however, enforces the inequality constraints in the out-of-sample evaluation period.

We generate synthetic asset returns, following the procedure outlined in Section \ref{sec:sims_1}, with $\rho = 0$, ${\text{SNR} = 0.005}$ and $d_z  = 10$. Each asset is assumed to have $3$ unique features ( $d_\theta = 3d_z$). The inequality constraints are standard box-constraints of the form:

$$
-\gamma \leq \bz_j \leq \gamma, \quad \forall j \in \{1, ..., d_z\},
$$
and we consider several values of $\gamma \in \{0.05, 0.10, 0.25, 0.50, 0.75, 1, 2, 5, 10\}$. We also vary the risk aversion parameter $\delta \in \{1, 5, 10, 25 \}$. Finally, asset mean returns are generated according to linear and nonlinear polynomial models of the form:
$$
\toi{\by} = \sum_{q = 1}^p \bP \diag(\toi{\bx}^q) \btheta_q + \tau \toi{\bepsilon},
$$
with $p \in \{1, 2, 4\}$.

Figures \ref{fig:sims_3_poly_1} - \ref{fig:sims_3_poly_4} compare the out-of-sample MVO cost for the IPO and IPO-GRAD methods as of function of the box constraint value, $\gamma$, and risk-aversion parameter, $\delta$. For each value of $\gamma, \delta$ and $p$, we report the mean and $95\%$-ile range over $30$ instances of simulated data.  First, we would expect the out-of-sample performance of the IPO and IPO-GRAD methods to converge as  $\gamma$ increases. Furthermore as $\delta$, increases, the point (along $\gamma$) at which the two solutions converge will naturally decrease. This effect is purely a consequence of the inequality constraints being non-active when either $\gamma$ and/or $\delta$ are sufficiently large.

In Figure \ref{fig:sims_3_poly_1} asset returns are generated according to a linear ground truth model $(p = 1)$. In all cases we observe that the IPO-GRAD does provide improved out-of-sample MVO costs when $\gamma$ is sufficiently small $(\gamma <0.5)$. However, for moderate and large values of $\gamma$, the IPO-GRAD method provides no improvement in out-of-sample MVO costs in comparison to the IPO method. Furthermore, in Figures \ref{fig:sims_3_poly_2} and \ref{fig:sims_3_poly_4}, asset returns are generated according to a quadratic $(p=2)$ and quartic $(p=4)$ ground truth model, respectively. We observe that over practically every value of $\gamma$ and $\delta$, the IPO method provides an equivalent, if not improved, out-of-sample MVO costs in comparison to the IPO-GRAD method. We note that, while not explicitly shown here,  the IPO-GRAD method produces lower in-sample (training) MVO costs over every experiment instance, and is potentially overfitting the training data. Moreover, we note that in all experiments, the variance of the out-of-sample MVO costs generated by the IPO-GRAD method is substantially larger than that of the IPO method. The lack of convexity and uniqueness of solution in the IPO-GRAD formulation, along with the likelihood of model overfit, provides a potential explanation for this effect.

Finally, Table \ref{table:sims_3} reports the average time (in seconds)  and $95\%$-ile range, taken by each method to compute the optimal regression coefficients. The results are averaged over all $360$ instances of simulated data. For the IPO-GRAD method we also report the number of iterations of gradient descent. We observe that the IPO-GRAD method typically requires around $60$ iterations of gradient descent and is on average $100$x - $1000$x slower than the corresponding IPO method. Note that the computation times reported here are for a relatively small portfolio and, given the computational complexity described above, we would expect the IPO method to provide an even larger computational advantage on medium and large sized portfolios. We therefore conclude that in the presence of inequality constraints, the IPO heuristic is a compelling alternative to the more computationally expensive IPO-GRAD solution.


\begin{figure}[H]
  \centering
  \begin{subfigure}[b]{0.35\linewidth}
    \includegraphics[width=\linewidth , trim={00mm 0cm 0cm 0cm},clip]{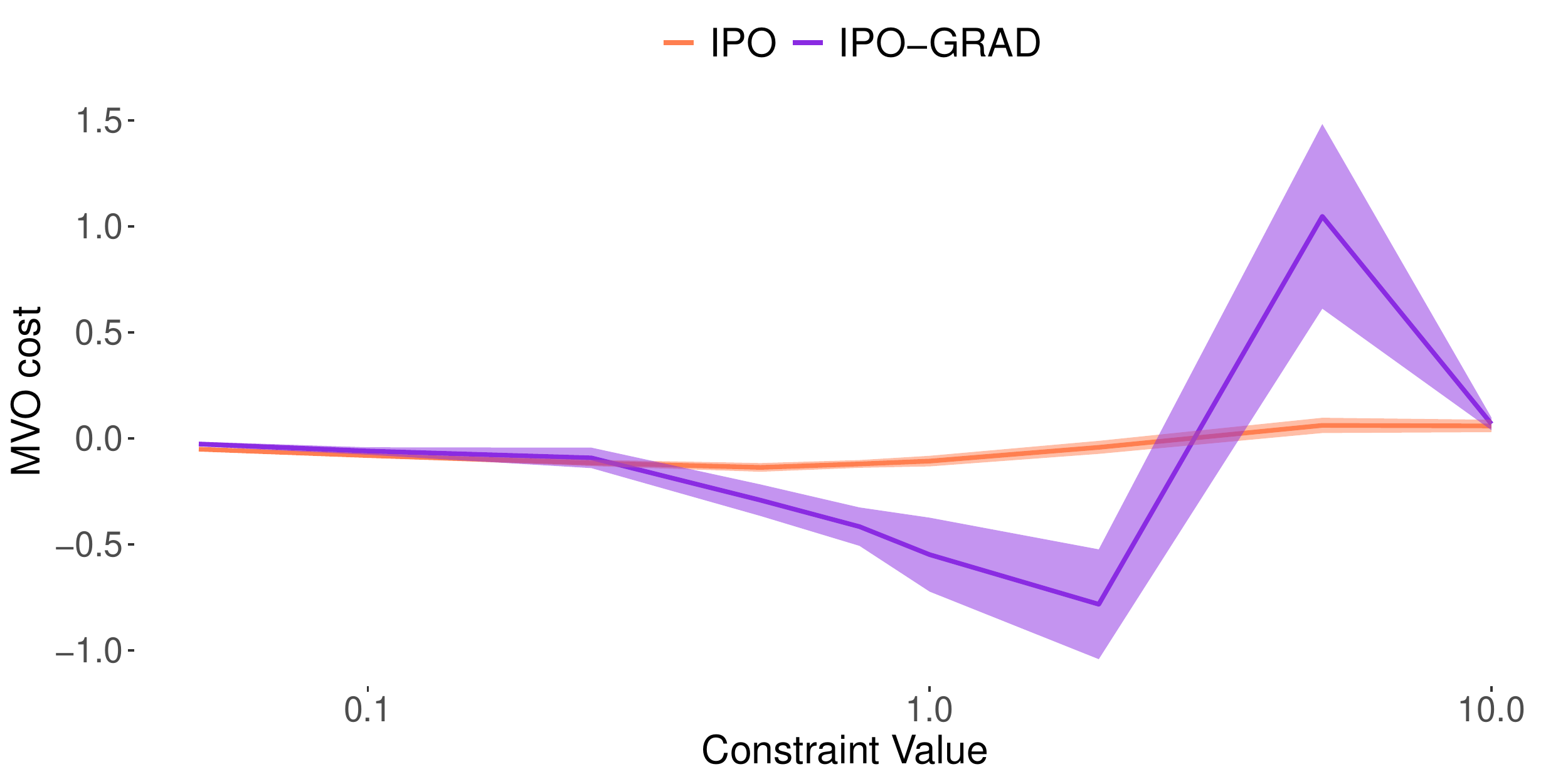}
    \caption{$\delta = 1$, $p = 1$.}
  \end{subfigure}
 \begin{subfigure}[b]{0.35\linewidth}
    \includegraphics[width=\linewidth , trim={00mm 0cm 0cm 0cm},clip]{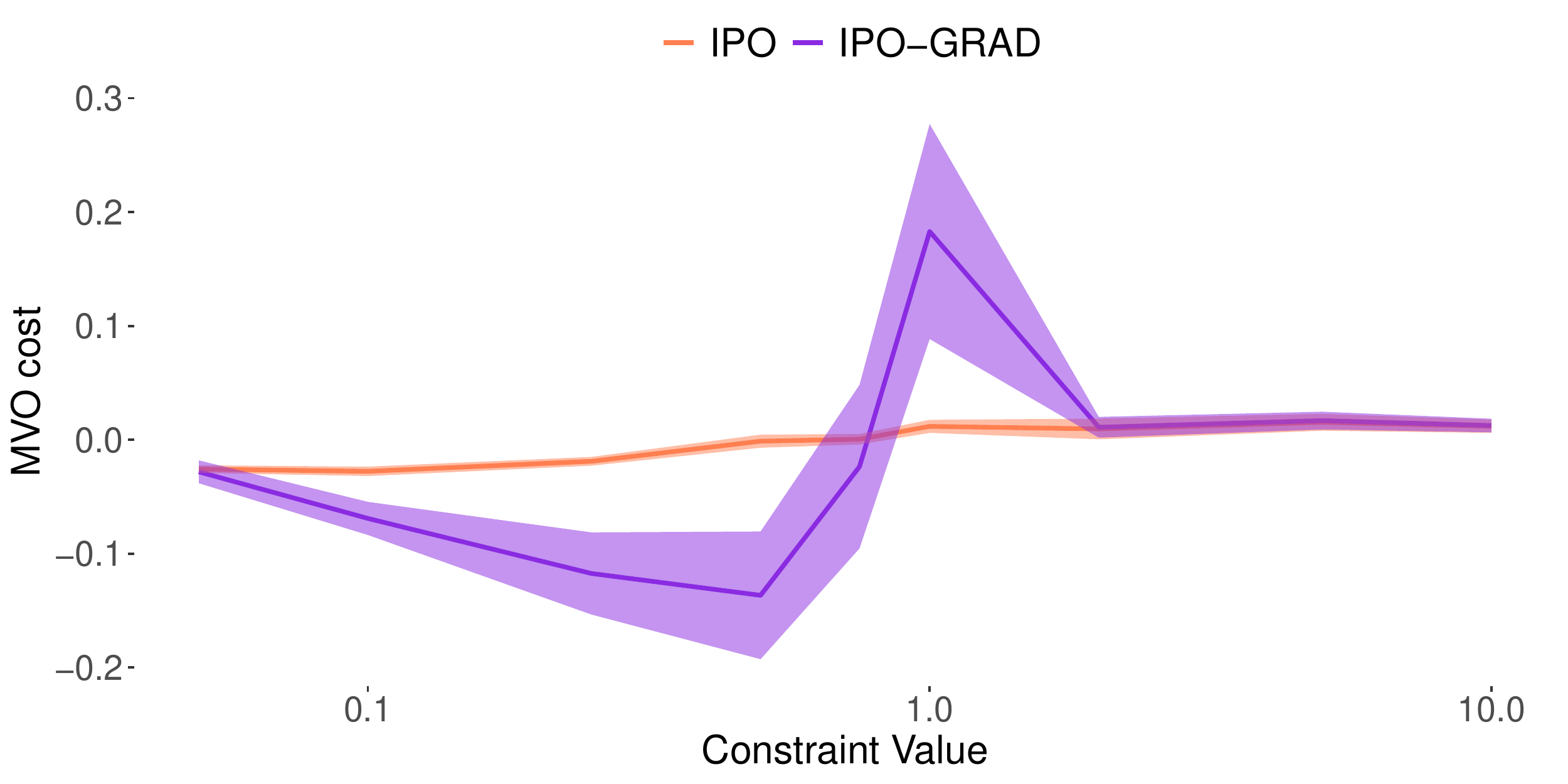}
    \caption{$\delta = 5$, $p = 1$.}
  \end{subfigure}
  \begin{subfigure}[b]{0.35\linewidth}
    \includegraphics[width=\linewidth , trim={00mm 0cm 0cm 0cm},clip]{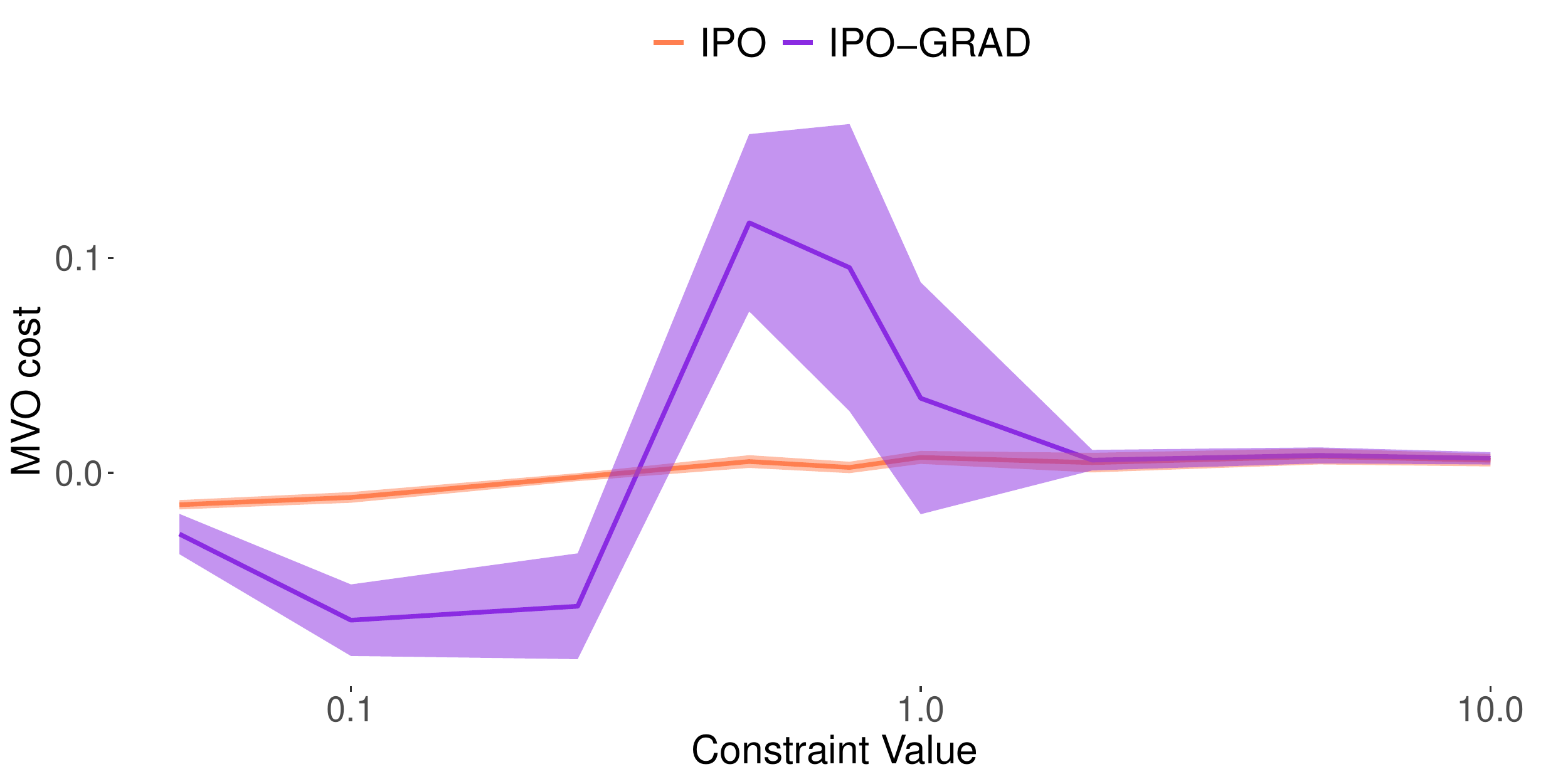}
    \caption{$\delta = 10$, $p = 1$.}
  \end{subfigure}
  \begin{subfigure}[b]{0.35\linewidth}
    \includegraphics[width=\linewidth , trim={00mm 0cm 0cm 0cm},clip]{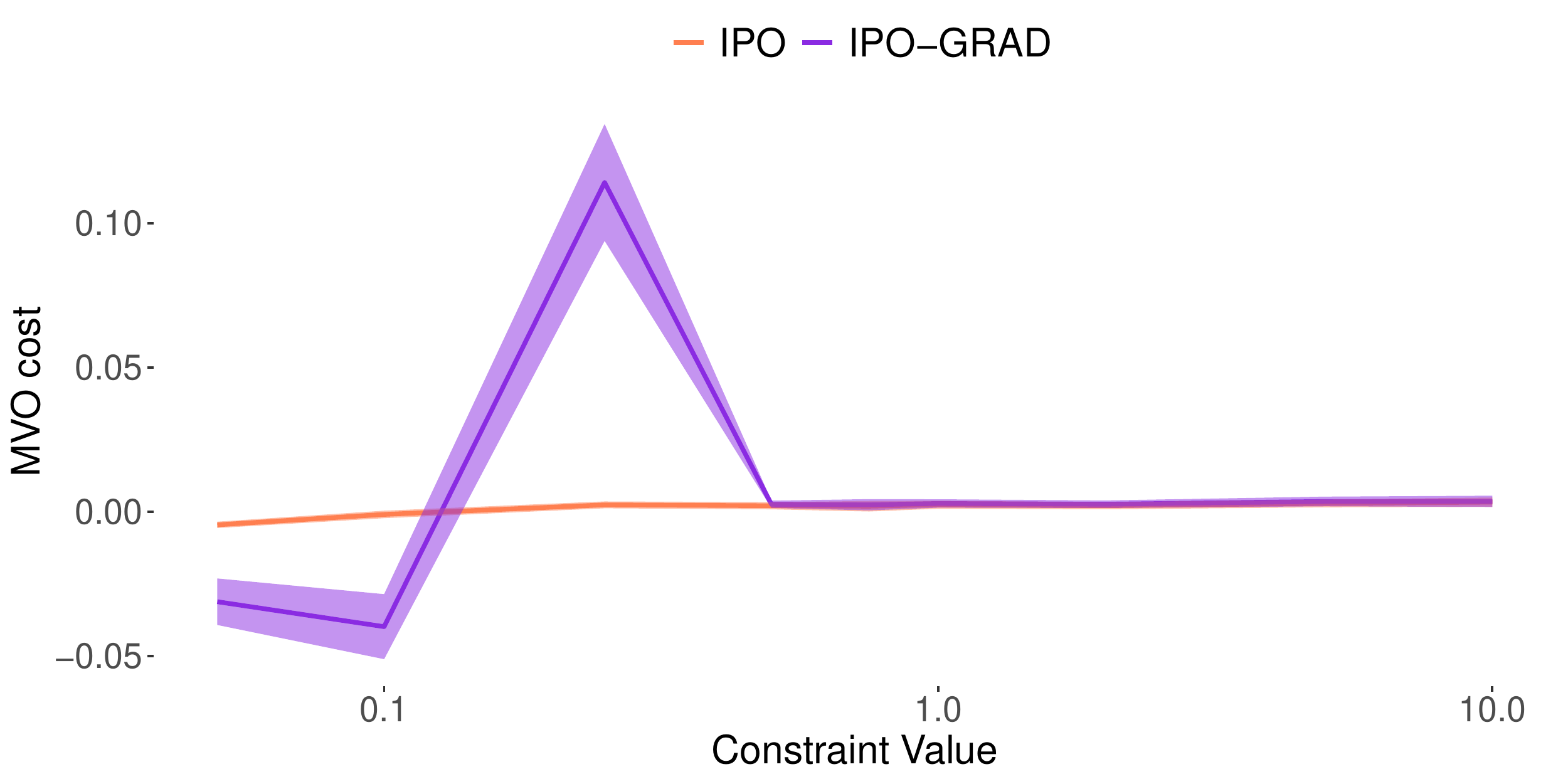}
    \caption{$\delta = 25$, $p = 1$.}
  \end{subfigure}
  \caption{Out-of-sample MVO costs as of function of box constraint value with $p = 1$.}
  \label{fig:sims_3_poly_1}
\end{figure}

\begin{figure}[H]
  \centering
  \begin{subfigure}[b]{0.35\linewidth}
    \includegraphics[width=\linewidth , trim={00mm 0cm 0cm 0cm},clip]{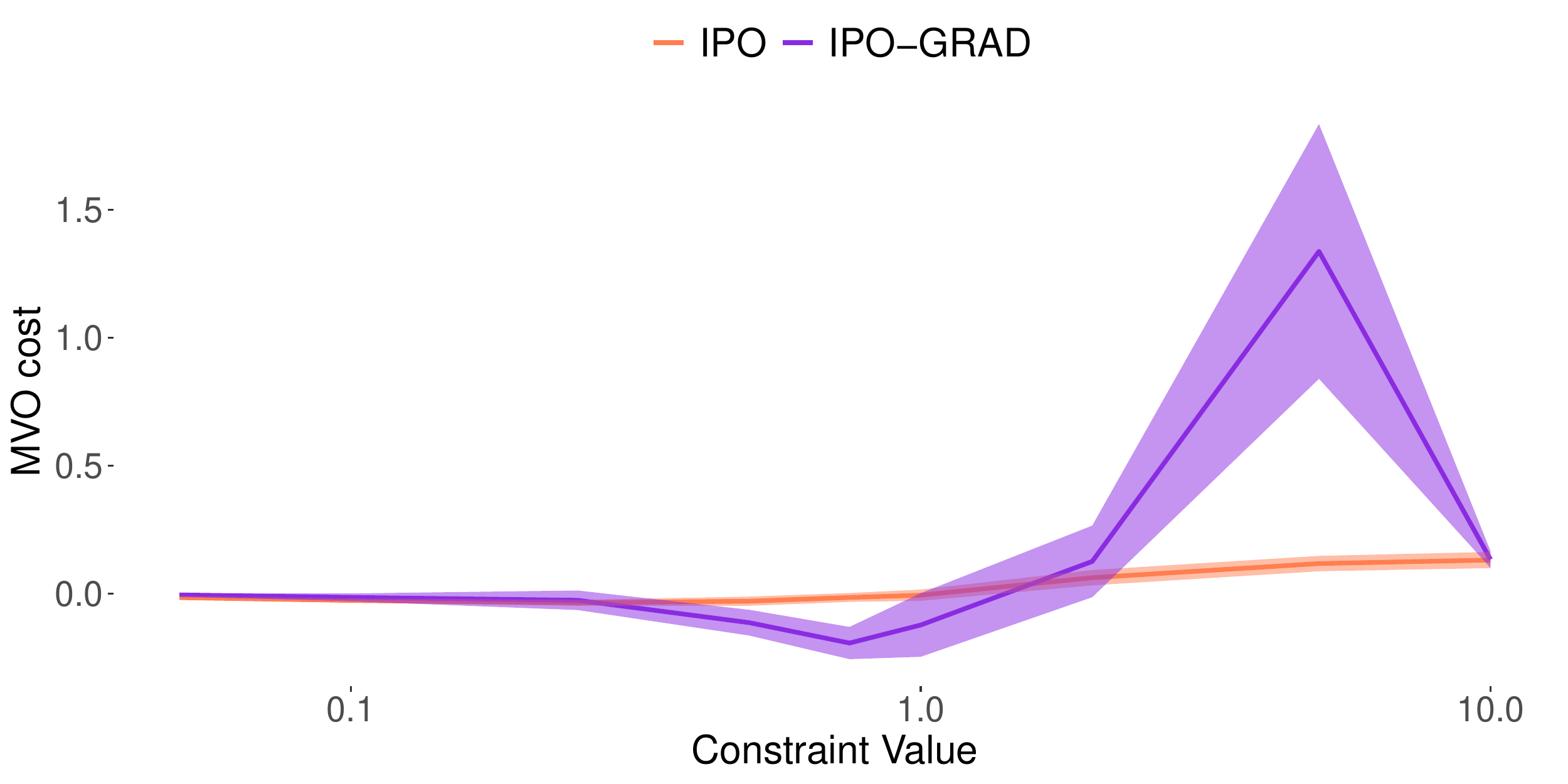}
    \caption{$\delta = 1$, $p = 2$.}
  \end{subfigure}
 \begin{subfigure}[b]{0.35\linewidth}
    \includegraphics[width=\linewidth , trim={00mm 0cm 0cm 0cm},clip]{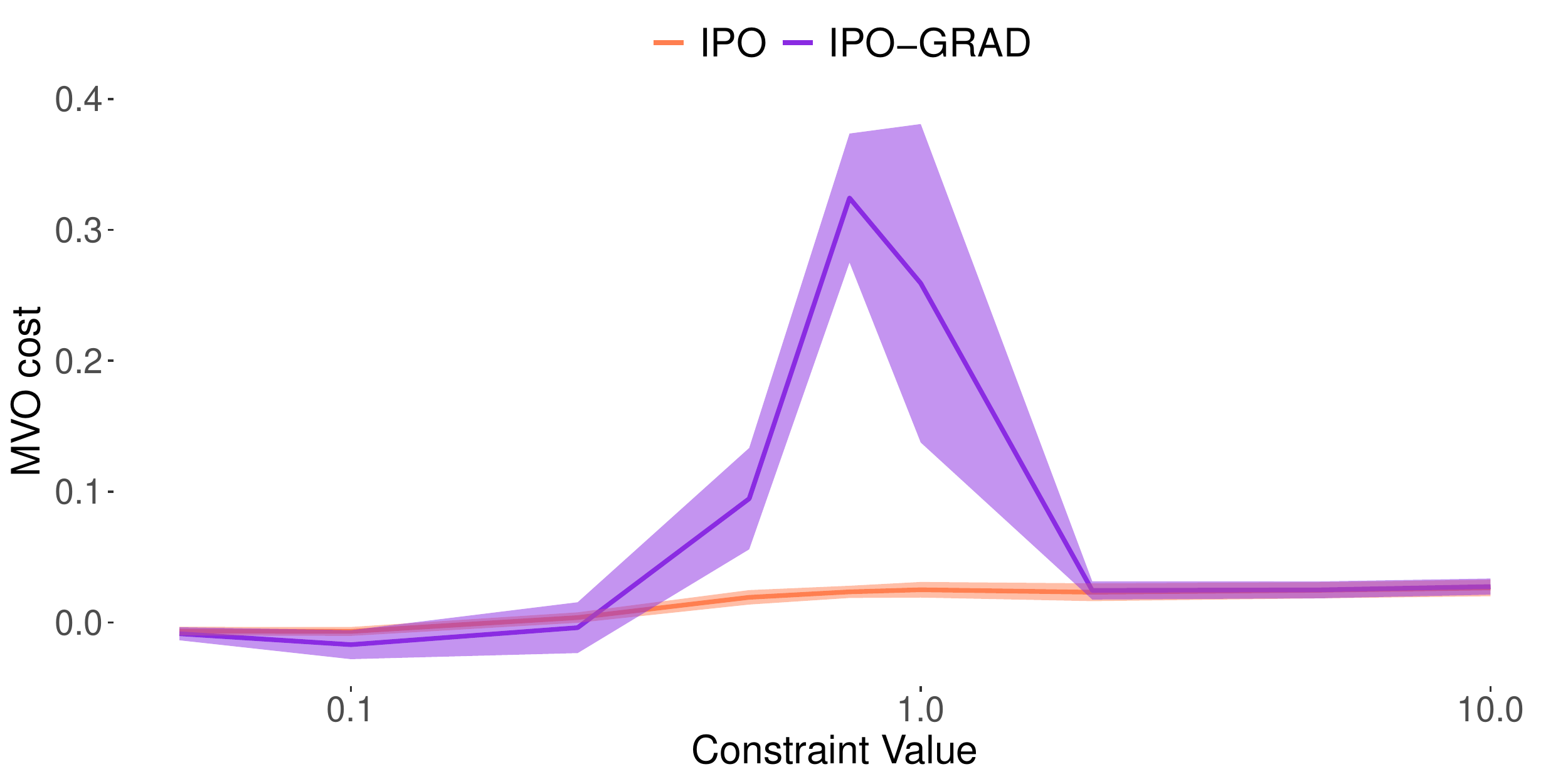}
    \caption{$\delta = 5$, $p = 2$.}
  \end{subfigure}
  \begin{subfigure}[b]{0.35\linewidth}
    \includegraphics[width=\linewidth , trim={00mm 0cm 0cm 0cm},clip]{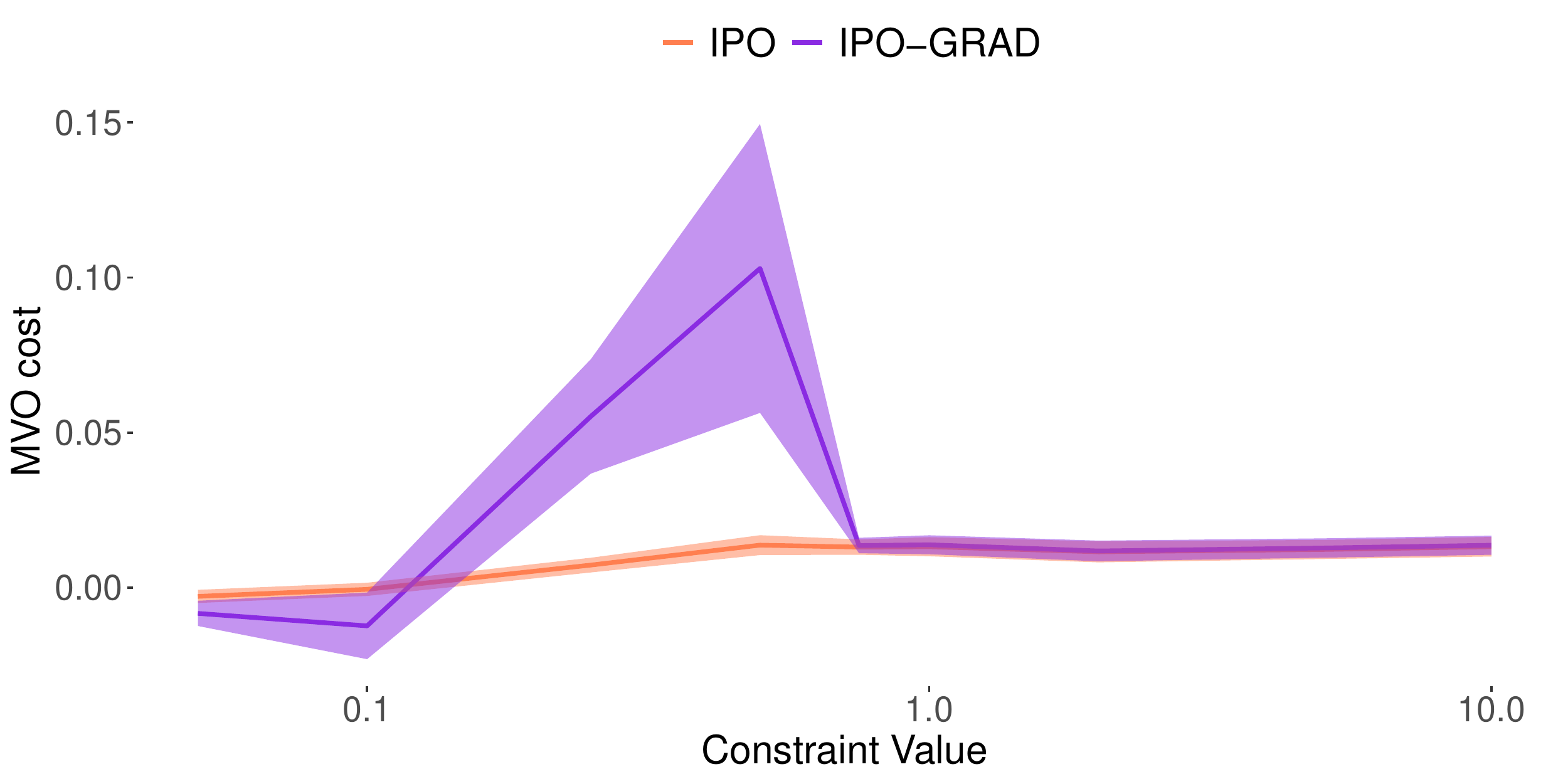}
    \caption{$\delta = 10$, $p = 2$.}
  \end{subfigure}
  \begin{subfigure}[b]{0.35\linewidth}
    \includegraphics[width=\linewidth , trim={00mm 0cm 0cm 0cm},clip]{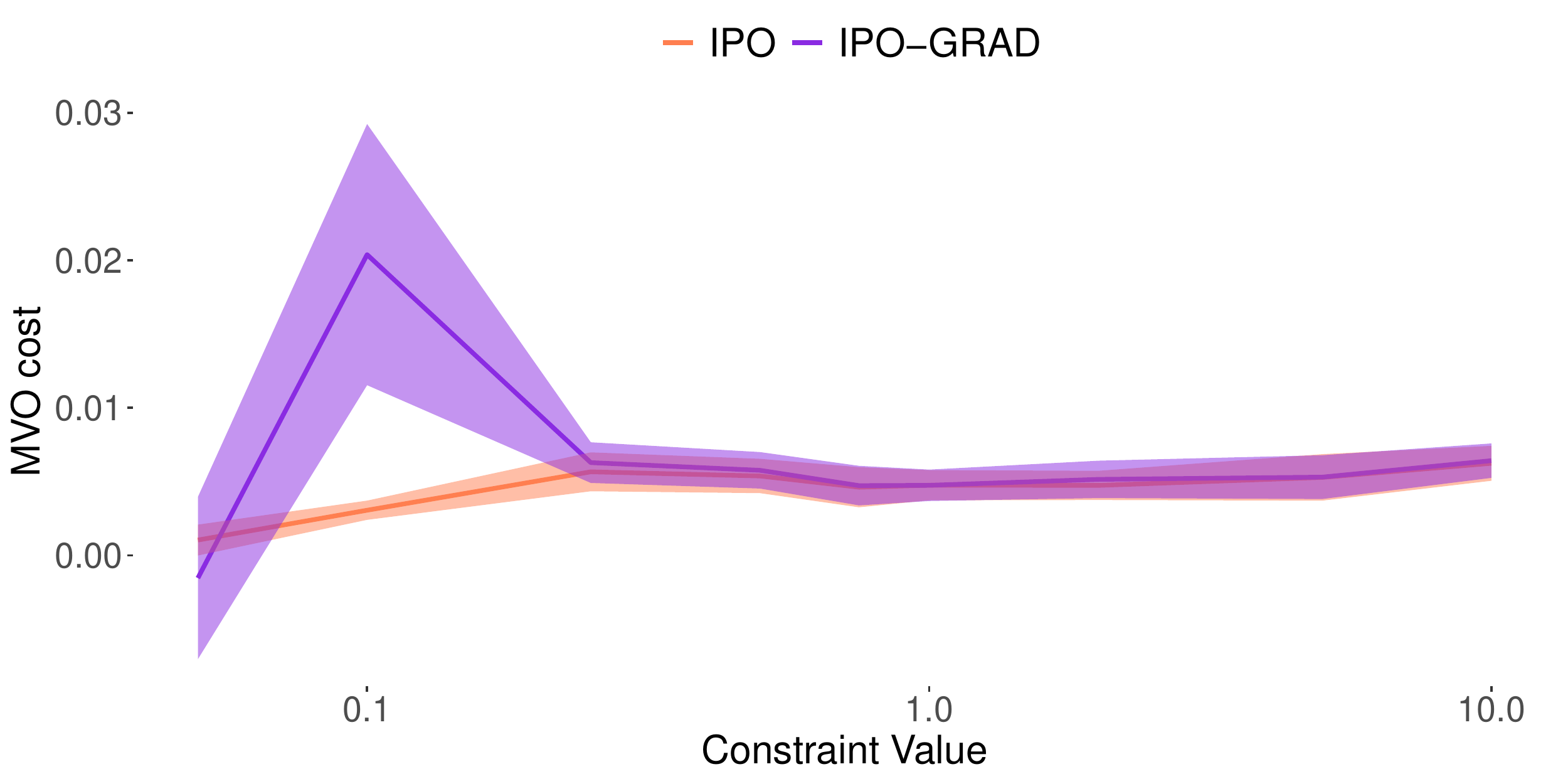}
    \caption{$\delta = 25$, $p = 2$.}
  \end{subfigure}
  \caption{Out-of-sample MVO costs as of function of box constraint value with $p = 2$.}
  \label{fig:sims_3_poly_2}
\end{figure}

\begin{figure}[H]
  \centering
  \begin{subfigure}[b]{0.35\linewidth}
    \includegraphics[width=\linewidth , trim={00mm 0cm 0cm 0cm},clip]{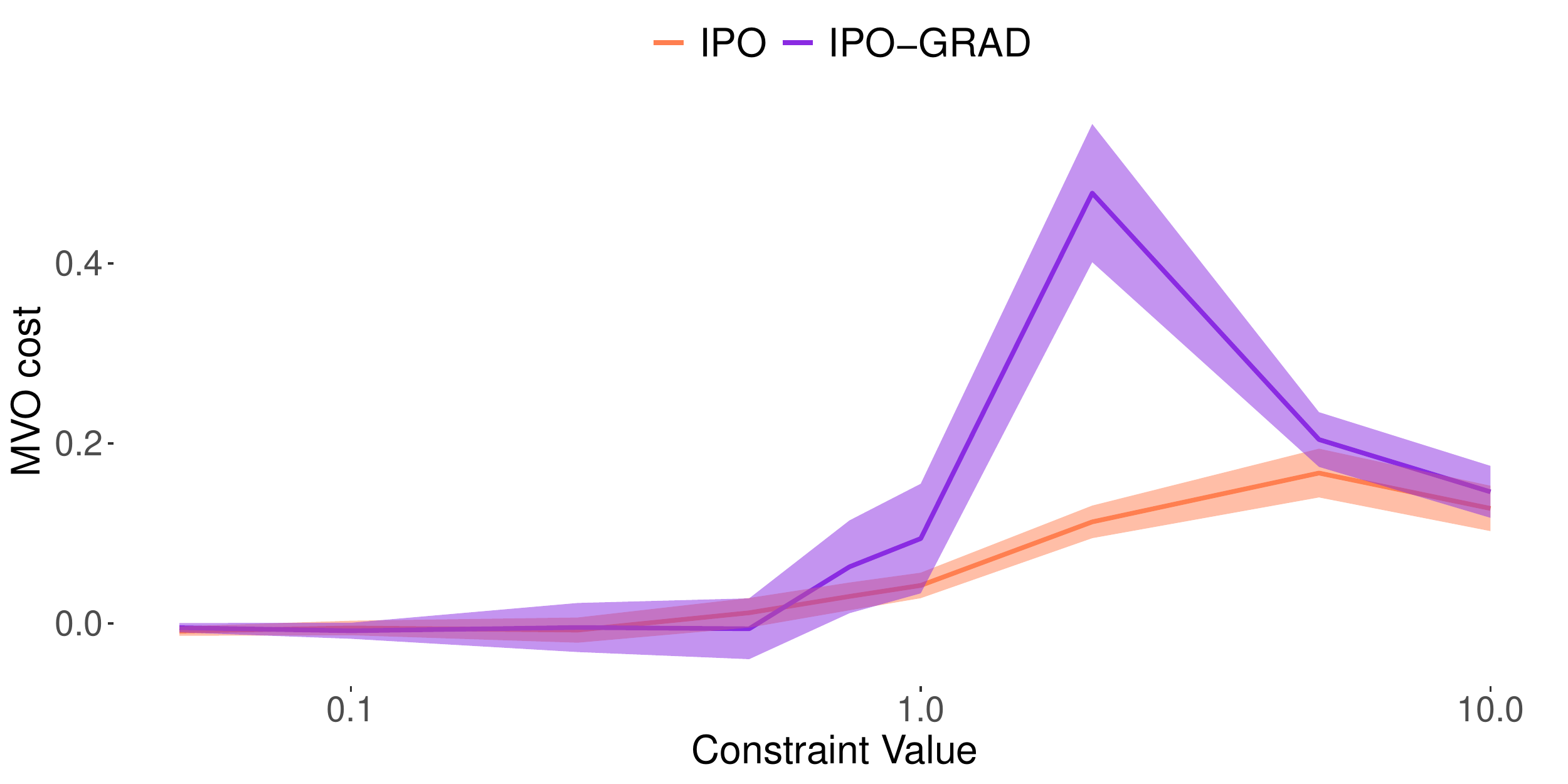}
    \caption{$\delta = 1$, $p = 4$.}
  \end{subfigure}
 \begin{subfigure}[b]{0.35\linewidth}
    \includegraphics[width=\linewidth , trim={00mm 0cm 0cm 0cm},clip]{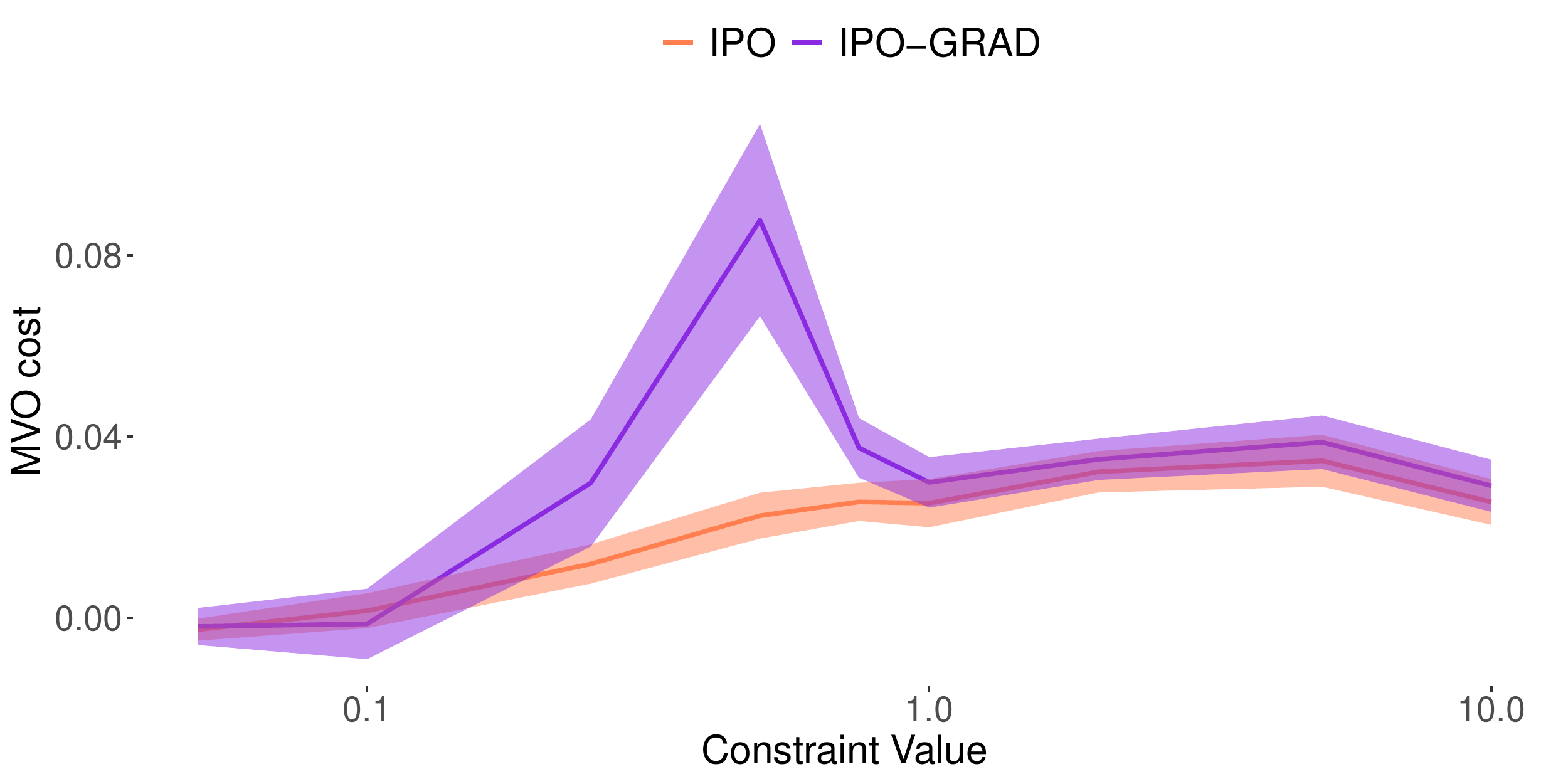}
    \caption{$\delta = 5$, $p = 4$.}
  \end{subfigure}
  \begin{subfigure}[b]{0.35\linewidth}
    \includegraphics[width=\linewidth , trim={00mm 0cm 0cm 0cm},clip]{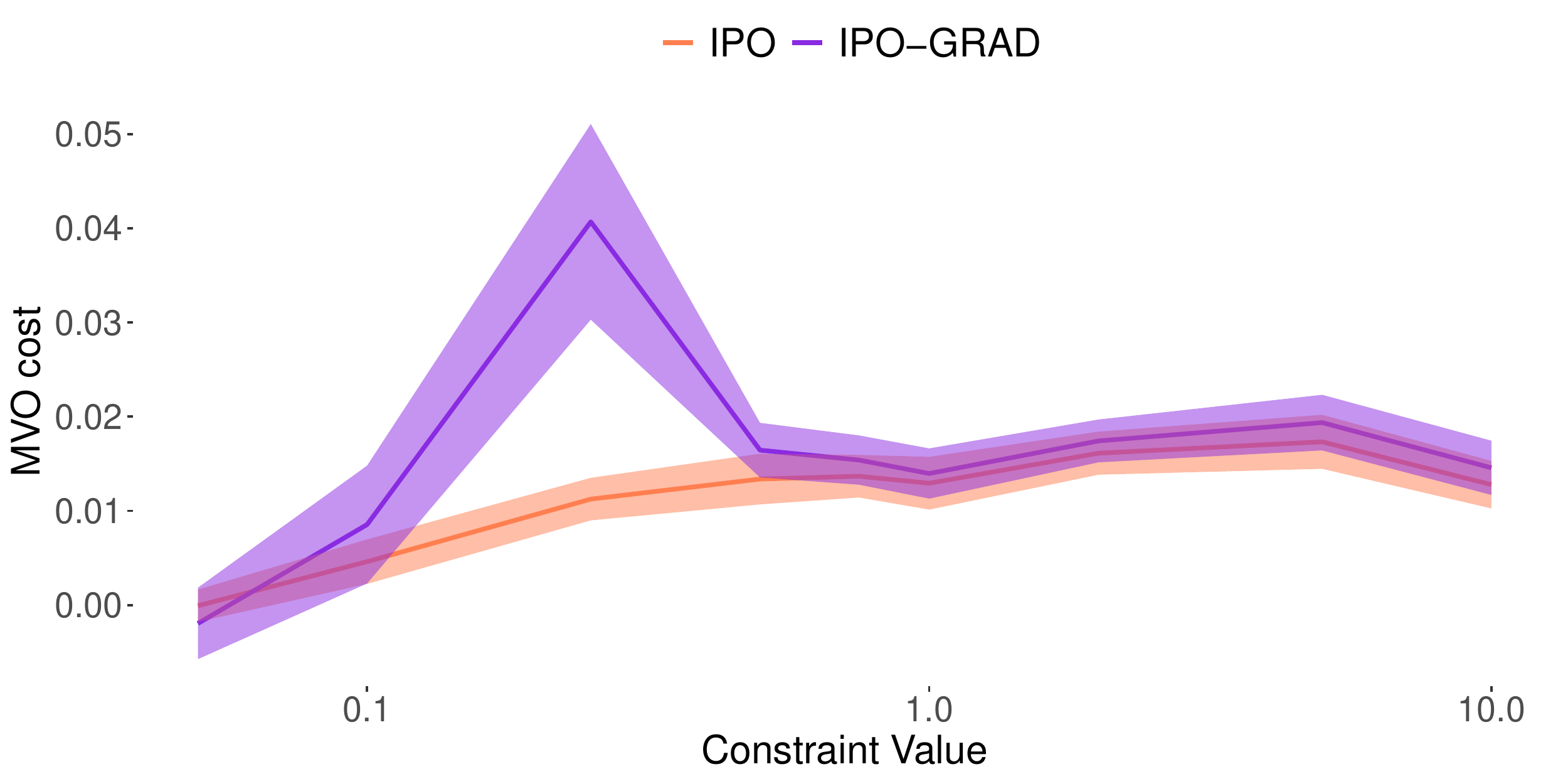}
    \caption{$\delta = 10$, $p = 4$.}
  \end{subfigure}
  \begin{subfigure}[b]{0.35\linewidth}
    \includegraphics[width=\linewidth , trim={00mm 0cm 0cm 0cm},clip]{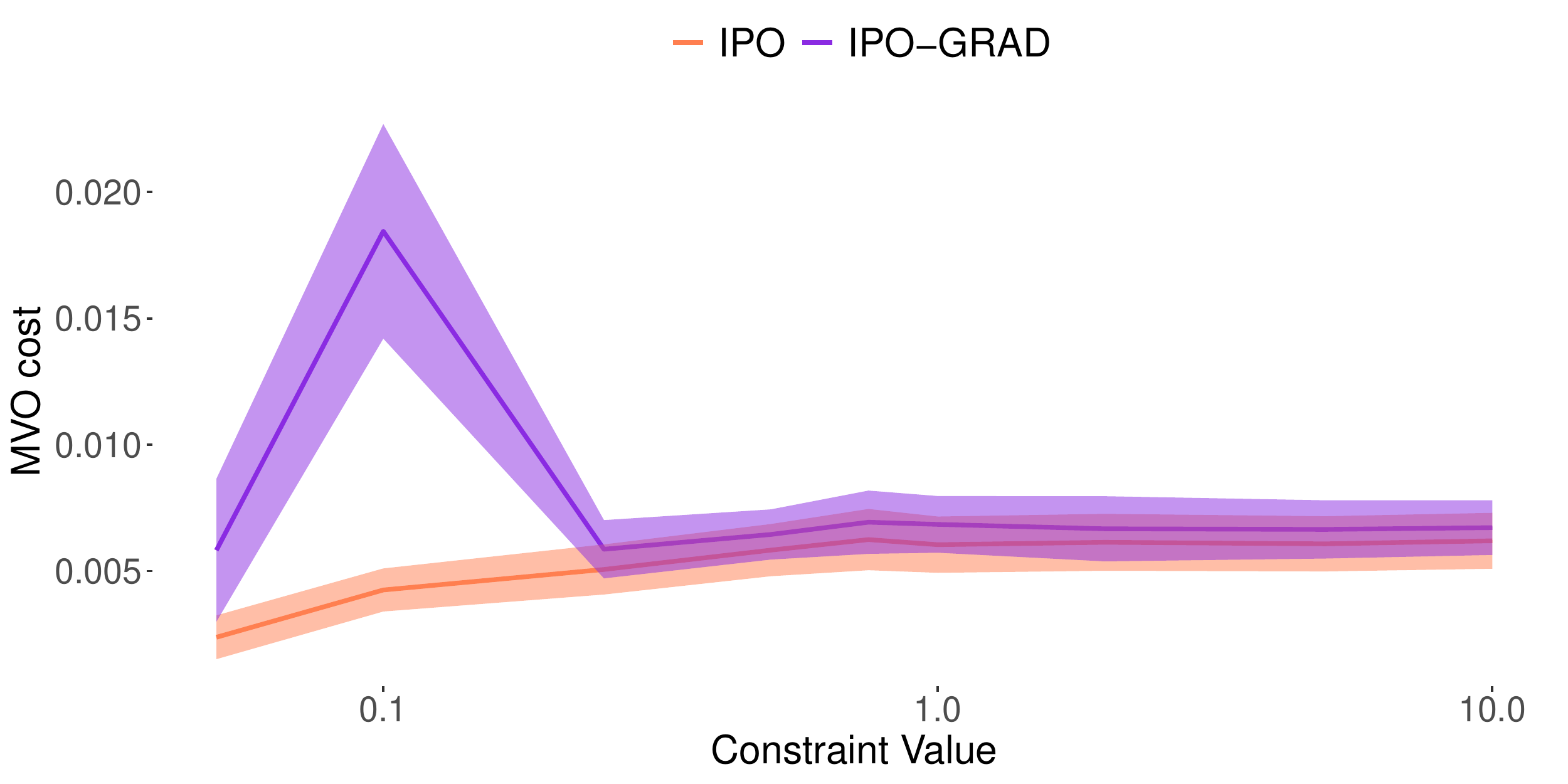}
    \caption{$\delta = 25$, $p = 4$.}
  \end{subfigure}
  \caption{Out-of-sample MVO costs as of function of box constraint value with $p = 4$.}
  \label{fig:sims_3_poly_4}
\end{figure}

\begin{table}[h]
\centering
\begin{tabular}{ l l l}
\hline
 IPO & IPO-GRAD & Iterations\\
\hline
0.023 & 14.389 & 60\\
(0.022,0.024) & (6.9417,22.5718) & (29,94)\\
\hline
\end{tabular}
\caption{Time in seconds for computing the optimal IPO and IPO-GRAD coefficients for an inequality constrained MVO problem. Results are averaged over $360$ instances of simulated data.}
\label{table:sims_3}
\end{table}

\section{Experiments}\label{sec:results}

\textbf{Experiment Setup:}
\newline
We consider an asset universe of $24$ commodity futures markets, described in Table \ref{table:universe}. The daily price data is given from March $1986$ through December $2020$, and is provided by Commodity Systems Inc. Futures contracts are rolled on, at most, a monthly basis in order to remain invested in the most liquid contract, as measured by open-interest and volume. Arithmetic returns are computed directly from the roll-adjusted price data.

In each experiment we follow \citet{Zum2006} and estimate the covariance matrix using an exponential moving average with a decay rate of $0.94$. We consider both univariate and multivariate prediction models. The feature, $\{\toi{\bx}\}$, for univariate models is the $252$-day average return, or trend, for each market. The feature therefore represents a measure of the well-documented `trend' factor, popular to many Commodity Trading Advisors (CTAs) and Hedge Funds (see for example  \citep{Baltas2012}, \citep{Roncalli2011},\citep{Moskowitz2012}). The features for multivariate models are the $252$-day trend and the carry for each market. We follow \citet{Koijen2018} and define  the carry as the expected convenience yield, or cost, for holding that commodity, and is estimated by the percent difference in price between the two futures contracts closest to expiry.

As we will see below, the majority of the IPO and OLS regression coefficients  are not statistically significant at an individual market level. Indeed this is common and well document in many applications of financial forecasting (see for example \citep{Huang2020,Moskowitz2012}). The lack of statistical significance may be indicative of low signal-to-noise levels and/or forecasting model misspecification. Furthermore, the absence of statistical significance does not prohibit the development of profitable portfolio level trading strategies and indeed we observe in Table \ref{table:pooled_reg} that the features are statistically significant at the $95\%$-ile level when evaluated across all markets.

\begin{table}[h]
\centering
\begin{tabular}{ l l l l l}
\hline
Feature & Coefficient & Std. Error  & T-Statistic & P-Value \\
\hline
Carry  & 0.3300 & 0.1654 & 1.9953 & 0.0460\\
Trend & 0.0942 & 0.0324 & 2.9101 & 0.0036\\
\hline
\end{tabular}
\caption{Univariate regression coefficients and t-statistic summary aggregated across all available markets.}
\label{table:pooled_reg}
\end{table}

Each day we form the optimal portfolio weight, $\bz^*(\toi{\hat{\by}})$ at the close of day $i$, and assume execution at the following close, $i+1$. In each experiment, described below, we consider two models for estimating asset returns:
\begin{enumerate}
\item \textbf{OLS:} ordinary-least squares with prediction coefficients, $\hat{\btheta}$.
\item \textbf{IPO:} integrated prediction and optimization, where $\btheta^*$ is determined by the IPO optimization framework described in Section \ref{sec:method}.
\end{enumerate}
We consider $6$ experiments:
\begin{enumerate}
\item Unconstrained MVO program with univariate regression.
\item Unconstrained MVO program with multivariate regression.
\item Equality constrained MVO program with univariate regression.
\item Equality constrained  MVO program with multivariate regression.
\item Inequality constrained MVO program with univariate regression.
\item Inequality constrained MVO program with multivariate regression.
\end{enumerate}

The equality constrained MVO programs are market-neutral: $\mathbb{S} = \{ \bz^T\bone = 0 \}$, whereas the inequality constrained MVO programs are both market-neutral and include lower bound and upper bound market constraints:
$$
\mathbb{S} = \{ \bz^T\bone = 0, -0.125 \leq \bz \leq 0.125  \}.
$$
Note that the results and discussion for the equality constrained MVO and multivariate models are very similar to that of the unconstrained and inequality constrained MVO with univariate prediction models and can be found in Appendix \ref{sec:app_exp}. In order to provide realistic annualized volatilities in the $10\%-20\%$ range,  we fix the risk-aversion parameter to $\delta = 50$. For each experiment, the initial parameter estimation is performed using the first $14$ years of data (March $1986$ through December $1999$). Out-of-sample testing begins in January $2000$ and ends in December $2020$.  We apply a walk-forward training and testing framework whereby the optimal regressions coefficients are updated every $2$ years using all available training data at that point in time. Performance is in excess of the risk-free rate and gross of trading costs.

Each model is evaluated on absolute and relative terms, with a focus on out-of-sample MVO cost and out-of-sample Sharpe ratio cost, provided  by Equation $\eqref{eq:mvo_sr_cost_oos}$.
\begin{equation}\label{eq:mvo_sr_cost_oos}
\begin{split}
c_{\text{MVO}}(\bz,  \by  )  = -\mu(\bz,\by) + \frac{\delta}{2} \sigma^2(\bz,\by), \quad \text{and} \quad  c_{\text{SR}}(\bz,  \by  ) =    -\frac{ \mu(\bz,\by) }{ \sigma(\bz,\by) }
\end{split}
\end{equation}
where
$$ \mu(\bz,\by) = \frac{1}{m}\sum_{i=1}^m \toit{\bz} \toi{ \by } \quad \text{and} \quad \sigma^2(\bz,\by) = \frac{1}{m} \sum_{i=1}^m(\toit{\bz} \toi{ \by } - \mu(\bz,\by))^2,$$
denote the mean and variance of realized daily returns. To quantify the consistency of observed performance metrics, we bootstrap the out-of-sample returns generated by each model using $1000$ samples as follows:

\begin{enumerate}
\item For each $k \in \{1,2,...,1000\}$, sample, without replacement, a batch, $B_k$, with $|B_k| = 252$ observations ($1$ year) from the out-of-sample period.
\item For each sample and model, compute the realized MVO and Sharpe ratio costs using Equation $\eqref{eq:mvo_sr_cost_oos}$.
\end{enumerate}
 In order to fairly compare the realized costs  we ensure that each model uses identical bootstrap observations. We report the dominance ratio (DR); namely the proportion of samples for which the realized cost of the IPO model is less than that of the OLS model.

Our experiments should be interpreted as a proof-of-concept, rather than a fully comprehensive financial study. That said, we believe that the results presented below provide compelling evidence for using IPO for estimating regression coefficients. In general, the IPO models exhibit lower out-of-sample MVO costs and improved economic outcomes in comparison to the traditional OLS-based `predict, then optimize' approach.

\subsection{Experiment 1:  unconstrained with univariate predictions} \label{sec:results_1}

Economic performance metrics and average out-of-sample MVO costs are provided in Table \ref{table:ipo_uncon_uni} for the time period of {2000-01-01} to {2020-12-31} for the unconstrained MVO portfolios with univariate prediction models. Equity growth charts for the same time period are provided in Figure \ref{fig:ipo_equity_uncon_uni}.  We first observe that the IPO model provides higher absolute and risk-adjusted performance, as measured by the MVO cost and Sharpe ratio. Indeed the IPO model produces an out-of-sample MVO cost that is approximately $50\%$ lower and a Sharpe ratio that is approximately $100\%$ larger than that of the OLS model. Furthermore, the IPO models provide more conservative risk metrics, as measured by portfolio volatility, value-at-risk (VaR), and average drawdown (Avg DD). 

In Figure \ref{fig:ipo_cost_uncon_uni} we compare the realized MVO and Sharpe ratio costs across $1000$ out-of-sample realizations. In general we observe that the IPO model exhibits consistently lower MVO costs and generally higher Sharpe ratios than the OLS model. In Figure \ref{fig:ipo_cost_uncon_uni}(a)  we report a dominance ratio of $97\%$ meaning that the IPO model realizes a lower MVO cost in $97\%$ of samples in comparison to the OLS model. Figure \ref{fig:ipo_cost_uncon_uni}(b) reports a dominance ratio of $68\%$.

In Figure \ref{fig:ipo_coef_uncon_uni} we report the estimated univariate regression coefficients and $\pm 1$ standard error bar for the last out-of-sample data fold. As stated earlier, it is clear that the majority of the IPO and OLS regression coefficients are not statistically significant at an individual market basis. Note that for some markets, the IPO model provides very different regression coefficients, in both magnitude and sign, compared to the OLS coefficients. In particular we observe that, with the exception of Cocoa (CC), all IPO regression coefficients are positive. In contrast, $33\%$ of OLS coefficients are negative.

\begin{figure}[h]
  \includegraphics[width=\linewidth,height=3.8cm, trim={0mm 0cm 0cm 0cm},clip]{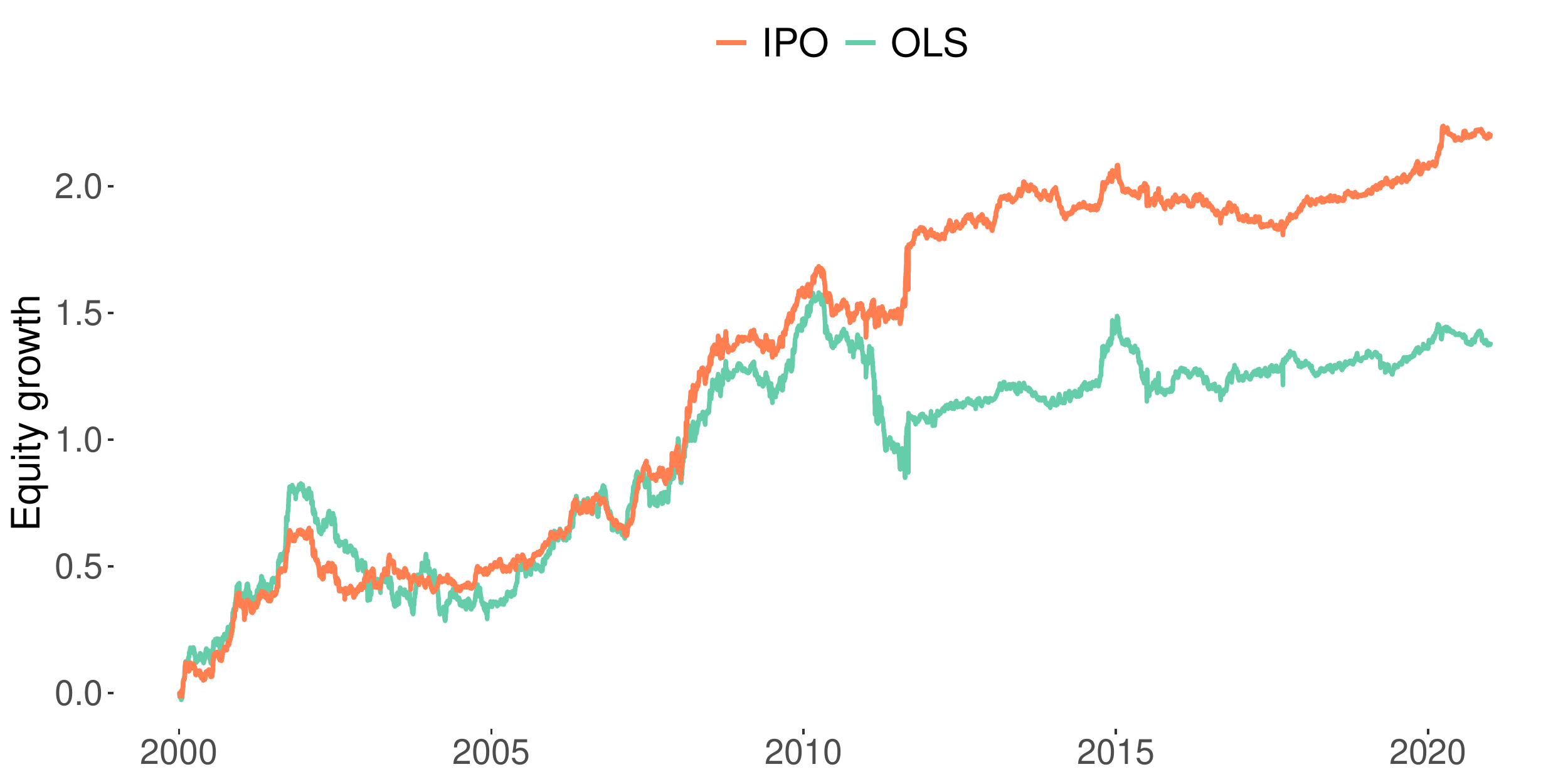}
   \caption{Out-of-sample log-equity growth for the unconstrained  mean-variance program and multivariate IPO and OLS prediction model. }
  \label{fig:ipo_equity_uncon_uni}
\end{figure}

\begin{table}[h]
\centering
\begin{tabular}{lrrrrrr}
\hline
  & Annual Return & Sharpe Ratio & Volatility & Avg Drawdown & Value at Risk & MVO Cost\\
\hline
IPO & 0.1026 & 0.7593 & 0.1352 & -0.0275 & -0.0107 & 0.3544\\
OLS & 0.0644 & 0.3735 & 0.1725 & -0.0426 & -0.0142 & 0.6792\\
\hline
\end{tabular}
\caption{Out-of-sample MVO costs and economic performance metrics for unconstrained mean-variance portfolios with univariate IPO and OLS prediction models.}
\label{table:ipo_uncon_uni}
\end{table}

\begin{figure}[h]
  \centering
  \begin{subfigure}[b]{0.40\linewidth}
    \includegraphics[width=\linewidth , trim={0mm 0cm 0cm 0cm},clip]{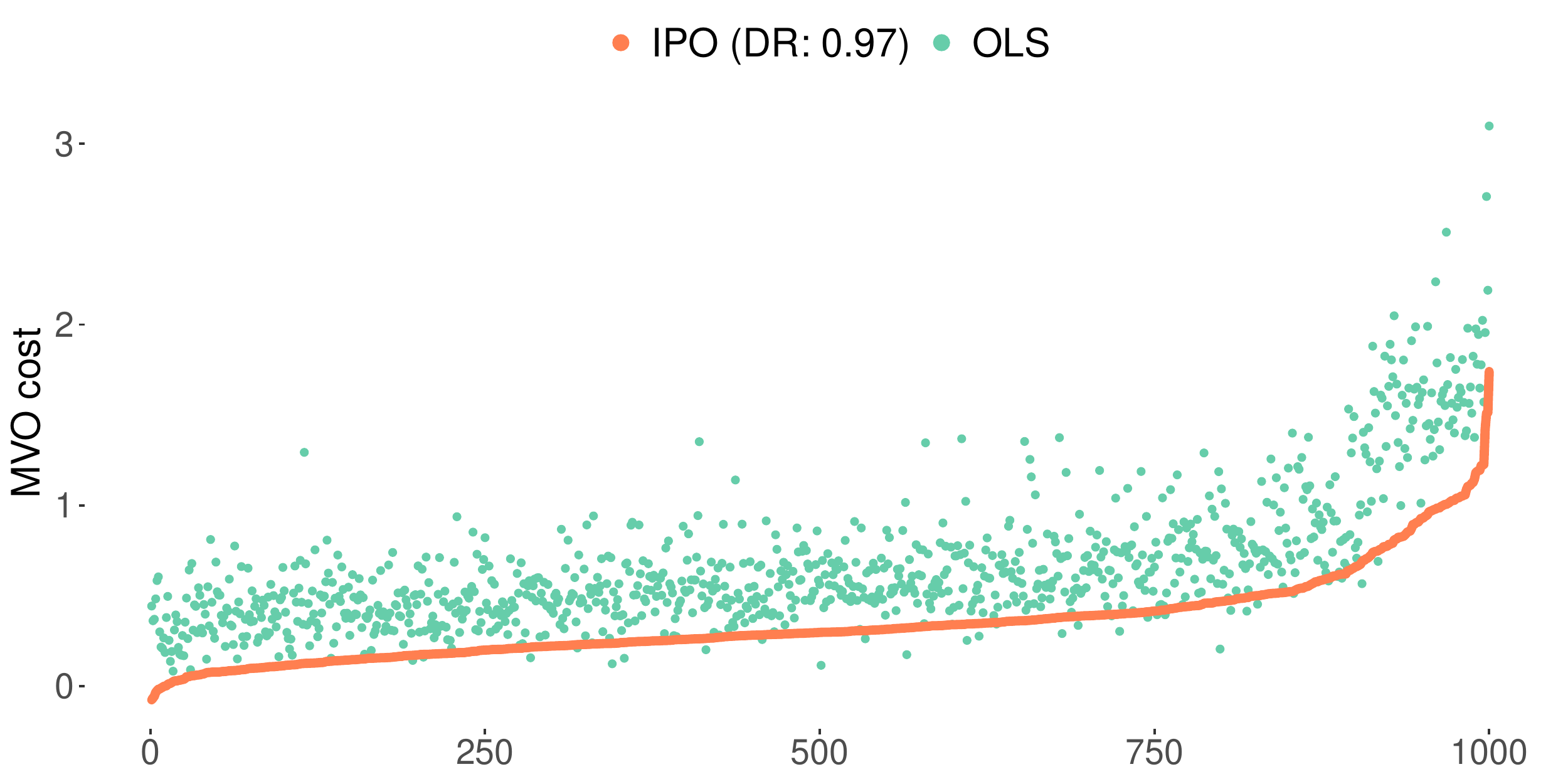}
    \caption{Out-of-sample MVO cost.}
  \end{subfigure}
  \begin{subfigure}[b]{0.40\linewidth}
    \includegraphics[width=\linewidth , trim={0mm 0cm 0cm 0cm},clip]{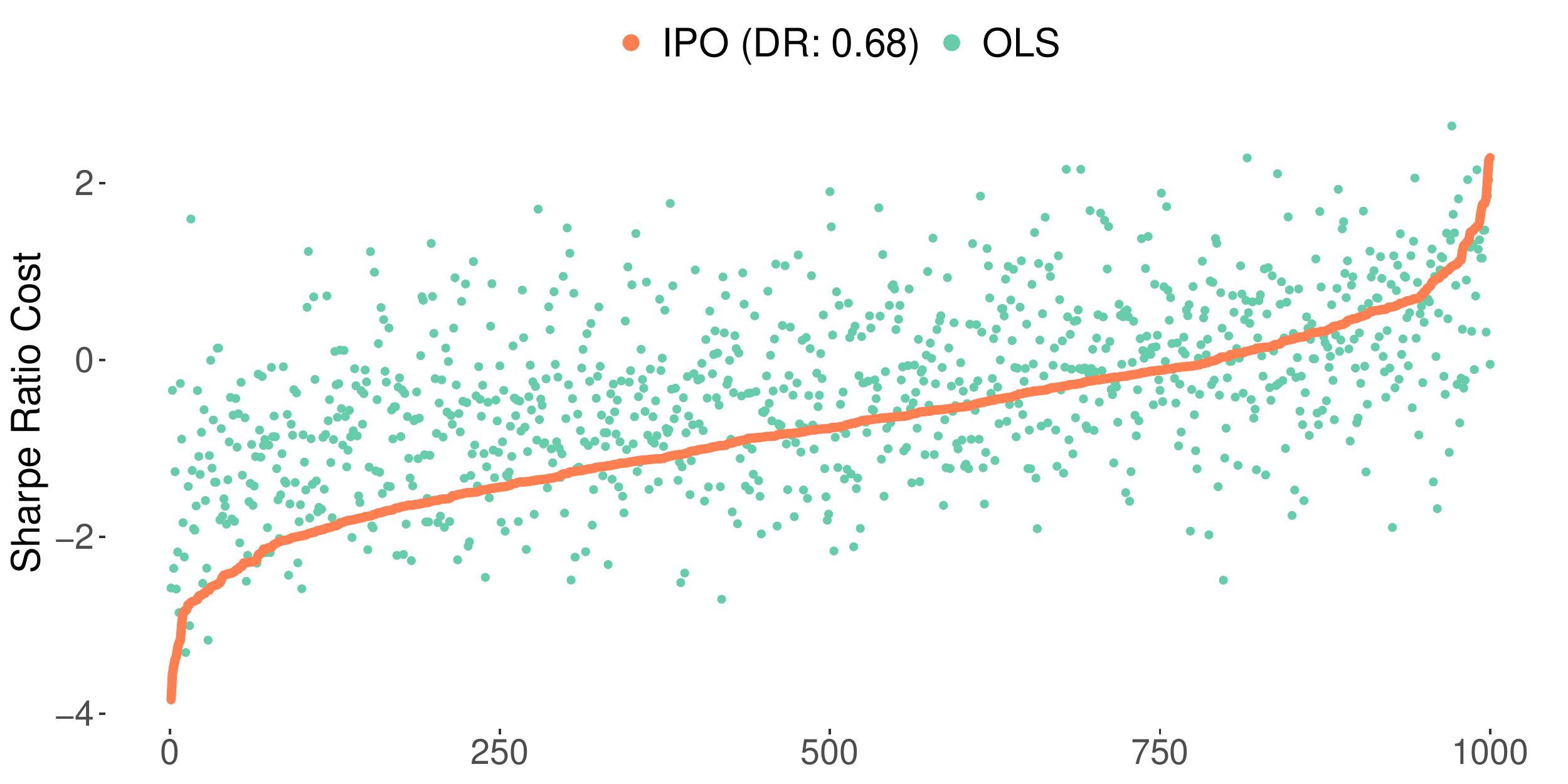}
    \caption{Out-of-sample Sharpe ratio cost.}
  \end{subfigure}
  \caption{Realized out-of-sample MVO and Sharpe ratio costs for the unconstrained  mean-variance program and univariate IPO and OLS prediction models.}
  \label{fig:ipo_cost_uncon_uni}
\end{figure}

\begin{figure}[h]
\includegraphics[width=\linewidth,height = 3.8cm, trim={0mm 0cm 0mm 0cm},clip]{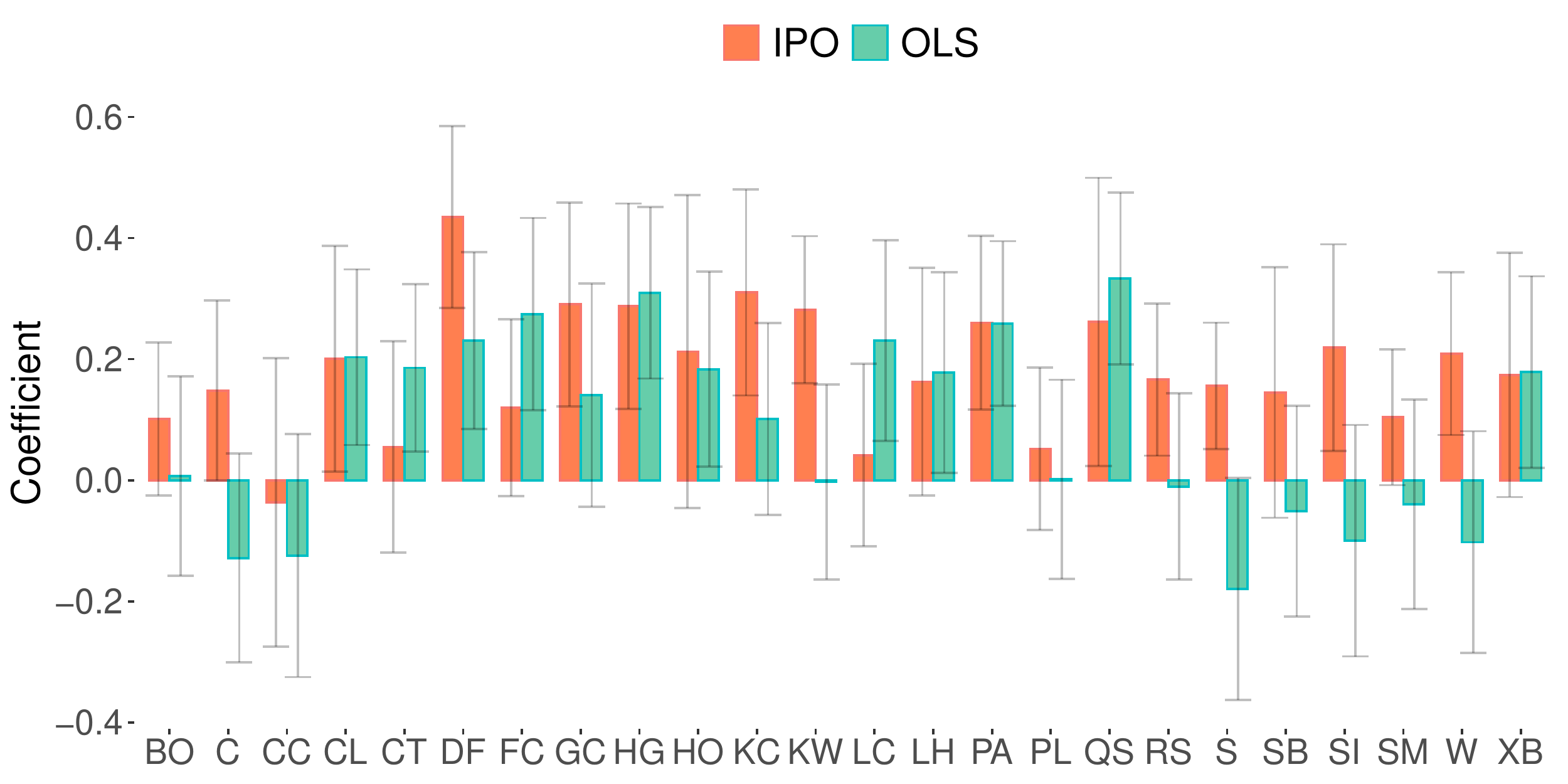}
\caption{Optimal IPO and OLS regression coefficients for the unconstrained  mean-variance program and univariate prediction model. }
\label{fig:ipo_coef_uncon_uni}
\end{figure}

\subsection{Experiment 5: inequality constrained with univariate predictions}

Economic performance metrics and average out-of-sample MVO costs are provided in Table \ref{table:ipo_ineqcon_uni} for the time period of {2000-01-01} to {2020-12-31} for the constrained MVO portfolios with univariate prediction models. Equity growth charts for the same time period are provided in Figure \ref{fig:ipo_equity_ineqcon_uni}. First, observe that the annual returns, risk and  MVO costs are substantially smaller in the presence of portfolio constraints. Indeed this is consistent with the fact that box constraints are themselves a form of portfolio model regularization \citep{Jag2003}. Nonetheless, we observe that the IPO model produces an out-of-sample MVO cost that is approximately $50\%$ lower and a Sharpe ratio that is approximately $85\%$ larger than that of the OLS model. In Figure \ref{fig:ipo_cost_ineqcon_multi} we compare the realized MVO and Sharpe ratio costs across $1000$ bootstrapped sample realizations. Again we observe that the IPO model produces lower MVO and Sharpe ratio costs on average. Observe, however, that the dominance ratios are more modest, with values in the $60\%$-$70\%$ range. This result is intuitive and we would expect the out-of-sample performance of the two models to converge as the portfolio constraints become more strict. Indeed the IPO and OLS model would yield identical results in the limit where the portfolio constraints define a single weight, irrespective of the mean and covariance estimation. Lastly, in Figure \ref{fig:ipo_coef_eqcon_uni} we report the estimated univariate regression coefficients and $\pm 1$ standard error bar for the last out-of-sample data fold. Recall that the IPO coefficients are obtained  by first dropping the inequality constraints and then solving analytically for $\btheta^*$ by Equation $\eqref{eq:theta_star_eqcon}$. The observations and differences between the optimal IPO and OLS coefficients are similar to those discussed in Section \ref{sec:results_1}.

\begin{figure}[h]
  \includegraphics[width=\linewidth,height=3.8cm, trim={0mm 0cm 0cm 0cm},clip]{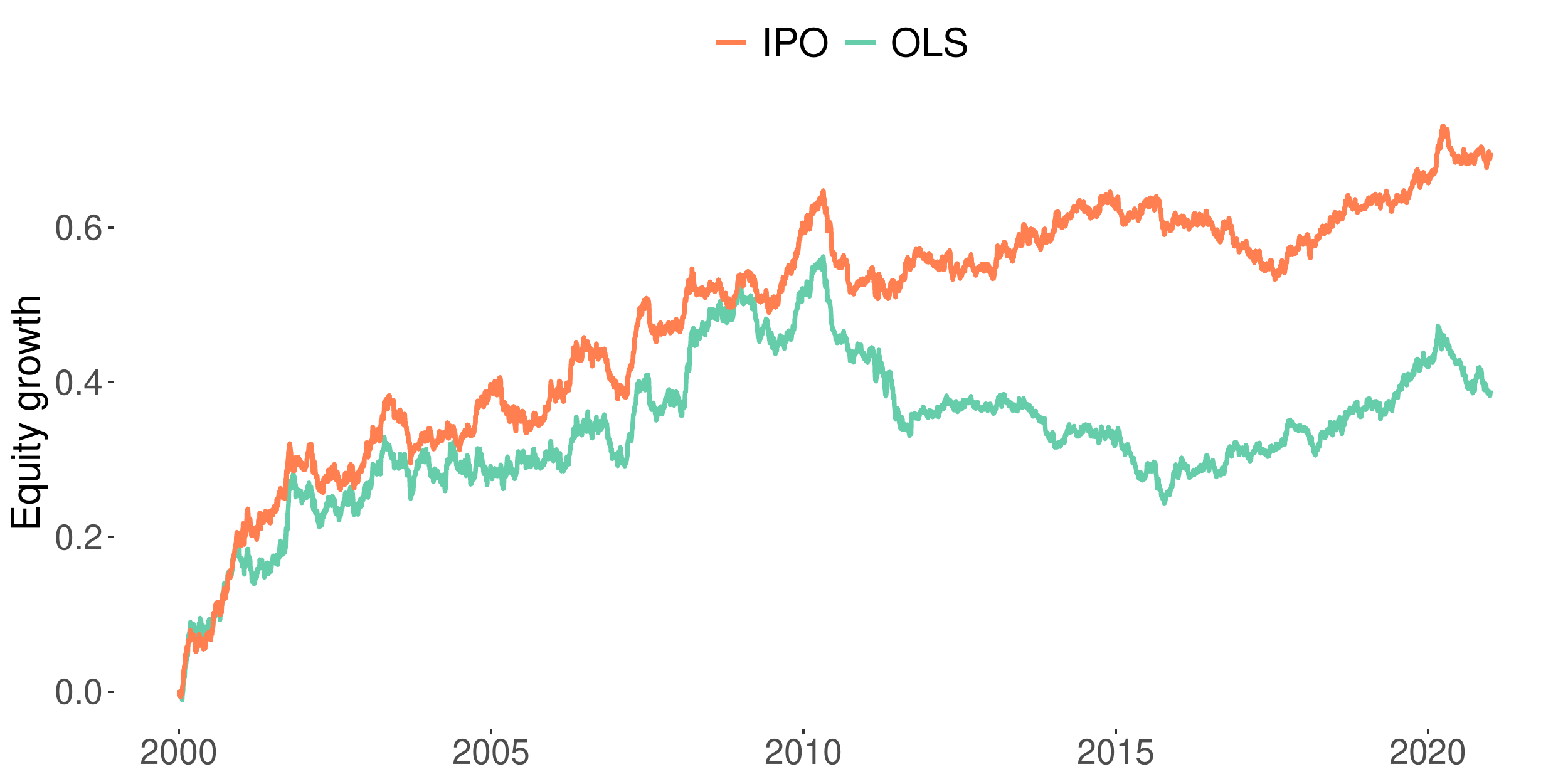}
   \caption{Out-of-sample log-equity growth for the inequality constrained  mean-variance program and multivariate IPO and OLS prediction model. }
  \label{fig:ipo_equity_ineqcon_uni}
\end{figure}

\begin{table}[h]
\centering
\begin{tabular}{lrrrrrr}
\hline
  & Annual Return & Sharpe Ratio & Volatility & Avg Drawdown & Value at Risk & MVO Cost\\
\hline
IPO & 0.0324 & 0.6310 & 0.0513 & -0.0116 & -0.0052 & 0.0335\\
OLS & 0.0181 & 0.3421 & 0.0529 & -0.0174 & -0.0053 & 0.0520\\
\hline
\end{tabular}
\caption{Out-of-sample MVO costs and economic performance metrics for inequality constrained mean-variance portfolios with univariate IPO and OLS prediction models.}
\label{table:ipo_ineqcon_uni}
\end{table}

\begin{figure}[H]
  \centering
  \begin{subfigure}[b]{0.40\linewidth}
    \includegraphics[width=\linewidth , trim={0mm 0cm 0cm 0cm},clip]{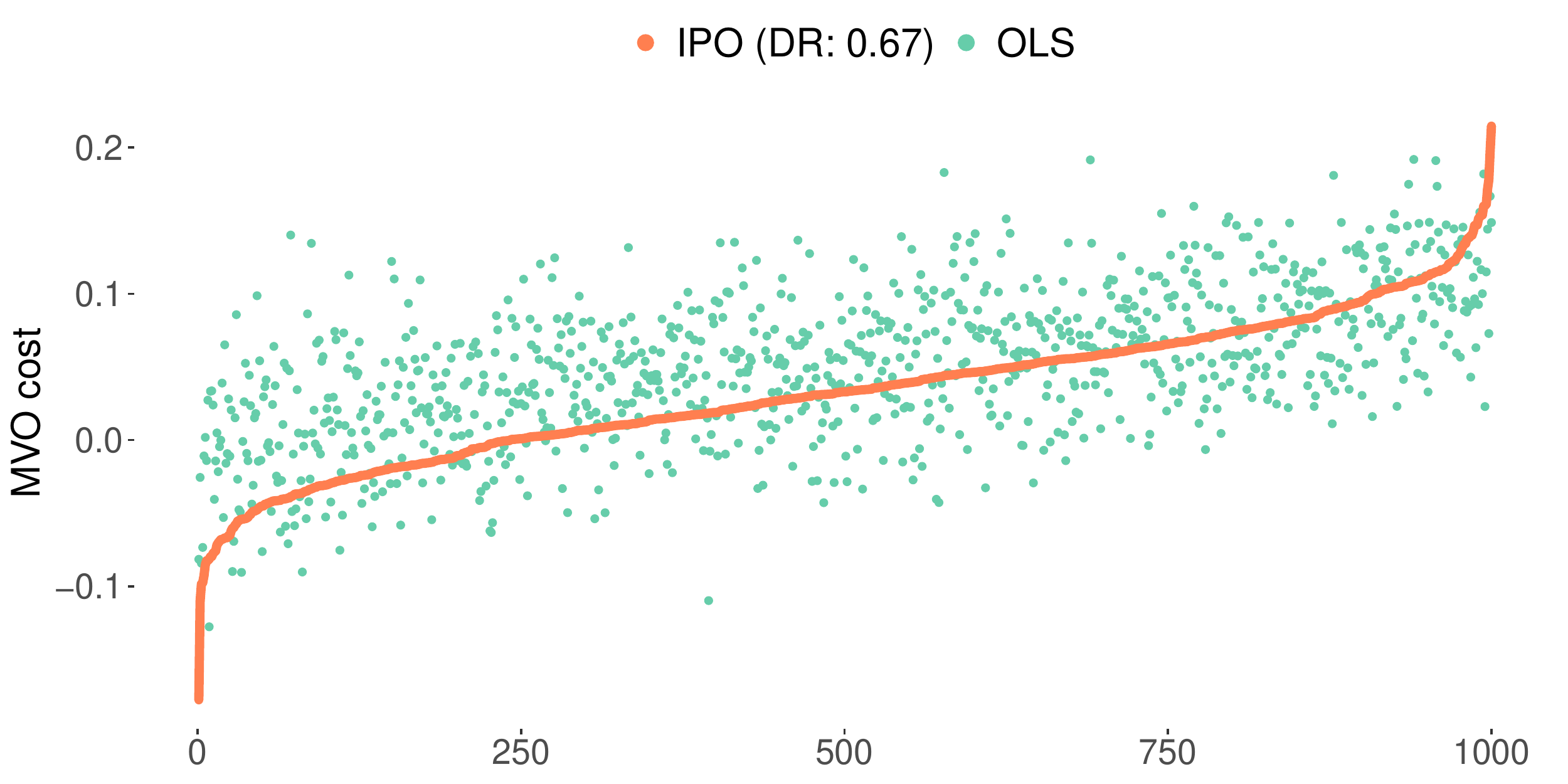}
    \caption{Out-of-sample MVO cost.}
  \end{subfigure}
  \begin{subfigure}[b]{0.40\linewidth}
    \includegraphics[width=\linewidth , trim={0mm 0cm 0cm 0cm},clip]{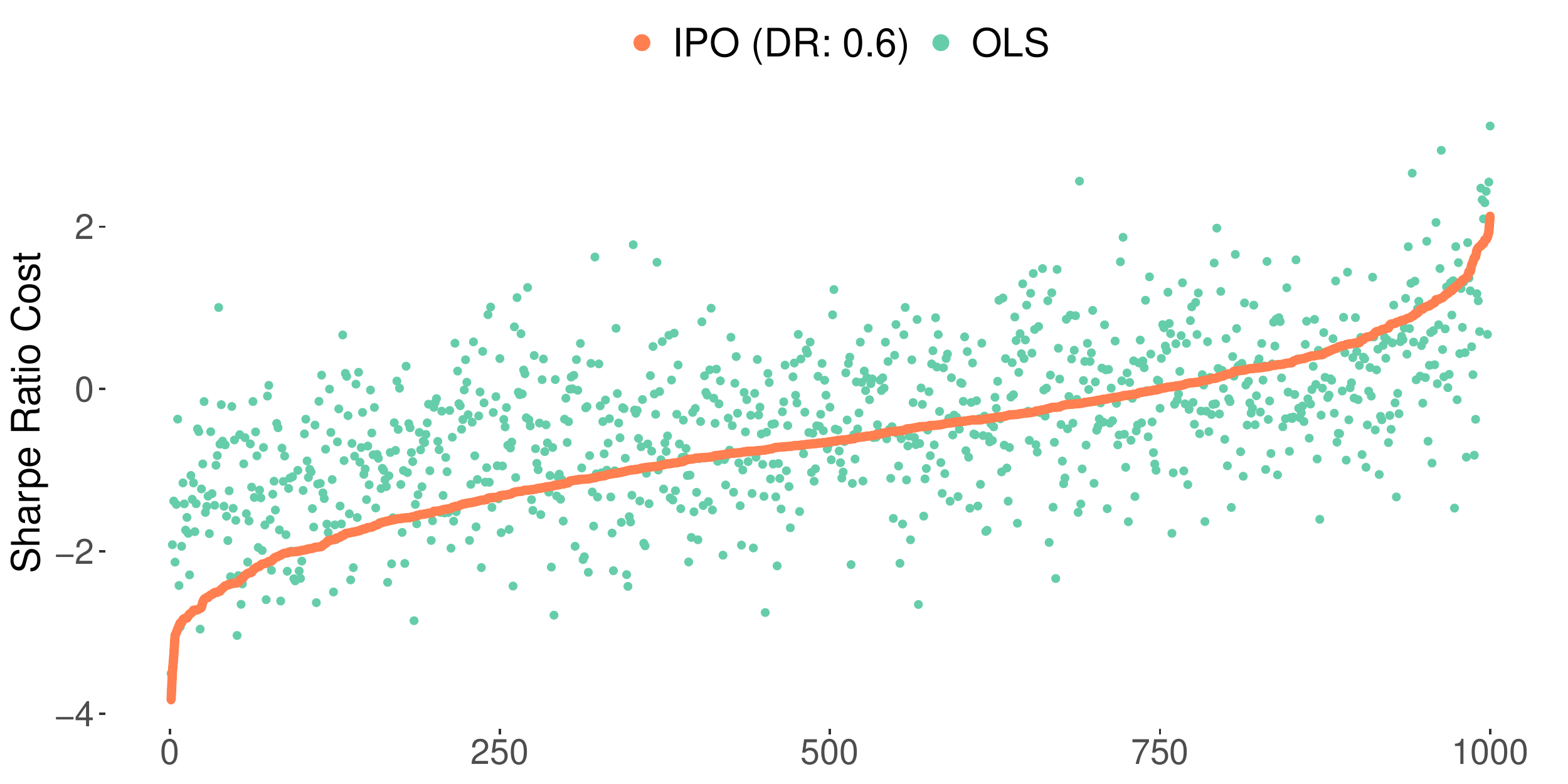}
    \caption{Out-of-sample Sharpe ratio cost.}
  \end{subfigure}
  \caption{Realized out-of-sample MVO and Sharpe ratio costs for the inequality constrained  mean-variance program and univariate IPO and OLS prediction models.}
  \label{fig:ipo_cost_ineqcon_uni}
\end{figure}

\begin{figure}[H]
\includegraphics[width=\linewidth,height = 3.8cm, trim={0mm 0cm 0mm 0cm},clip]{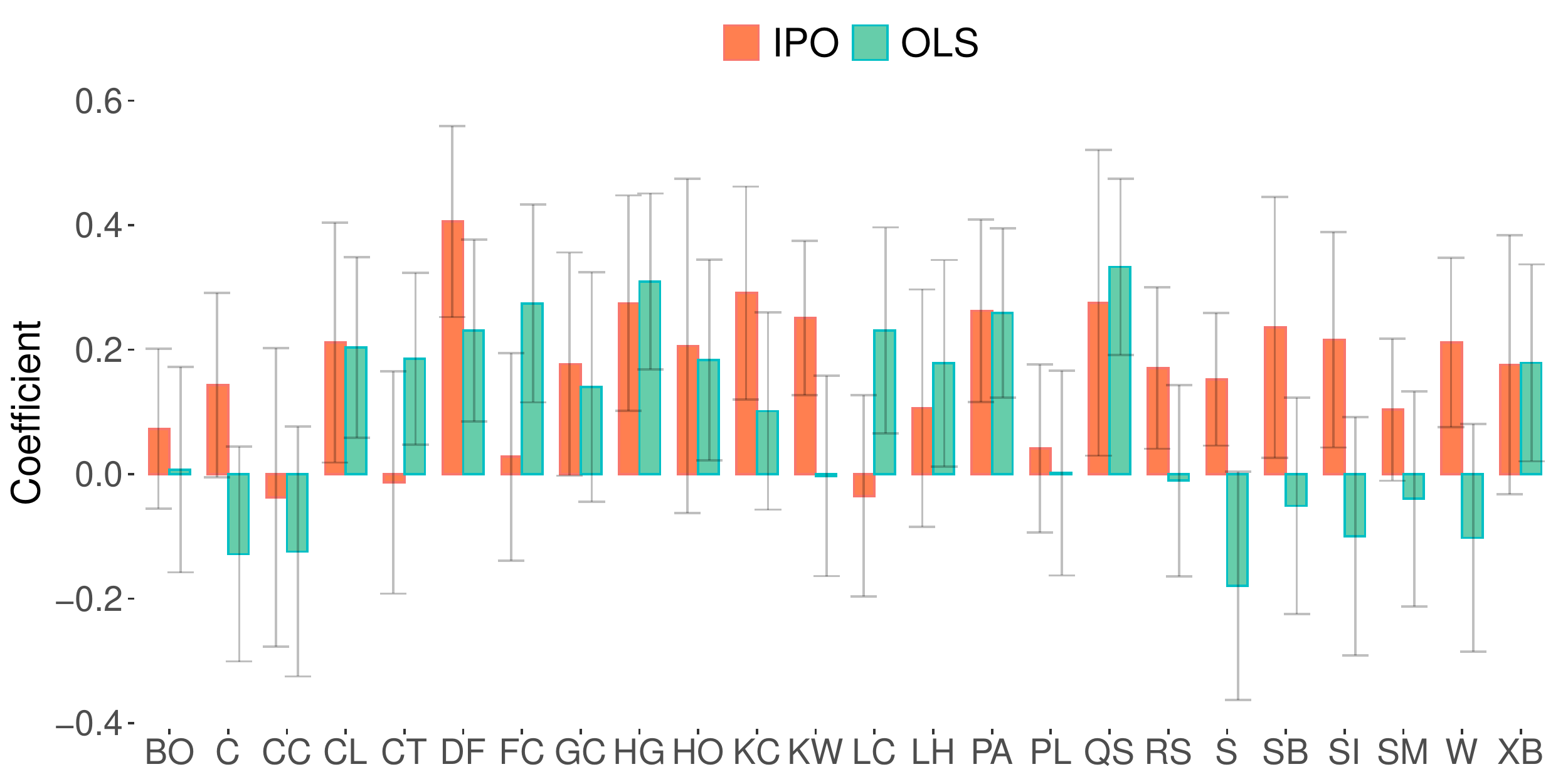}
\caption{Optimal IPO and OLS regression coefficients for the equality constrained  mean-variance program and univariate prediction model. }
\label{fig:ipo_coef_eqcon_uni}
\end{figure}

\section{Conclusion and future work}
In this paper we proposed an integrated prediction and optimization (IPO) framework for optimizing regression coefficients in the context of a mean-variance portfolio optimization. We structured the integrated problem as a bi-level program with a series of lower-level mean-variance optimization programs. We investigated the IPO framework under both univariate and multivariate regression settings and considered the MVO program under various forms of constraints. In a general setting, we presented the current state-of-the-art approach (IPO-GRAD) and restructured the IPO problem as a neural network with a differentiable quadratic programming layer. Where possible, we provided closed-form analytical solutions for the optimal IPO regression coefficients, $\btheta^*$, and the sufficient conditions for uniqueness. We described the sampling distribution properties of $\btheta^*$ and provided the conditions for which $\btheta^*$ is an unbiased estimator of $\btheta$ and provided the expression for the variance.

Extensive numerical simulations demonstrate the computational and performance advantage of the analytical IPO methodology. We demonstrated that, over a wide range of realistic signal-to-noise ratios, the IPO model outperforms the OLS model in terms of minimizing out-of-sample MVO costs. This is true even when the underlying `ground-truth' return generating process is linear in the feature variables. We demonstrated, for a wide range of portfolio sizes, the computational advantage of computing the IPO coefficients analytically, which is on average $10$x-$1000$x faster than the IPO-GRAD methodology. We briefly discussed the computational complexity of the IPO-GRAD methodology and proposed a heuristic which drops the inequality constraints during parameter estimation and invokes the analytical IPO solution. We find that in many instances the IPO-GRAD model overfits the training data, whereas the analytical IPO model produces solutions with lower out-of-sample variance, and in somes cases, improved out-of-sample MVO costs.
We concluded with several experiments using global futures data, under various forms of constraints and prediction model specifications. Out-of-sample results demonstrate that the IPO model provided lower realized MVO costs and superior economic performance in comparison to the traditional  OLS `predict then optimize' approach.

In the presence of general inequality constraints we determined that the current state-of-the-art IPO model is computationally expensive and has a tendency to overfit the training data. We believe that methods for regularizing both the prediction and the decision optimization program, as well as methods for choosing the `best' feature subsets in an integrated setting are interesting areas of future research.
\section{Data Availability Statement}

The data used for experiments was obtained from Commodity Systems Inc: $\text{ https://www.csidata.com}$.


\bibliographystyle{plainnat}
\bibliography{/Users/tars/workspace/phd_thesis/Bibliography/Bibliography}

\newpage
\appendix
\section{Appendix}\label{sec:app_proof}

\subsection{Neural network implementation details}
In the general case we seek to determine a locally optimal solution, $\btheta^*$, to the IPO program by restructuring Program $\eqref{eq:ed1}$ as an end-to-end neural network and applying (stochastic) gradient descent. For compactness, we have temporarily dropped the index notation. From the multivariate chain-rule, the gradient of the IPO objective, $\nabla_\theta L$, can be expressed as:
\begin{equation} \label{grad}
\nabla_\theta L = \frac{\partial L }{\partial \bz^*}\frac{\partial \bz^*}{\partial \hat{\by}} \frac{\partial \hat{\by} }{\partial \btheta}.
\end{equation}
In our case, the MVO cost function, $c$, is smooth and twice differentiable over decision variables, $\bz$, and therefore it is relatively straightforward to compute the gradient $\partial L / \partial \bz^*$. The Jacobian, $\partial \bz^*/ \partial \hat{\by}$, requires differentiation through the $\argmin$ operator. \citet{Amos2017}  demonstrate that rather than forming the Jacobian directly, we can instead compute $\partial L/ \partial  \hat{\by} $ by implicit differentiation of the system of equations  provided by the Karush–Kuhn–Tucker (KKT) conditions at the optimal solution $\bz^*$ to program $\eqref{eq:mvo_ineq}$.

We follow the work of \citet{Amos2017} and begin by first writing the Lagrangian of program $\eqref{eq:mvo_ineq}$:

\begin{equation} \label{eq:lagrange}
\mathcal{L}(\bz, \blambda, \bnu) = -\bz^T \hat{\by} + \frac{\delta}{2} \bz^T \hat{\bV} \bz + \blambda^T(\bG \bz - \bh) + \bnu^T(\bA \bz - \blb),
\end{equation}
where $\blambda \in \mathbb{R}^{d_{iq}}$ and $\bnu \in \mathbb{R}^{d_{eq}}$ are the dual variables of the inequality and equality constraints, respectively.  The KKT optimality conditions for stationarity, primal feasibility, and complementary slackness are given by equations $\eqref{eq:kkt}$.

\begin{equation} \label{eq:kkt}
\begin{split}
-\hat{\by} + \frac{\delta}{2}  \hat{\bV} \bz^*  + \bG^T \blambda^{*T} + \bA^T \bnu^* = 0\\
(\bG \bz* - \bh) \leq 0\\
\blambda^* \geq 0 \\
\blambda^* \cdot (\bG \bz^* - \bh)= 0\\
\bA \bz^* = \blb
\end{split}
\end{equation}
Following \citet{Amos2017}, we take the differentials of these conditions to yield the following system of equations:
\begin{equation} \label{eq:diff}
\begin{split}
\begin{bmatrix}
\delta  \hat{\bV} &  \bG^T &   \bA^T \\
\diag(\blambda^*)\bG & \diag (\bG \bz^* - \bh ) & 0\\
\bA & 0 & 0
\end{bmatrix}
\begin{bmatrix}
{\bdz }\\
{\bdlambda }\\
{\bdnu}
\end{bmatrix}
= -
\begin{bmatrix}
\delta {\bdV} \bz^* - {\bdy } + {\bdG}^T \blambda^* + {\bdA}^T \bnu* \\
\diag(\blambda ^*) {\bdG } \bz^* - \diag(\blambda ^*) {\bdh}\\
{\bdA} \bz^* - {\bdb}
\end{bmatrix}.
\end{split}
\end{equation}

\citet{Amos2017} make two important observations about the system of equations $\eqref{eq:diff}$. The first, is that the left side matrix gives the optimality conditions of the convex quadratic problem, which, when solving by interior-point methods, must be factorized in order to obtain the solution to the decision program \citep{Boyd2004}. Secondly, the right side gives the differentials of the relevant functions at the achieved solution with respect to any of the input parameters. In particular, we seek to compute the Jacobian $\partial \bz^* / \partial \hat{\by}$. As explained by \citet{Amos2017}, the Jacobian $\partial \bz^*/\partial \hat{\by}$ is obtained by letting $\bdy = \bI$ (setting all other differential terms to zero) and solving the system of equations for $\bdz$. From a computation standpoint, the required Jacobian is therefore effectively obtained `for free' upon factorization of the left matrix when obtaining the solution, $\bz^*$, in the forward pass. 

In practice it is inefficient to compute the Jacobian matrix $\partial \bz^* / \partial \hat{\by}$. Instead we compute $\partial c / \partial \hat{\hat{\by}}$ directly by multiplying the backward pass vector by the inverse of the transposed left-hand-side matrix, as shown in equation $\eqref{eq:diff_sol}$.

\begin{equation} \label{eq:diff_sol}
\begin{split}
\begin{bmatrix}
\bar{ \bd }_{\bz} \\
\bar{ \bd }_{\blambda} \\
\bar{ \bd }_{\bnu}
\end{bmatrix}
=
\begin{bmatrix}
\frac{\partial c}{\partial \hat{\by}}\\
-\\
-
\end{bmatrix}
= -
\begin{bmatrix}
\delta  \hat{\bV} &  \bG^T\diag(\blambda^*) &   \bA^T \\
\bG & \diag (\bG \bz^* - \bh ) & 0\\
\bA & 0 & 0
\end{bmatrix} ^{-1}
\begin{bmatrix}
\big( \frac{\partial c}{\partial \bz^*} \big)^T \\
0\\
0
\end{bmatrix}.
\end{split}
\end{equation}

More generally, equation $\eqref{eq:diff_sol}$ allows for efficient computation of the  gradients with respect to any of the MVO input problem variables.  For the reader's interest, we state the gradients for all other problem variables and refer the reader to \citet{Amos2017} for their derivation.

\begin{align*}
\frac{\partial c   }{\partial \hat{\bV}} & = \frac{1}{2} \Big(\bar{ \bd }_{\bz}  \bz^{*T} + \bz^* \bar{ \bd }_{\bz}^T \Big) & \qquad \frac{\partial c   }{\partial \hat{\by}} & = \bar{ \bd }_{\bz} \\
\frac{\partial c   }{\partial \bA} & =  \bar{ \bd }_{\bnu}  \bz^{*T} + \bnu^* \bar{ \bd }_{\bz} ^T   & \qquad \frac{\partial c   }{\partial \blb} & = -\bar{ \bd }_{\bnu}  \\
\frac{\partial c   }{\partial \bG} & = \diag(\blambda ^*) \bar{ \bd }_{\blambda}  \bz^{*T} + \blambda^* \bar{ \bd }_{\bz} ^T   & \qquad  \frac{\partial c   }{\partial \bh} & = - \diag(\blambda ^*) \bar{ \bd }_{\blambda}
\end{align*}


\subsection{Proof of Proposition \ref{prop:l_uncon_uni} }

We begin with the following proposition that will become useful later.

\begin{prop}\label{prop:bvb}
Let $\bV \in \mathbb{R}^{m\times m}$ be a symmetric positive definite matrix. Let $\bm B \in \mathbb{R}^{m\times n}$ and consider the quadratic form $\bA = \bm B^T \bV \bm B$. Then $\bA$ is a symmetric positive definite matrix if $\bm B$ has full column rank.
\end{prop}

\begin{proof}\label{proof:bvb_proof}
The symmetry of $\bA$ follows directly from the definition. To prove positive definiteness, let $\bx \in  \mathbb{R}^n$ be a non-zero vector and consider the quadratic form $\bx^T \bA \bx$:
$$
\bx^T \bA \bx = \bx^T \bm B^T \bV \bm B \bx = \by^T \bV \by.
$$
Clearly $\by^T \bV \by > 0$ for all $\by \neq 0$ and $\by^T \bV \by = 0 \iff \bm B\bx = 0$. But $\bm B$ has full column rank and therefore the only solution to $\bm B \bx=0$ is the trivial solution $\bx=0$. It follows then that $\bx^T \bm B^T \bV \bm B \bx > 0$ and therefore $\bA$ is positive definite.
\end{proof}

Let $\mathbb{S} = \mathbb{R}^{d_z}$, then the solution to the  MVO Program $\eqref{eq:mvo}$ is given by:
\begin{equation}\label{eq:z_star_uncon_general}
\bz^*(\toi{\hat{\by}}) = \frac{1}{\delta} \toinv{\hat{\bV}} \toi{\hat{\by}} =  \frac{1}{\delta} \toinv{\hat{\bV}} \bP \diag(\toi{\bx}) \btheta.
\end{equation}
Direct substitution of $\eqref{eq:z_star_uncon_general}$ into Equation $\eqref{eq:mvo_full}$ yields the following quadratic objective:
\begin{equation}\label{eq_a:l_theta_uncon}
L(\btheta) = \frac{1}{2}\btheta^T \bH(\bx)  \btheta -  \btheta^T \bd(\bx,\by)
\end{equation}
where
\begin{equation}\label{eq_a:d_uncon}
\bd(\bx,\by)  = \frac{1}{m\delta} \sum_{i = 1}^m  \Big( \diag(\toi{\bx}) \bP^T \toinv{\hat{\bV}}  \toi{\by} \Big)
\end{equation}
and
\begin{equation}\label{eq_a:h_uncon}
\bH(\bx) =  \frac{1}{m \delta} \sum_{i = 1}^m  \Big( \diag(\toi{\bx}) \bP^T \toinv{\hat{\bV}}  \toi{\bV}  \toinv{\hat{\bV}} \bP \diag(\toi{\bx})  \Big).
\end{equation}
Applying Proposition \ref{prop:bvb} it follows then that if there exists an $\toi{\bx}$ such that $\toi{\bx_j} \neq 0 \quad \forall j \in 1,...,d_x$ then $\bH(\bx,\by) \succ 0$ and therefore $\eqref{eq_a:qp_uncon_general}$ is a convex quadratic function.
\begin{equation} \label{eq_a:qp_uncon_general}
\begin{split}
\minimize_{ \btheta \in \bTheta} \quad & \frac{1}{2}\btheta^T \bH(\bx)  \btheta -  \btheta^T \bd(\bx,\by).
\end{split}
\end{equation}
In the absence of constraints on $\btheta$, then the first-order conditions are necessary and sufficient for optimality, with optimal IPO coefficients given by:
\begin{equation} \label{eq_a:theta_star_uncon_general}
\btheta^* = \bH(\bx)^{-1} \bd(\bx,\by)
\end{equation}


\subsection{Proof of Proposition \ref{prop:uncon_bias_general} }

Let $\btheta^*$ and $\bd_{\bu}(\bx)$ be as defined by Equation $\eqref{eq:theta_star}$ and Equation $\eqref{eq:du}$, respectively. It follows then that:
\begin{equation}\label{eq_a:proof_uncon_bias_general}
\begin{split}
\mathbb{E}[\btheta^*] & = \mathbb{E}\big[ \bH(\bx)^{-1}  \bd(\bx,\by) \big] \\
& = \bH(\bx)^{-1} \frac{1}{m\delta} \sum_{i = 1}^m  \Big( \diag(\toi{\bx}) \bP^T \toinv{\hat{\bV}} \mathbb{E}\big[  \toi{\by} \big] \Big)\\
& = \bH(\bx)^{-1} \frac{1}{m\delta} \sum_{i = 1}^m  \Big( \diag(\toi{\bx}) \bP^T \toinv{\hat{\bV}} \bP \diag(\toi{\bx}) \btheta \Big)\\
& =  \bH(\bx)^{-1} \bd_{\bu}(\bx) \btheta
\end{split}
\end{equation}
Corollary \ref{cor:uncon_unbias} follows directly from Equation $\eqref{eq_a:proof_uncon_bias_general}$. Observe that when $\toi{\hat{\bV}} = \toi{\bV} \forall i \in \{1, ..., m\}$, then
$$
\bH(\bx) =  \frac{1}{m \delta} \sum_{i = 1}^m  \Big( \diag(\toi{\bx}) \bP^T \toinv{\hat{\bV}} \bP \diag(\toi{\bx})  \Big).
$$
It follows then that:
\begin{equation}\label{eq:proof_uncon_unbias}
\begin{split}
\mathbb{E}[\btheta^*] & = \mathbb{E}\big[ \bH(\bx)^{-1}  \bd(\bx,\by) \big] \\
& = \bH(\bx)^{-1} \frac{1}{m\delta} \sum_{i = 1}^m  \Big( \diag(\toi{\bx}) \bP^T \toinv{\hat{\bV}} \bP \diag(\toi{\bx}) \btheta \Big)\\
& =   \bH(\bx)^{-1}\bH(\bx) \btheta\\
& = \btheta.
\end{split}
\end{equation}

\subsection{Proof of Proposition \ref{prop:uncon_var} }

Let $\{\toi{\by}\}_{i=1}^m$ be independent random variables with  $\toi{\by} \sim \mathcal{N}(\bP \diag(\toi{\bx}) \btheta , \bSigma )$. Let $\hat{\bSigma}$ and $\bM$ be as defined by Equation $\eqref{eq:sigma_hat}$ and Equation $\eqref{eq:m_uncon}$, respectively. It follows then that:

\begin{equation}\label{eq_a:proof_uncon_var}
\begin{split}
\Var(\btheta^*) & = \Var \big( \bH(\bx)^{-1}  \bd(\bx,\by) \big) \\
& = \bH(\bx)^{-1} \Var \Big( \frac{1}{m\delta} \sum_{i = 1}^m   \diag(\toi{\bx}) \bP^T \toinv{\hat{\bV}}  \toi{\by} \Big) \bH(\bx)^{-1} \\
& = \bH(\bx)^{-1}  \frac{1}{m^2 \delta^2 } \sum_{i = 1}^m  \Big( \diag(\toi{\bx}) \bP^T \toinv{\hat{\bV}} \Var(\toi{\by})   \toinv{\hat{\bV}} \bP \diag(\toi{\bx}) \Big)\bH(\bx)^{-1}\\
& =  \bH(\bx)^{-1}  \frac{1}{m^2 \delta^2 } \sum_{i = 1}^m  \Big( \diag(\toi{\bx}) \bP^T \toinv{\hat{\bV}} \hat{\bSigma}  \toinv{\hat{\bV}} \bP \diag(\toi{\bx}) \Big)\bH(\bx)^{-1} \\
& = \bH(\bx)^{-1} \bM \bH(\bx)^{-1}.
\end{split}
\end{equation}

\subsection{Proof of Proposition \ref{prop:uncon_tracking_error} }
Let $\bz^*(\toi{\by})$ and $\bz^*(\toi{\hat{\by}})$ be as defined in Equation $\eqref{eq:mvo}$ and Equation $\eqref{eq:z_star_uncon}$, respectively. Recall, the objective function of the minimum tracking-error representation of the IPO program is:
\begin{equation} \label{eq_a:uncon_tracking_error_obj}
L_{\text{te}}(\btheta) = \frac{1}{2m} \sum_{i = 1}^m \lVert \bz^*(\toi{\hat{\by}}) - \bz^*(\toi{\by})  \rVert^2_{\toi{\bV}}
\end{equation}
The first-order necessary conditions for optimality of Program $\eqref{eq:mvo}$ state:
\begin{equation}\label{eq_a:mvo_fonc}
 \toi{\bV} \bz^*( \toi{\by} ) =  \toi{\by}
\end{equation}
Expanding  Equation $\eqref{eq_a:uncon_tracking_error_obj}$ and substituting in Equation $\eqref{eq_a:mvo_fonc}$  completes the proof:
\begin{equation} \label{eq_a:uncon_tracking_error_obj}
\begin{split}
L_{\text{te}}(\btheta) & = \frac{1}{2m} \sum_{i = 1}^m \big( \bz^*(\toi{\hat{\by}}) - \bz^*(\toi{\by}) \big)^T \toi{\bV} \big( \bz^*(\toi{\hat{\by}}) - \bz^*(\toi{\by}) \big)\\
& =  \frac{1}{m} \sum_{i = 1}^m   \frac{1}{2} \bz^*(\toi{\hat{\by}})^T \toi{\bV} \bz^*(\toi{\hat{\by}}) -\bz^*(\toi{\hat{\by}})^T \toi{\bV}\bz^*(\toi{\by})  + \frac{1}{2} \bz^*(\toi{\by})^T \toi{\bV} \bz^*(\toi{\by}) \\
& =  \frac{1}{m} \sum_{i = 1}^m   \frac{1}{2} \bz^*(\toi{\hat{\by}})^T \toi{\bV} \bz^*(\toi{\hat{\by}}) -\bz^*(\toi{\hat{\by}})^T\toi{\by}  + \bz^*(\toi{\by})^T \toi{\bV} \bz^*(\toi{\by}) .
\end{split}
\end{equation}
Note that the proof of Proposition \ref{prop:eqcon_tracking_error} follows a similar argument for the case of equality constrained MVO portfolios.

\subsection{Proof of Proposition \ref{prop:l_eqcon_uni} }
In the presence of equality constraints then the solution to the MVO Program is given by:
\begin{equation}\label{eq_a:z_star_eqcon_general}
\begin{split}
\bz^*(\toi{\hat{\by}}) &  = \frac{1}{\delta} \bF ( \bF^T \toi{\hat{\bV}} \bF )^{-1} \bF^T  \bP \diag(\toi{\bx})\btheta +  (\bI - \bF ( \bF^T \toi{\hat{\bV}} \bF )^{-1} \bF^T \toi{\hat{\bV}} )\bz_0,
\end{split}
\end{equation}
where   $\bz_0$ be a particular element of $ \mathbb{S} =\{\bA \bz = \blb\}$ and $\bF$ is a basis for the nullspace of $\bA$.

Direct substitution of $\eqref{eq_a:z_star_eqcon_general}$ into Equation $\eqref{eq:mvo_full}$ yields the following quadratic objective:
\begin{equation}\label{eq_a:l_theta_eqcon}
L(\btheta) = \frac{\delta}{2} \sum_{i = 1}^m \toi{L_1}(\btheta) - \sum_{i=1}^m \toi{L_2}(\btheta),
\end{equation}
where
\begin{equation}\label{eq_a:l1}
\begin{split}
\toi{L_1}(\btheta) & = \bz^*(\toi{\hat{\by}})^T \toi{\bV} \bz^*(\toi{\hat{\by}})\\
& = \frac{1}{\delta^2} \btheta^T\diag(\toi{\bx}) \bP^T \bF ( \bF^T \toi{\hat{\bV}} \bF )^{-1} \bF^T \toi{\bV} \bF ( \bF^T \toi{\hat{\bV}} \bF )^{-1} \bF^T  \bP \diag(\toi{\bx})\btheta \\
 & \quad + \frac{2}{\delta} \btheta^T\diag(\toi{\bx}) \bP^T \bF ( \bF^T \toi{\hat{\bV}} \bF )^{-1} \bF^T \toi{\bV} (\bI - \bF ( \bF^T \toi{\hat{\bV}} \bF )^{-1} \bF^T \toi{\hat{\bV}} )\bz_0\\
 & \quad + \bz_0^T (\bI - \toi{\hat{\bV}}\bF ( \bF^T \toi{\hat{\bV}} \bF )^{-1} \bF^T  )\toi{\bV} (\bI - \bF ( \bF^T \toi{\hat{\bV}} \bF )^{-1} \bF^T \toi{\hat{\bV}} )\bz_0
\end{split}
\end{equation}

and

\begin{equation}\label{eq_a:l2}
\begin{split}
\toi{L_2}(\btheta) & = \bz^*(\toi{\hat{\by}})^T\toi{\by}\\
& = \frac{1}{\delta} \btheta^T\diag(\toi{\bx}) \bP^T \bF ( \bF^T \toi{\hat{\bV}} \bF )^{-1} \bF^T\toi{\by} + \bz_0^T (\bI - \bF ( \bF^T \toi{\hat{\bV}} \bF )^{-1} \bF^T \toi{\hat{\bV}} )\toi{\by}
\end{split}
\end{equation}

Simplifying Equation $\eqref{eq_a:l_theta_eqcon}$ and removing constant terms yields the following
quadratic objective:
\begin{equation}\label{eq_a:l_theta_eqcon_simp}
L(\btheta) = \frac{1}{2}\btheta^T \bH_{\text{eq}}(\bx)  \btheta -  \btheta^T \bd_{\text{eq}}(\bx,\by)
\end{equation}
where
\begin{equation}\label{eq_a:d_uncon}
\bd_{\text{eq}}(\bx,\by)  = \frac{1}{m\delta} \sum_{i = 1}^m  \Big( \diag(\toi{\bx}) \bP^T \bF ( \bF^T \toi{\hat{\bV}} \bF )^{-1} \bF^T (\toi{\by} - \toi{\bV} (\bI - \bF ( \bF^T \toi{\hat{\bV}} \bF )^{-1} \bF^T \toi{\hat{\bV}} )\bz_0)  \Big)
\end{equation}
and
\begin{equation}\label{eq_a:h_uncon}
\bH_{\text{eq}}(\bx) =  \frac{1}{m \delta} \sum_{i = 1}^m  \Big( \diag(\toi{\bx}) \bP^T \bF ( \bF^T \toi{\hat{\bV}} \bF )^{-1} \bF^T \toi{\bV} \bF ( \bF^T \toi{\hat{\bV}} \bF )^{-1} \bF^T  \bP \diag(\toi{\bx})  \Big).
\end{equation}
Again, applying Proposition \ref{prop:bvb} it follows then that if there exists an $\toi{\bx}$ such that $\toi{\bx_j} \neq 0 \quad \forall j \in 1,...,d_x$ then $\bH_\text{eq}(\bx) \succ 0$ and therefore $\eqref{eq_a:qp_eqcon_general}$ is a convex quadratic program:
\begin{equation} \label{eq_a:qp_eqcon_general}
\begin{split}
\minimize_{ \btheta \in \bTheta} \quad & \frac{1}{2}\btheta^T \bH_{\text{eq}}(\bx)  \btheta -  \btheta^T \bd_{\text{eq}}(\bx,\by).
\end{split}
\end{equation}
In the absence of constraints on $\btheta$, then the first-order conditions are necessary and sufficient for optimality, with optimal IPO coefficients given by:
\begin{equation} \label{eq_a:theta_star_eqcon_general}
\btheta_{\text{eq}}^* = \bH_{\text{eq}}(\bx)^{-1}  \bd_{\text{eq}}(\bx,\by).
\end{equation}

\subsection{Proof of Proposition \ref{prop:eqcon_bias_general} }

Let $\btheta_{\text{eq}}^*$ and $\bd_{\be}(\bx)$ be as defined by Equation $\eqref{eq:theta_star_eqcon}$ and Equation $\eqref{eq:de}$, respectively. It follows then that:
\begin{equation}\label{eq_a:proof_eqcon_bias_general}
 \begin{split}
\mathbb{E}[\btheta_{\text{eq}}^*] & = \mathbb{E}\big[ \bH_{\text{eq}}(\bx)^{-1}  \bd_{\text{eq}}(\bx,\by) \big] \\
& = \bH_{\text{eq}}(\bx)^{-1} \frac{1}{m\delta} \sum_{i = 1}^m  \Big( \diag(\toi{\bx}) \bP^T \bF ( \bF^T \toi{\hat{\bV}} \bF )^{-1} \bF^T \mathbb{E}\big[  \toi{\by} \big] \Big)\\
& = \bH_{\text{eq}}(\bx)^{-1} \frac{1}{m\delta} \sum_{i = 1}^m  \Big( \diag(\toi{\bx}) \bP^T \bF ( \bF^T \toi{\hat{\bV}} \bF )^{-1} \bF^T \bP \diag(\toi{\bx})  \Big)\\
& =  \bH_{\text{eq}}(\bx)^{-1} \bd_{\be}(\bx) \btheta
\end{split}
\end{equation}
Corollary \ref{cor:eqcon_unbias} follows directly from Equation $\eqref{eq_a:proof_eqcon_bias_general}$. Observe that when $\toi{\hat{\bV}} = \toi{\bV} \forall i \in \{1, ..., m\}$, then
$$
\bH_{\text{eq}}(\bx) =  \frac{1}{m\delta} \sum_{i = 1}^m  \Big( \diag(\toi{\bx}) \bP^T \bF ( \bF^T \toi{\hat{\bV}} \bF )^{-1} \bF^T \bP \diag(\toi{\bx})  \Big).
$$
It follows then that:
\begin{equation}\label{eq_a:proof_eqcon_unbias}
\begin{split}
\mathbb{E}[\btheta_{\text{eq}}^*] & = \mathbb{E}\big[ \bH_{\text{eq}}(\bx)^{-1}  \bd_{\text{eq}}(\bx,\by) \big] \\
& = \bH_{\text{eq}}(\bx)^{-1} \frac{1}{m\delta} \sum_{i = 1}^m  \Big( \diag(\toi{\bx}) \bP^T \bF ( \bF^T \toi{\hat{\bV}} \bF )^{-1} \bF^T \bP \diag(\toi{\bx})  \Big)\\
& =   \bH_{\text{eq}}(\bx)^{-1}\bH_{\text{eq}}(\bx) \btheta\\
& = \btheta
\end{split}
\end{equation}

\subsection{Proof of Proposition \ref{prop:eqcon_var} }

Let $\{\toi{\by}\}_{i=1}^m$ be independent random variables with  $\toi{\by} \sim \mathcal{N}(\bP \diag(\toi{\bx}) \btheta , \bSigma )$. Let $\hat{\bSigma}$ and $\bM_{\text{eq}}$ be as defined by Equation $\eqref{eq:sigma_hat}$ and Equation $\eqref{eq:m_eq}$, respectively. It follows then that:
{\footnotesize{
\begin{equation}\label{eq_a:proof_eqcon_var}
\begin{split}
\Var(\btheta_{\text{eq}}^*) & = \Var \big( \bH_{\text{eq}}(\bx)^{-1}  \bd{\text{eq}}(\bx,\by) \big) \\
& = \bH_{\text{eq}}(\bx)^{-1} \Var \Big( \diag(\toi{\bx}) \bP^T \bF ( \bF^T \toi{\hat{\bV}} \bF )^{-1} \bF^T (\toi{\by} - \toi{\bV} (\bI - \bF ( \bF^T \toi{\hat{\bV}} \bF )^{-1} \bF^T \toi{\hat{\bV}} )\bz_0)  \Big) \bH_{\text{eq}}(\bx)^{-1} \\
& = \bH_{\text{eq}}(\bx)^{-1}  \frac{1}{m^2 \delta^2 } \sum_{i = 1}^m  \Big( \diag(\toi{\bx}) \bP^T \bF ( \bF^T \toi{\hat{\bV}} \bF )^{-1} \bF^T \Var(\toi{\by}) \bF ( \bF^T \toi{\hat{\bV}} \bF )^{-1} \bF^T  \bP \diag(\toi{\bx}) \Big) \bH_{\text{eq}}(\bx)^{-1}\\
& =  \bH_{\text{eq}}(\bx)^{-1}  \frac{1}{m^2 \delta^2 } \sum_{i = 1}^m  \Big( \diag(\toi{\bx}) \bP^T \bF ( \bF^T \toi{\hat{\bV}} \bF )^{-1} \bF^T \hat{\bSigma} \bF ( \bF^T \toi{\hat{\bV}} \bF )^{-1} \bF^T  \bP \diag(\toi{\bx}) \Big) \bH_{\text{eq}}(\bx)^{-1} \\
& = \bH_{\text{eq}}(\bx)^{-1} \bM_{\text{eq}} \bH_{\text{eq}}(\bx)^{-1}.
\end{split}
\end{equation}
}}

\section{Additional Experiments}\label{sec:app_exp}
\subsection{Experiment 2: unconstrained with multivariate predictions }\label{sec:results_2}

Economic performance metrics and average out-of-sample MVO costs are provided in Table \ref{table:ipo_uncon_multi} for the time period of {2000-01-01} to {2020-12-31} for the unconstrained MVO portfolios with multivariate prediction models. Equity growth charts for the same time period are provided in Figure \ref{fig:ipo_equity_uncon_multi}.  Again we  observe that the IPO model provides higher absolute and risk-adjusted performance and in general  more conservative risk metrics. The IPO model produces an out-of-sample MVO cost that is approximately $50\%$ lower and a Sharpe ratio that is approximately $100\%$ larger than that of the OLS model. In Figure \ref{fig:ipo_cost_uncon_multi} we compare the realized MVO and Sharpe ratio costs across $1000$ out-of-sample realizations. Again we observe that the IPO model exhibits consistently lower MVO costs with a dominance ratio of $99\%$ and generally lower Sharpe ratio costs with a dominance ratio of $65\%$.

In Figure \ref{fig:ipo_coef_uncon_multi} we report the estimated regression coefficients and $\pm 1$ standard error bar for the last out-of-sample data fold. As before, the majority of the IPO and OLS regression coefficients are not statistically significant at an individual market basis. Figures \ref{fig:ipo_coef_uncon_multi} (a) and \ref{fig:ipo_coef_uncon_multi} (b)  report the estimated regression coefficients for the Carry and Trend features, respectively. Again we observe that the IPO model provides very different regression coefficients, in both magnitude and sign, compared to the OLS coefficients.  Observe that in the multivariate regression model, $50\%$ of the OLS Trend coefficients are negative. In contrast, the IPO model has only 3 (12.5\%) negative coefficients: Cocoa (CC), Live Cattle (LC) and Platinum (PL). Furthermore, in many cases such as Feeder Cattle (FC) and Soymeal (SM), the OLS coefficients are relatively large $(>0.30)$ whereas the corresponding IPO coefficients are effectively zero. Lastly note that the magnitude of the coefficients is approximately $10$x larger than the corresponding trend coefficients and is a result of the carry feature values being approximately an order of magnitude smaller.

\begin{figure}[h]
  \includegraphics[width=\linewidth,height=3.8cm, trim={0mm 0cm 0cm 0cm},clip]{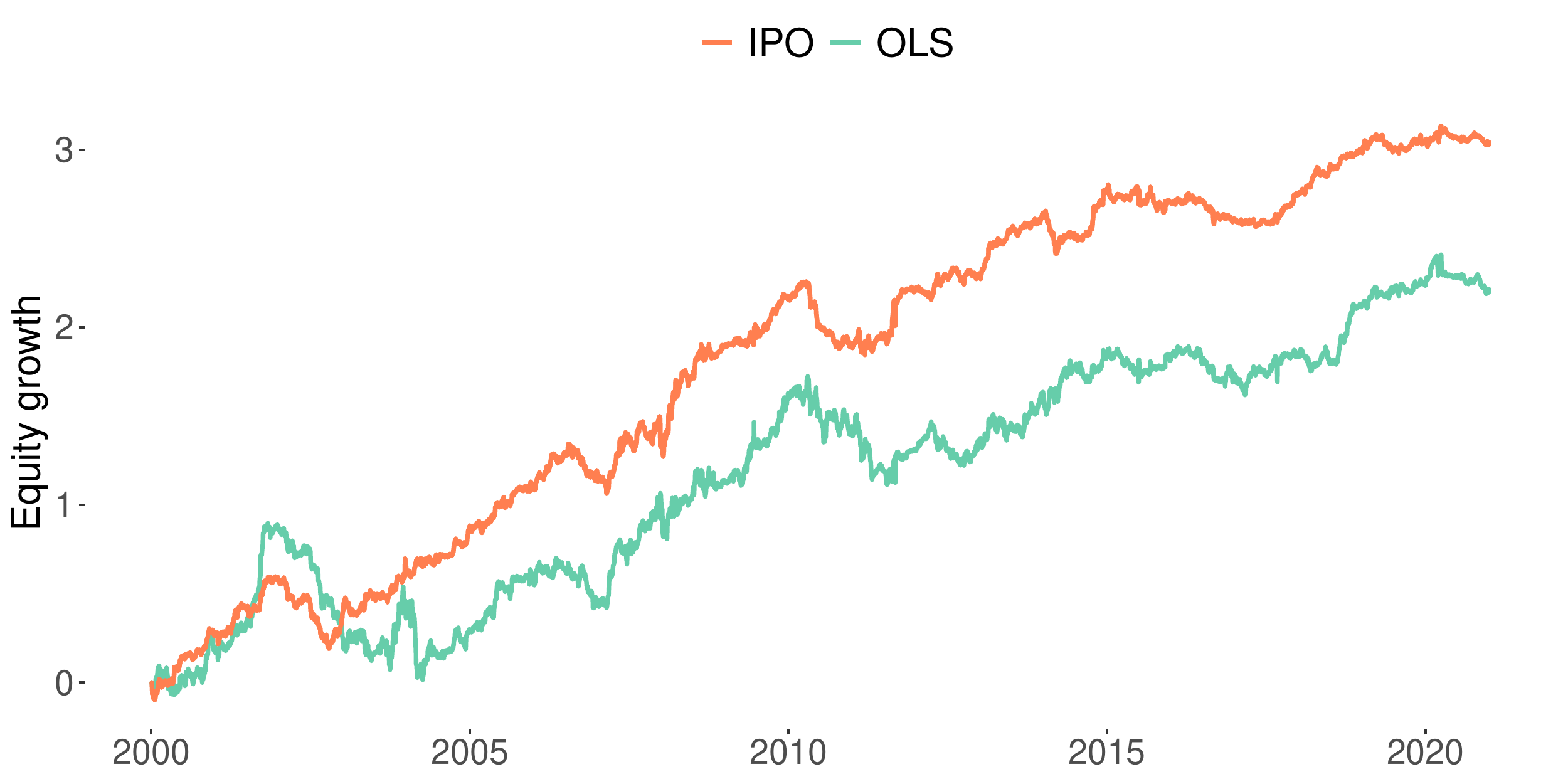}
   \caption{Out-of-sample log-equity growth for the unconstrained  mean-variance program and multivariate IPO and OLS prediction model. }
  \label{fig:ipo_equity_uncon_multi}
\end{figure}

\begin{table}[h]
\centering
\begin{tabular}{lrrrrrr}
\hline
  & Annual Return & Sharpe Ratio & Volatility & Avg Drawdown & Value at Risk & MVO Cost\\
\hline
IPO & 0.1416 & 0.8835 & 0.1603 & -0.0294 & -0.0138 & 0.5004\\
OLS & 0.1034 & 0.4477 & 0.2310 & -0.0438 & -0.0208 & 1.2308\\
\hline
\end{tabular}
\caption{Out-of-sample MVO costs and economic performance metrics for unconstrained mean-variance portfolios with multivariate IPO and OLS prediction models.}
\label{table:ipo_uncon_multi}
\end{table}

\begin{figure}[h]
  \centering
  \begin{subfigure}[b]{0.40\linewidth}
    \includegraphics[width=\linewidth , trim={0mm 0cm 0cm 0cm},clip]{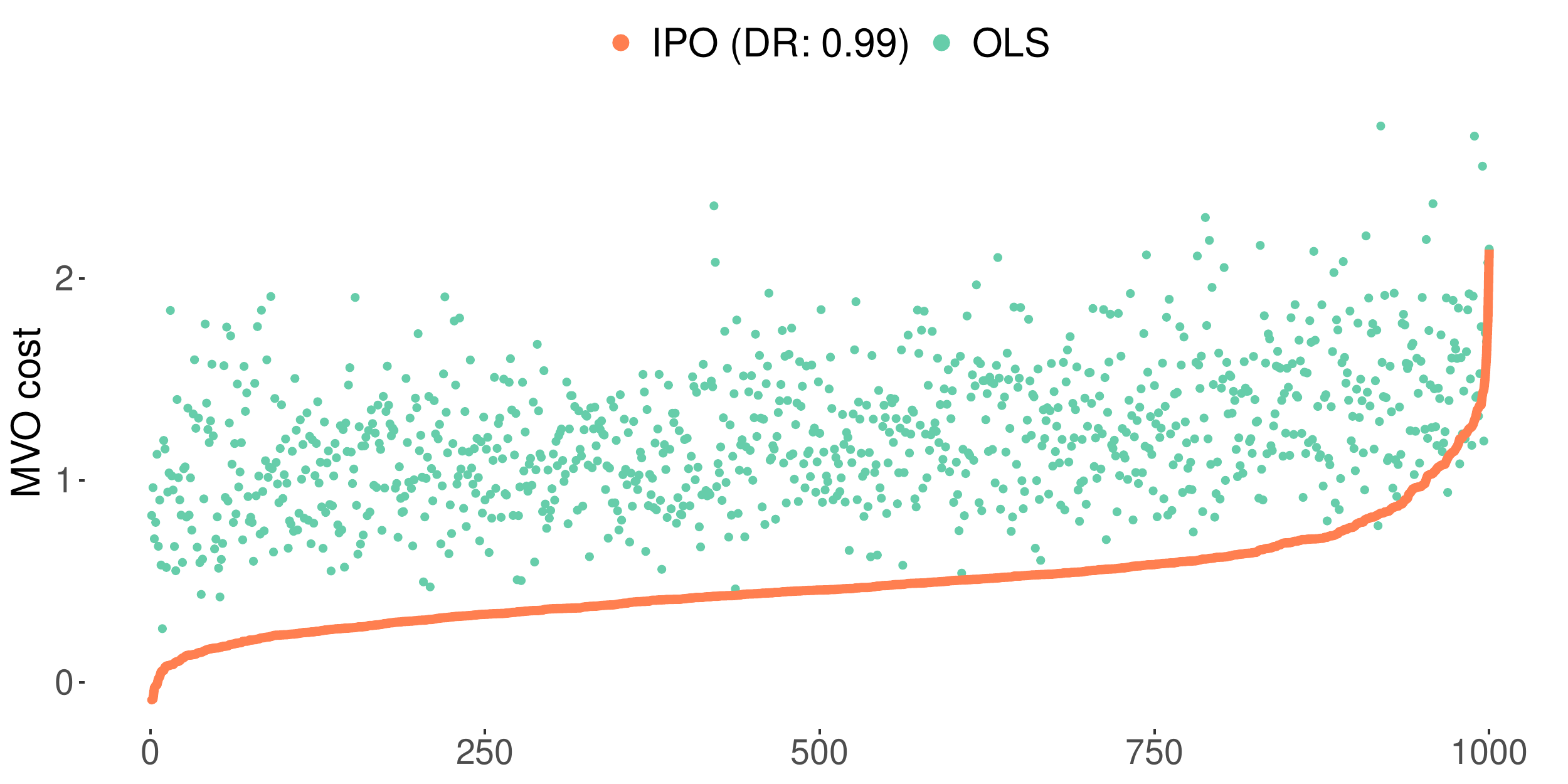}
    \caption{Out-of-sample MVO cost.}
  \end{subfigure}
  \begin{subfigure}[b]{0.40\linewidth}
    \includegraphics[width=\linewidth , trim={0mm 0cm 0cm 0cm},clip]{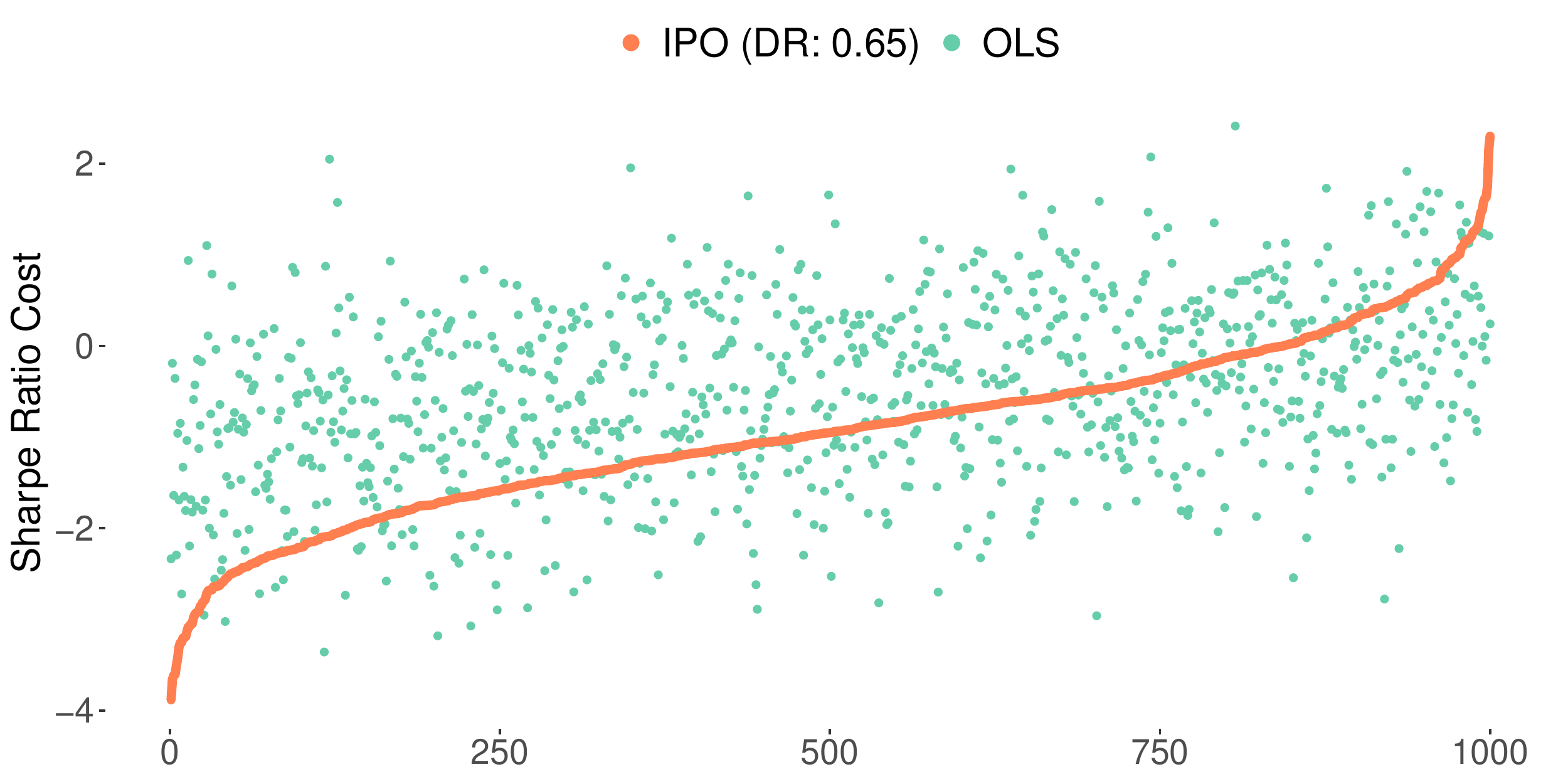}
    \caption{Out-of-sample Sharpe ratio cost.}
  \end{subfigure}
  \caption{Realized out-of-sample MVO and Sharpe ratio costs for the unconstrained  mean-variance program and multivariate IPO and OLS prediction models.}
  \label{fig:ipo_cost_uncon_multi}
\end{figure}

\begin{figure}[h]
\centering
\begin{subfigure}[b]{\linewidth}
\includegraphics[width=\linewidth,height = 3.8cm, trim={0mm 0cm 0mm 0cm},clip]{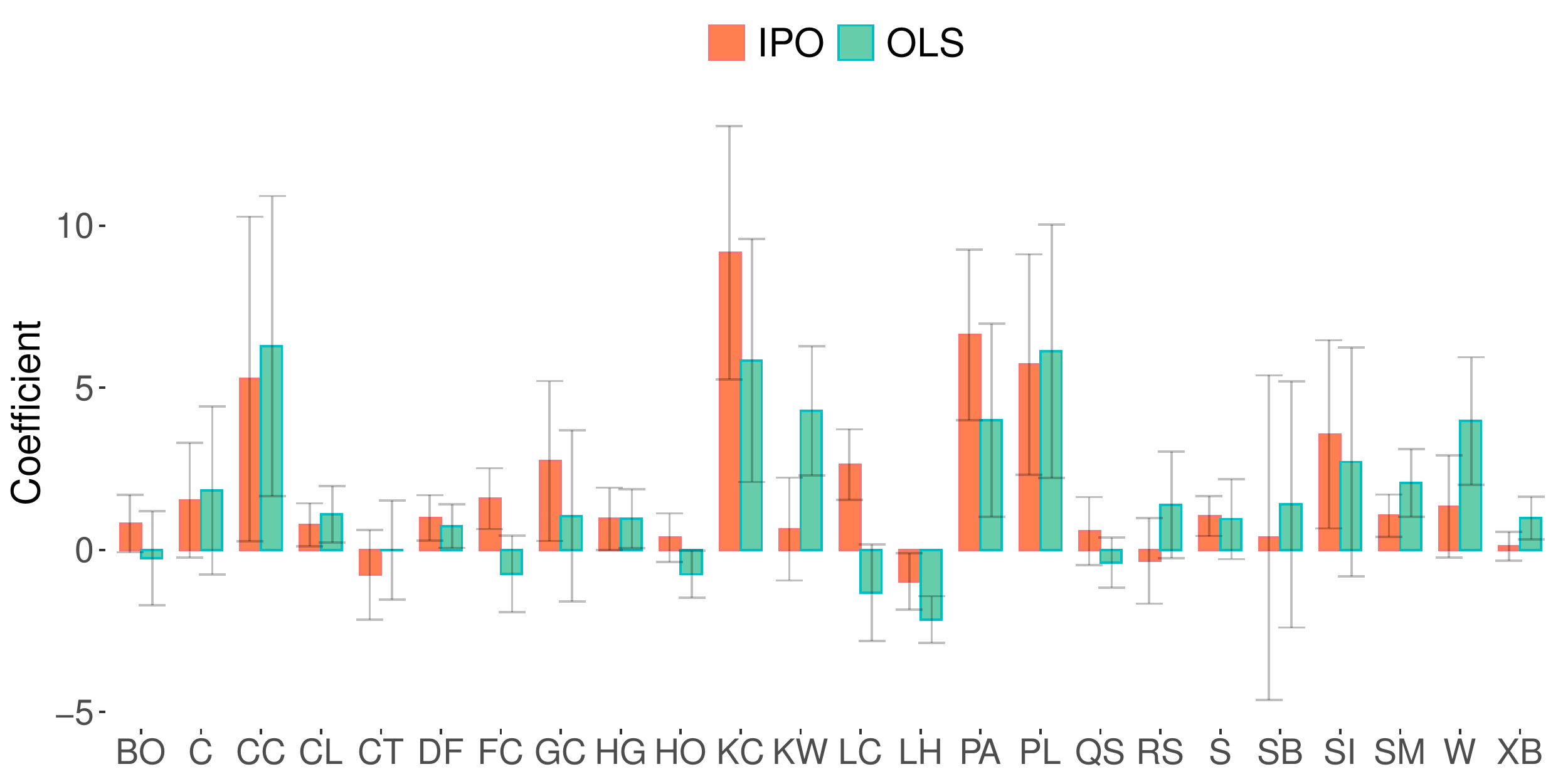}
\caption{Auxiliary feature: Carry.}
\end{subfigure}
\begin{subfigure}[b]{\linewidth}
\includegraphics[width=\linewidth,height = 3.8cm, trim={0mm 0cm 0mm 0cm},clip]{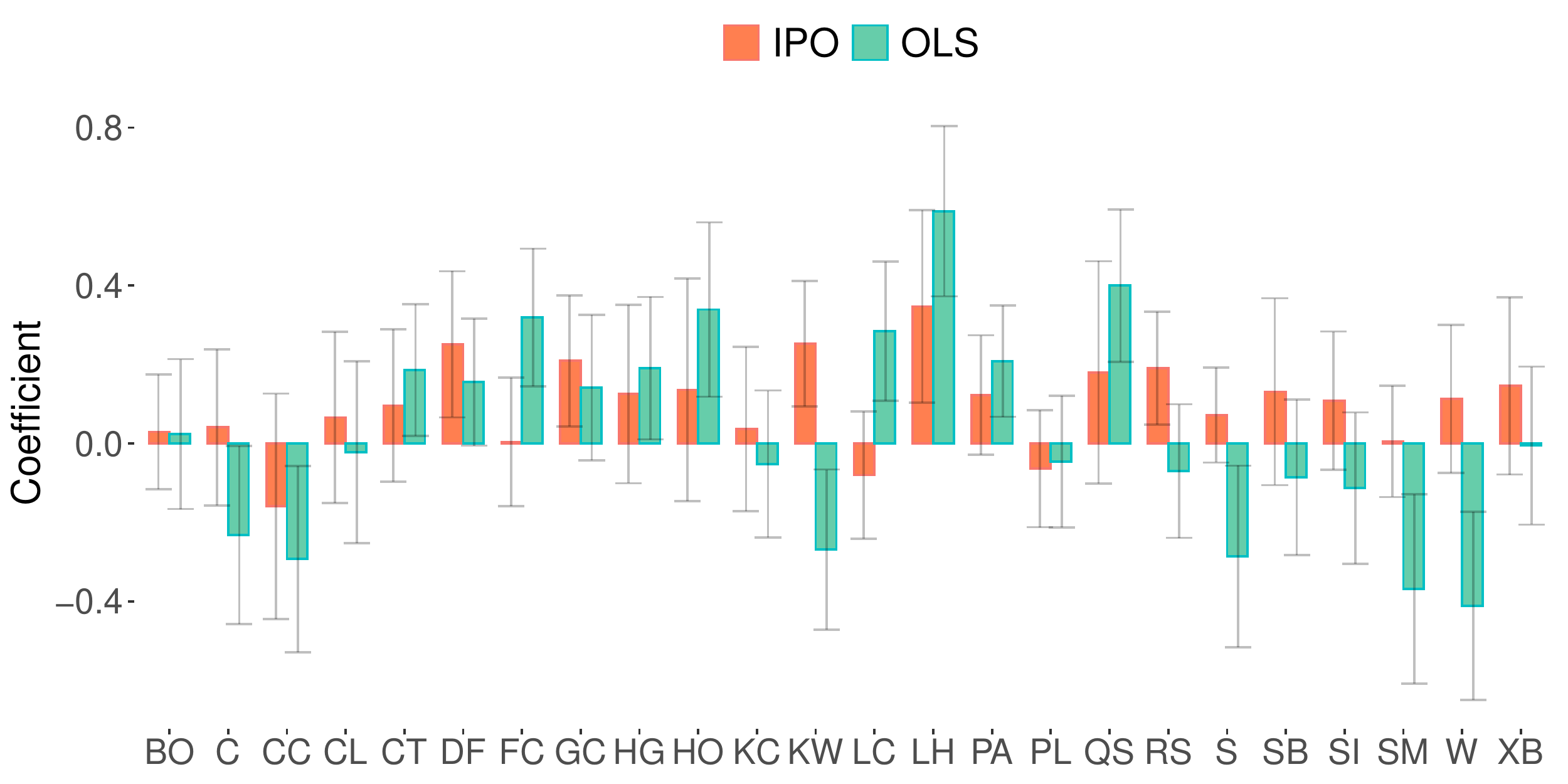}
\caption{Auxiliary feature: Trend.}
\end{subfigure}
\caption{Optimal IPO and OLS regression coefficients for the unconstrained  mean-variance program and multivariate prediction model. }
\label{fig:ipo_coef_uncon_multi}
\end{figure}

\subsection{Experiment 3 and 4: equality constrained with univariate and multivariate predictions}
Economic performance metrics and average out-of-sample MVO costs are provided in Tables \ref{table:ipo_eqcon_uni} and \ref{table:ipo_eqcon_multi} for the equality constrained MVO portfolios with univariate and multivariate prediction models, respectively. Equity growth charts for the time period of {2000-01-01} to {2020-12-31} are provided in Figures \ref{fig:ipo_equity_eqcon_uni} and \ref{fig:ipo_equity_eqcon_multi}.  As in the unconstrained case, we  observe that the IPO model provides higher absolute and risk-adjusted performance, and in general produces more conservative risk metrics. Figures \ref{fig:ipo_cost_eqcon_uni} and \ref{fig:ipo_cost_eqcon_multi} demonstrate that the IPO model produces consistently lower out-of-sample MVO costs, with dominance ratios of $93\%$ and $99\%$, respectively, and generally lower Sharpe ratio costs with dominance ratios of  $67\%$ and $66\%$, respectively. The regression coefficients are identical to those provided in Figures \ref{fig:ipo_coef_eqcon_uni} and \ref{fig:ipo_coef_eqcon_multi}, and we refer to Section \ref{sec:results} for relevant discussion.

\begin{figure}[h]
  \includegraphics[width=\linewidth,height=3.8cm, trim={0mm 0cm 0cm 0cm},clip]{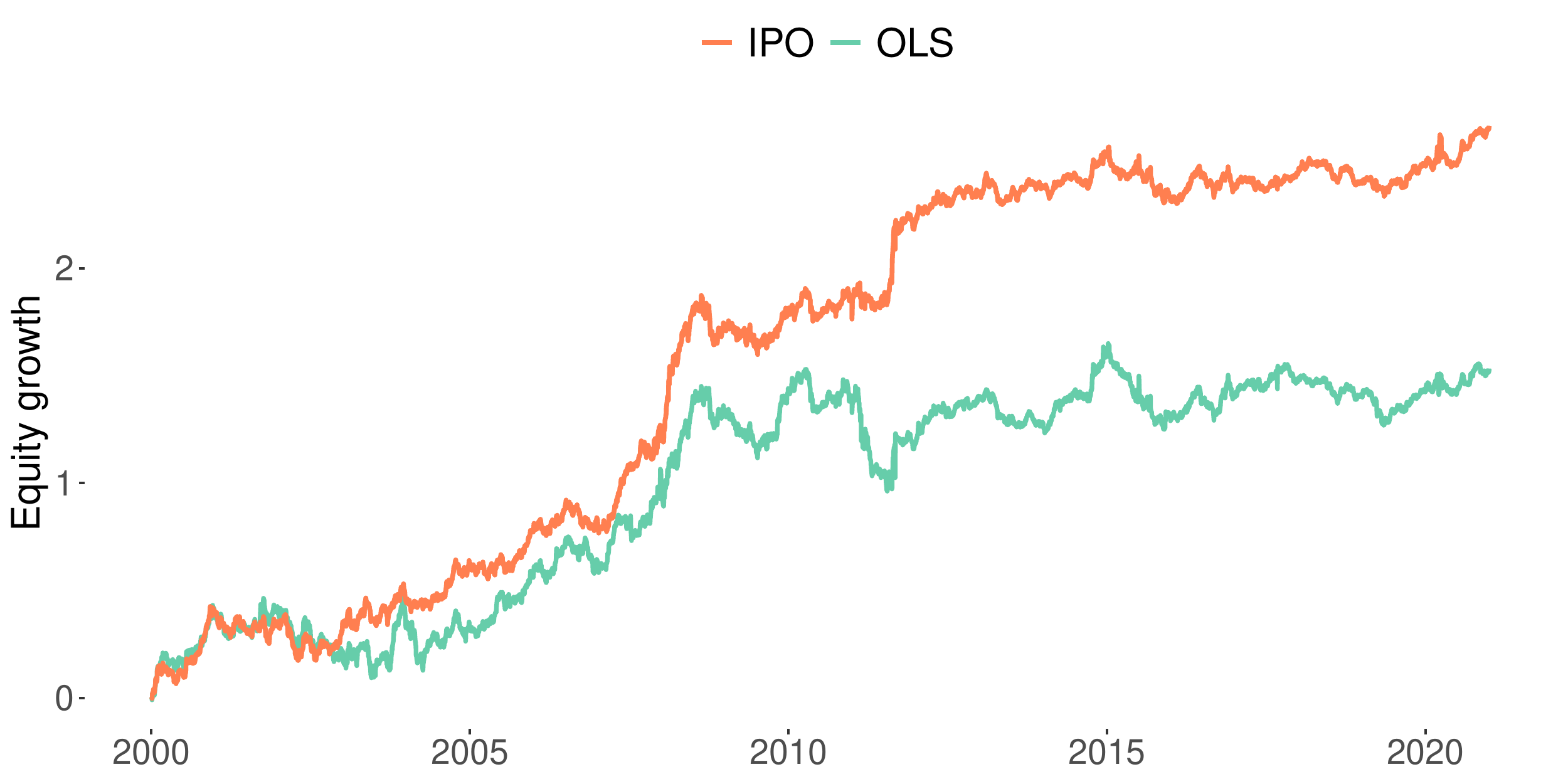}
   \caption{Out-of-sample log-equity growth for the equality constrained  mean-variance program and multivariate IPO and OLS prediction model. }
  \label{fig:ipo_equity_eqcon_uni}
\end{figure}

\begin{table}[h]
\centering
\begin{tabular}{lrrrrrr}
\hline
  & Annual Return & Sharpe Ratio & Volatility & Avg Drawdown & Value at Risk & MVO Cost\\
\hline
IPO & 0.1238 & 0.7665 & 0.1616 & -0.0290 & -0.0142 & 0.5288\\
OLS & 0.0713 & 0.3803 & 0.1876 & -0.0471 & -0.0170 & 0.8082\\
\hline
\end{tabular}
\caption{Out-of-sample MVO costs and economic performance metrics for equality constrained mean-variance portfolios with univariate IPO and OLS prediction models.}
\label{table:ipo_eqcon_uni}
\end{table}

\begin{figure}[h]
  \centering
  \begin{subfigure}[b]{0.40\linewidth}
    \includegraphics[width=\linewidth , trim={0mm 0cm 0cm 0cm},clip]{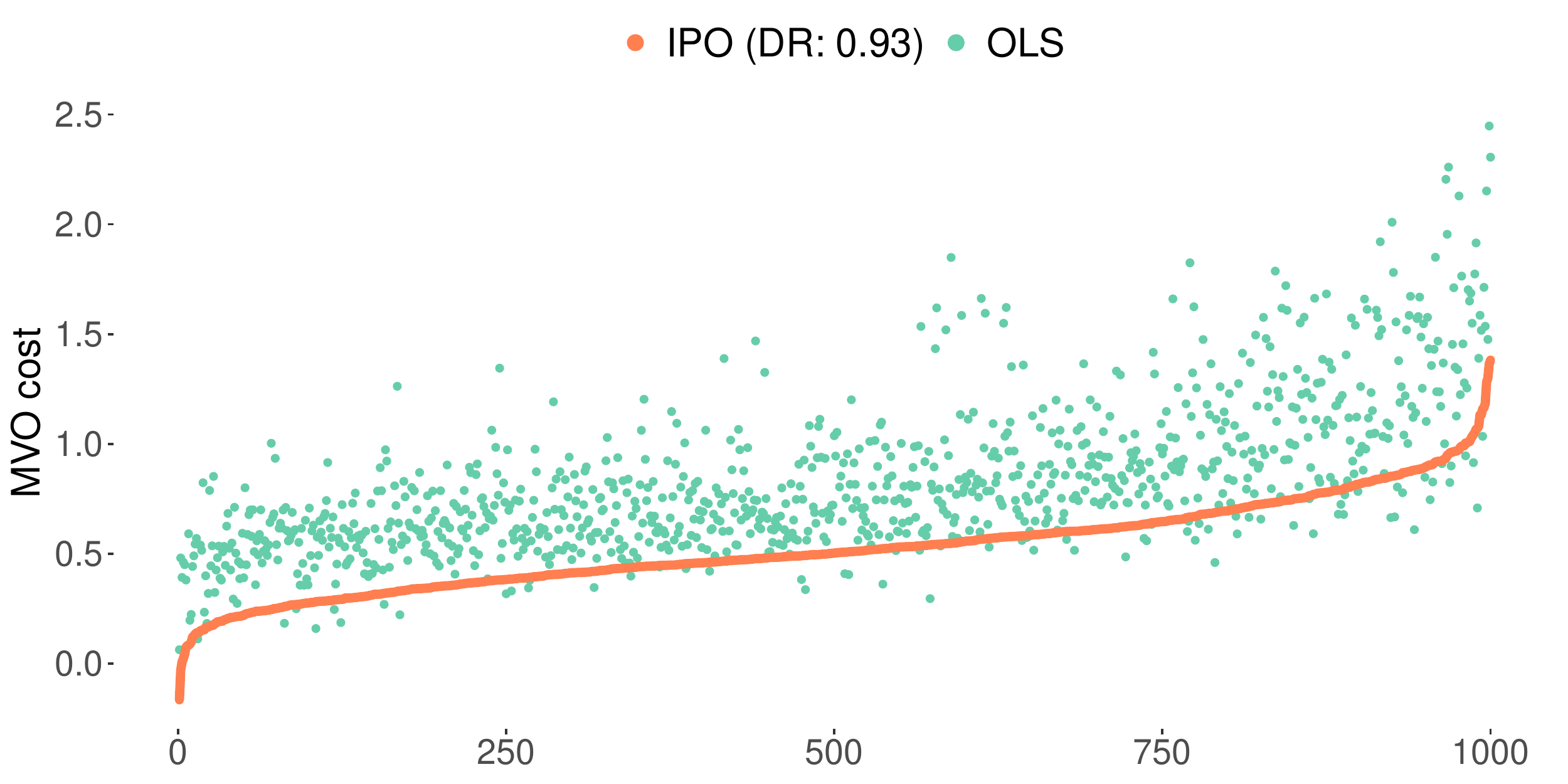}
    \caption{Out-of-sample MVO cost.}
  \end{subfigure}
  \begin{subfigure}[b]{0.40\linewidth}
    \includegraphics[width=\linewidth , trim={0mm 0cm 0cm 0cm},clip]{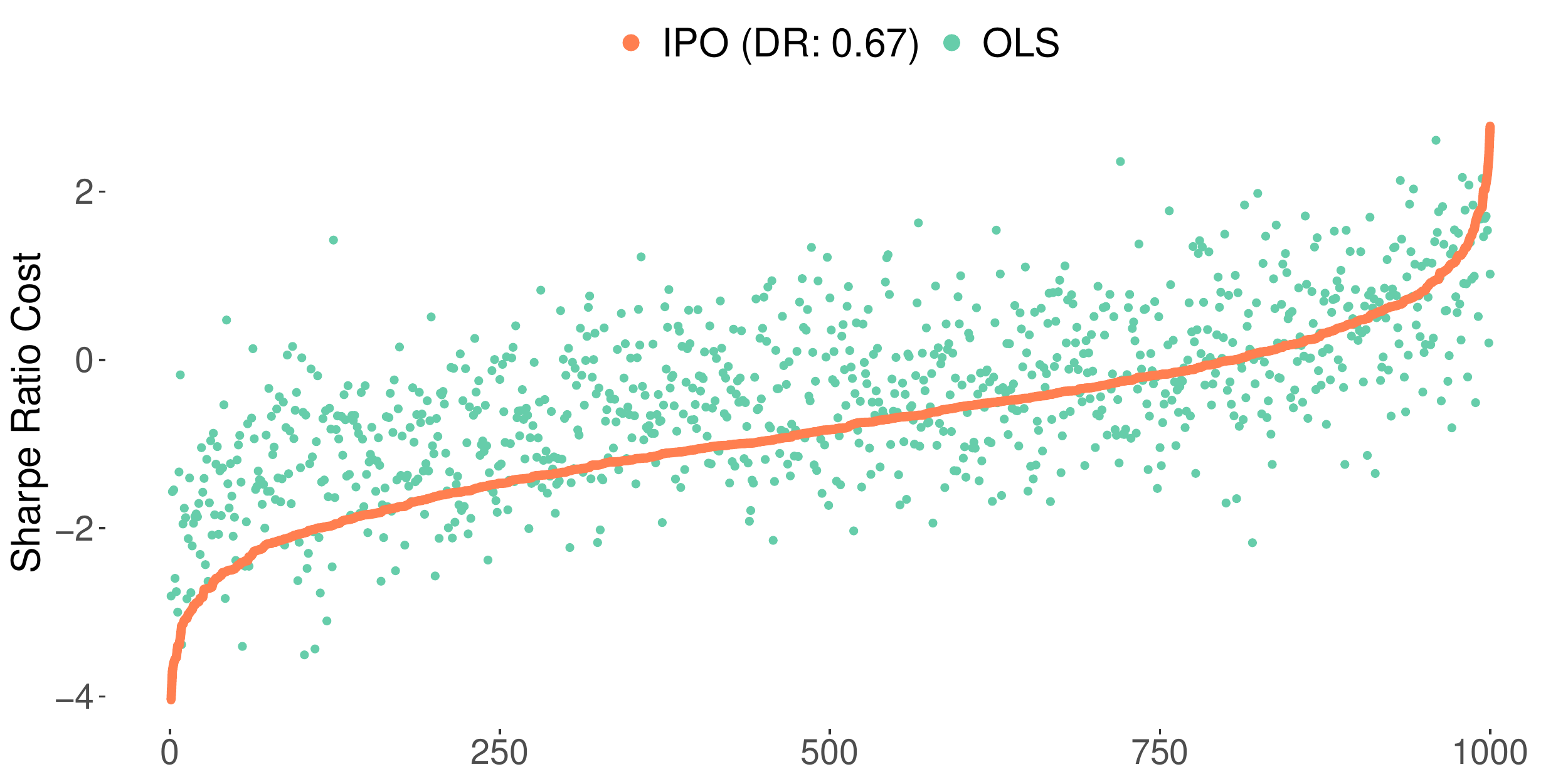}
    \caption{Out-of-sample Sharpe ratio cost.}
  \end{subfigure}
  \caption{Realized out-of-sample MVO and Sharpe ratio costs for the equality constrained  mean-variance program and univariate IPO and OLS prediction models.}
  \label{fig:ipo_cost_eqcon_uni}
\end{figure}


\begin{figure}[h]
  \includegraphics[width=\linewidth,height=3.8cm, trim={0mm 0cm 0cm 0cm},clip]{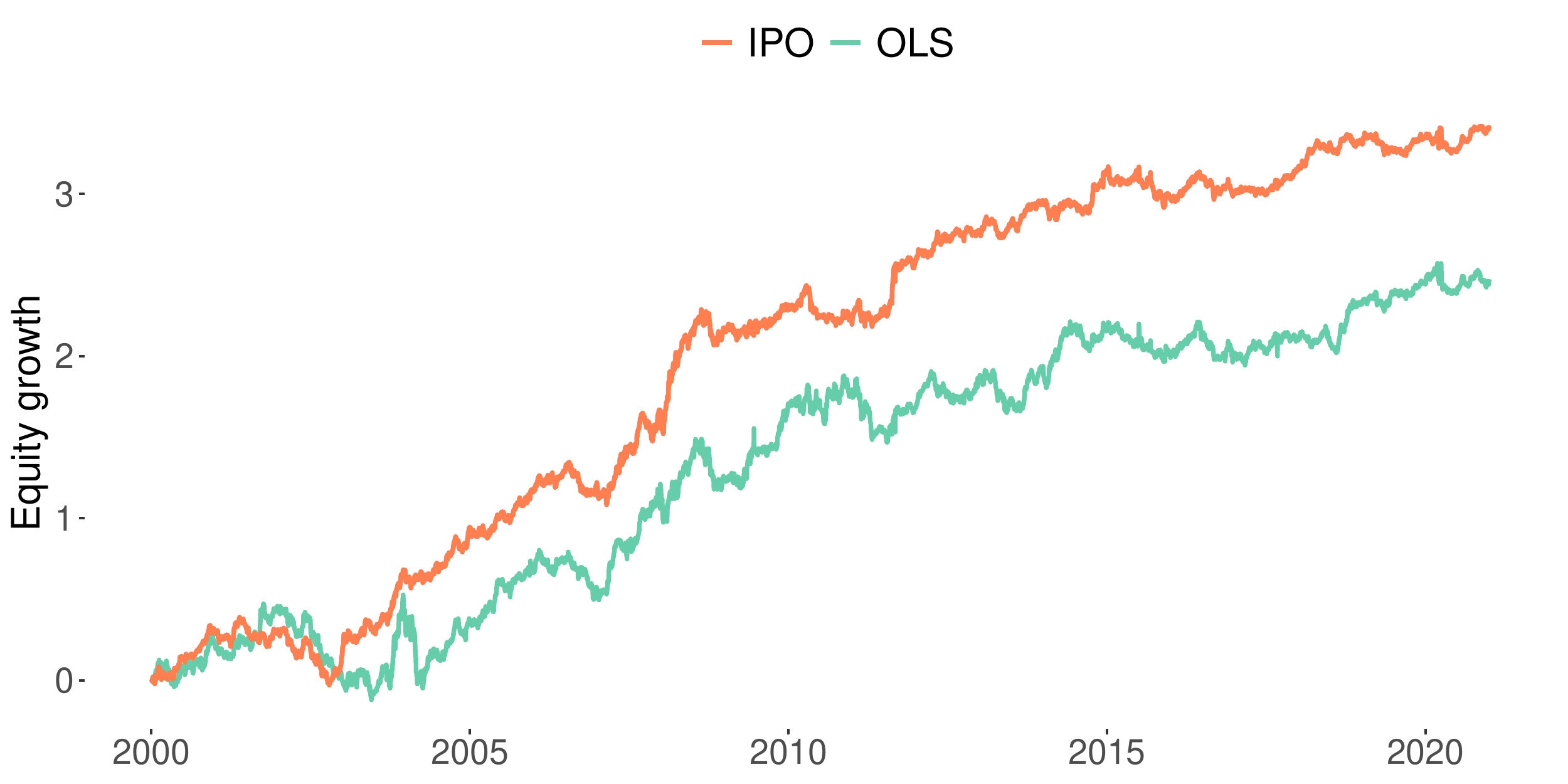}
   \caption{Out-of-sample log-equity growth for the equality constrained  mean-variance program and multivariate IPO and OLS prediction model. }
  \label{fig:ipo_equity_eqcon_multi}
\end{figure}

\begin{table}[h]
\centering
\begin{tabular}{lrrrrrr}
\hline
  & Annual Return & Sharpe Ratio & Volatility & Avg Drawdown & Value at Risk & MVO Cost\\
\hline
IPO & 0.1590 & 0.8851 & 0.1797 & -0.0339 & -0.0163 & 0.6482\\
\hline
OLS & 0.1151 & 0.4784 & 0.2406 & -0.0497 & -0.0215 & 1.3315\\
\hline
\end{tabular}
\caption{Out-of-sample MVO costs and economic performance metrics for equality constrained mean-variance portfolios with multivariate IPO and OLS prediction models.}
\label{table:ipo_eqcon_multi}
\end{table}

\begin{figure}[h]
  \centering
  \begin{subfigure}[b]{0.40\linewidth}
    \includegraphics[width=\linewidth , trim={0mm 0cm 0cm 0cm},clip]{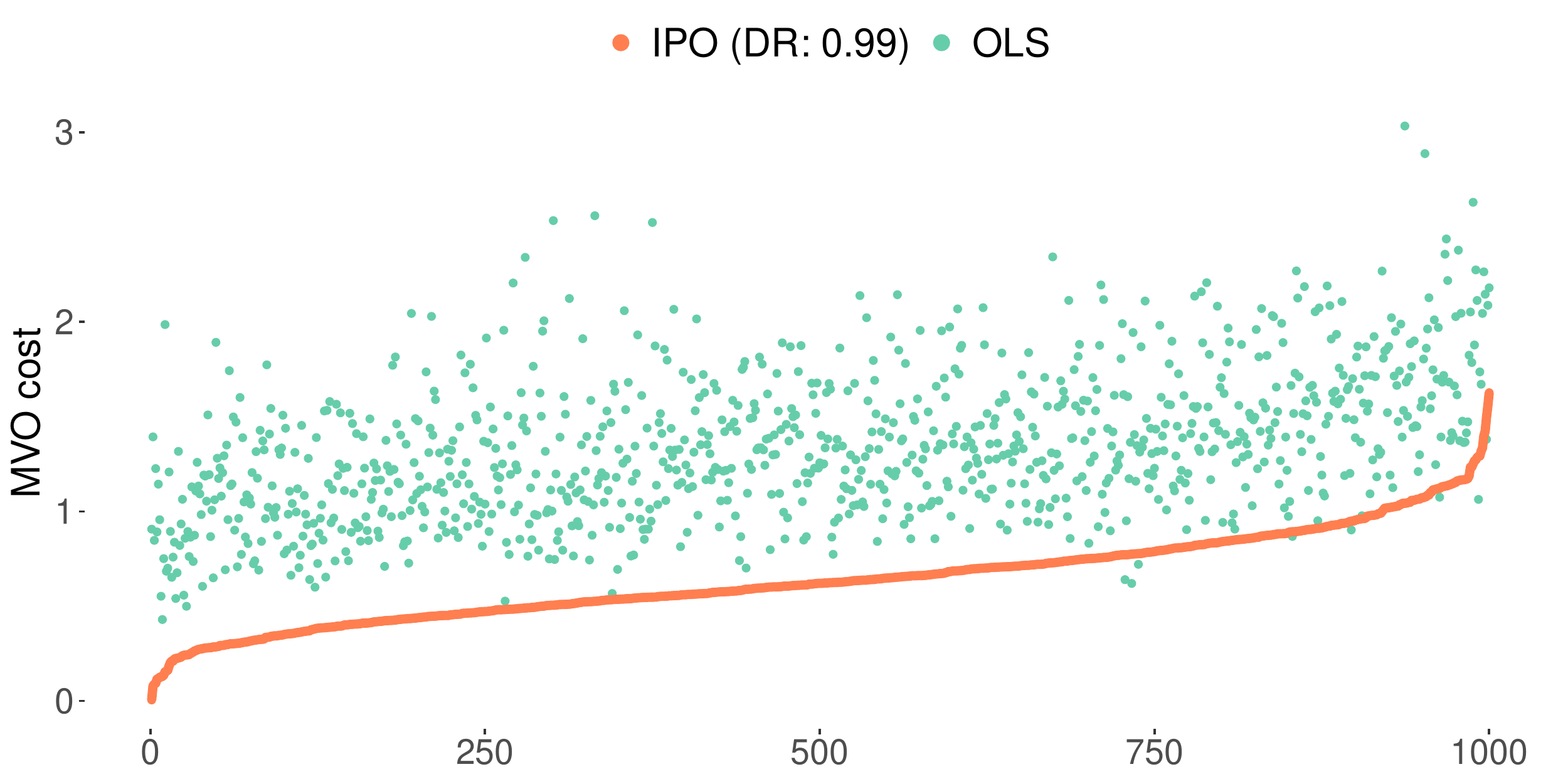}
    \caption{Out-of-sample MVO cost.}
  \end{subfigure}
  \begin{subfigure}[b]{0.40\linewidth}
    \includegraphics[width=\linewidth , trim={0mm 0cm 0cm 0cm},clip]{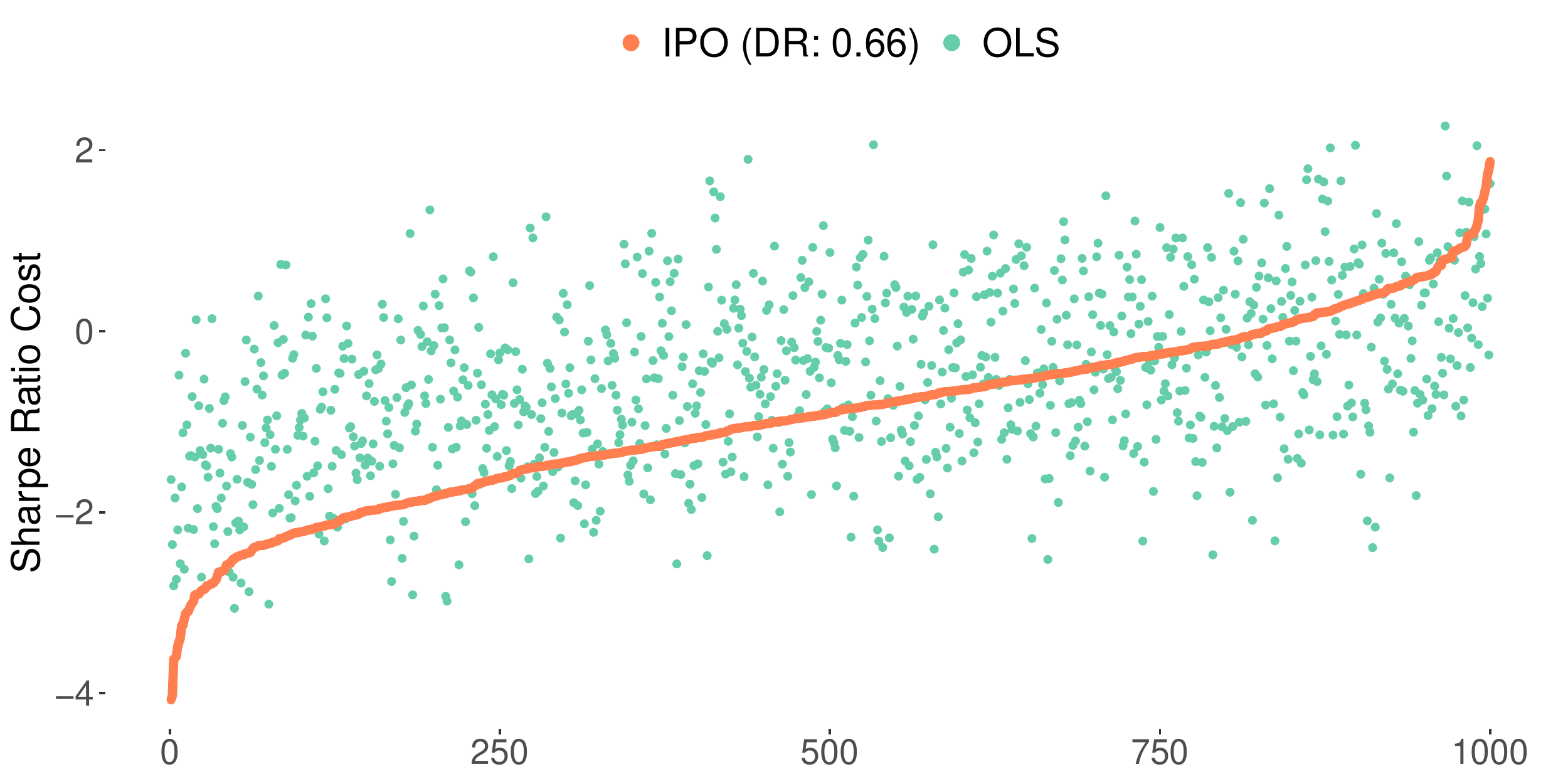}
    \caption{Out-of-sample Sharpe ratio cost.}
  \end{subfigure}
  \caption{Realized out-of-sample MVO and Sharpe ratio costs for the equality constrained  mean-variance program and multivariate IPO and OLS prediction models.}
  \label{fig:ipo_cost_eqcon_multi}
\end{figure}

\subsection{Experiment 6: inequality constrained with multivariate predictions}
Economic performance metrics and average out-of-sample MVO costs are provided in Table \ref{table:ipo_ineqcon_multi} for the time period of {2000-01-01} to {2020-12-31} for the inequality constrained MVO portfolios with multivariate prediction models. Equity growth charts for the same time period are provided in Figure \ref{fig:equity_ineqcon_multi}. Once again we  observe that the IPO model provides modestly higher absolute and risk-adjusted performance and in general  more conservative risk metrics. The IPO model produces an out-of-sample MVO cost that is approximately $60\%$ lower and a Sharpe ratio that is approximately $25\%$ larger than that of the OLS model. In Figure \ref{fig:ipo_cost_ineqcon_multi} we compare the realized MVO and Sharpe ratio costs across $1000$ out-of-sample realizations. Again we observe more modest dominance ratios with values in the $55\%$-$65\%$ range. We observe that the IPO model provides a modest improvement to performance in comparison to the OLS model; a likely result of lower prediction model misspecification and improved portfolio regularization by virtue of the box constraints. The estimated regression coefficients are provided in Figure \ref{fig:ipo_coef_eqcon_multi} and the findings are similar to those described in Section \ref{sec:results_2}.

\begin{figure}[h]
  \includegraphics[width=\linewidth,height=3.8cm, trim={0mm 0cm 0cm 0cm},clip]{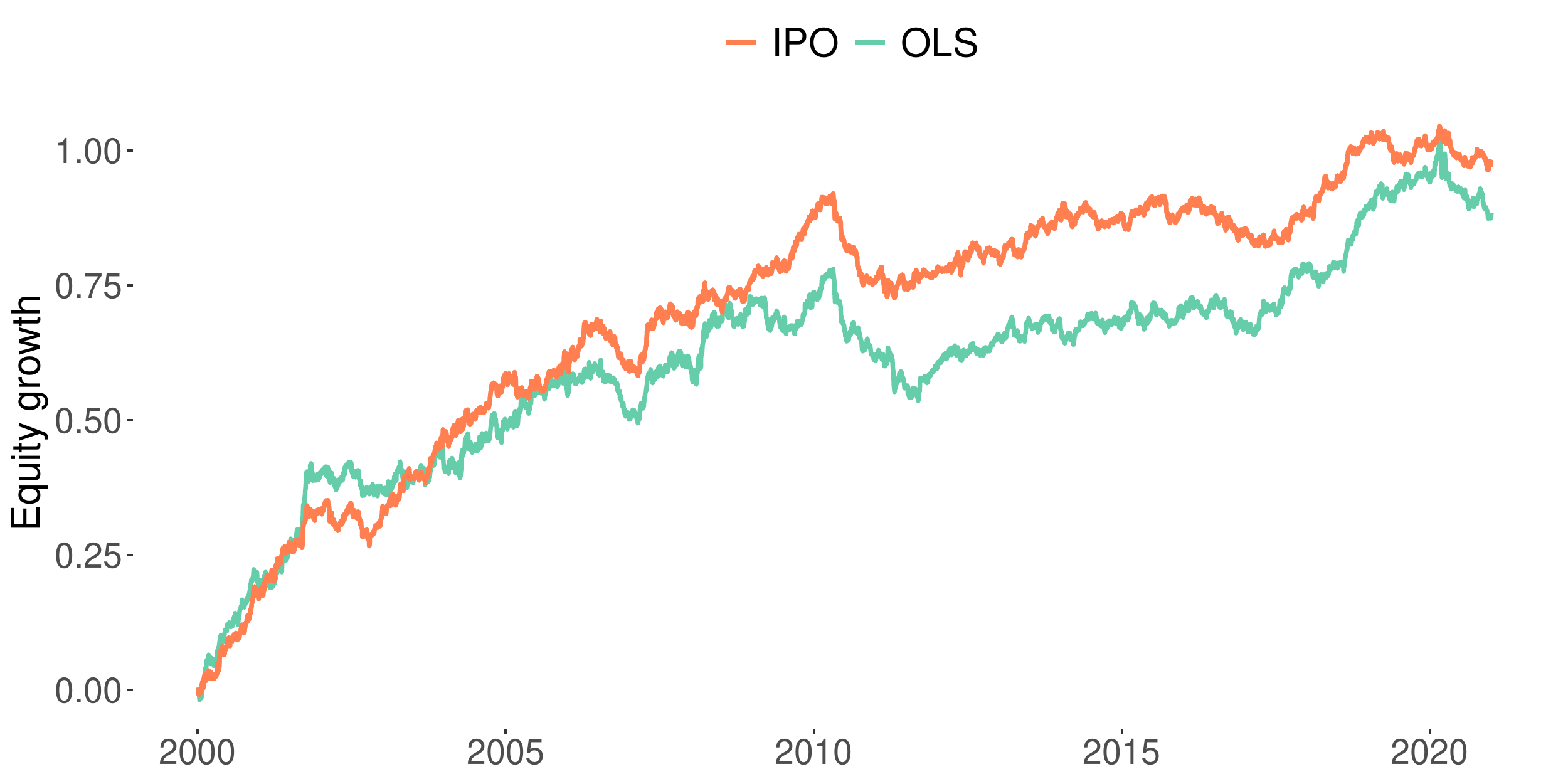}
   \caption{Out-of-sample log-equity growth for the inequality constrained  mean-variance program and multivariate IPO and OLS prediction model. }
  \label{fig:equity_ineqcon_multi}
\end{figure}

\begin{table}[h]
\centering
\begin{tabular}{lrrrrrr}
\hline
  & Annual Return & Sharpe Ratio & Volatility & Avg Drawdown & Value at Risk & MVO Cost\\
\hline
IPO & 0.0456 & 0.7937 & 0.0574 & -0.0119 & -0.0057 & 0.0369\\
OLS & 0.0411 & 0.6488 & 0.0634 & -0.0145 & -0.0063 & 0.0593\\
\hline
\end{tabular}
\caption{Out-of-sample MVO costs and economic performance metrics for inequality constrained mean-variance portfolios with multivariate IPO and OLS prediction models.}
\label{table:ipo_ineqcon_multi}
\end{table}

\begin{figure}[h]
  \centering
  \begin{subfigure}[b]{0.40\linewidth}
    \includegraphics[width=\linewidth , trim={0mm 0cm 0cm 0cm},clip]{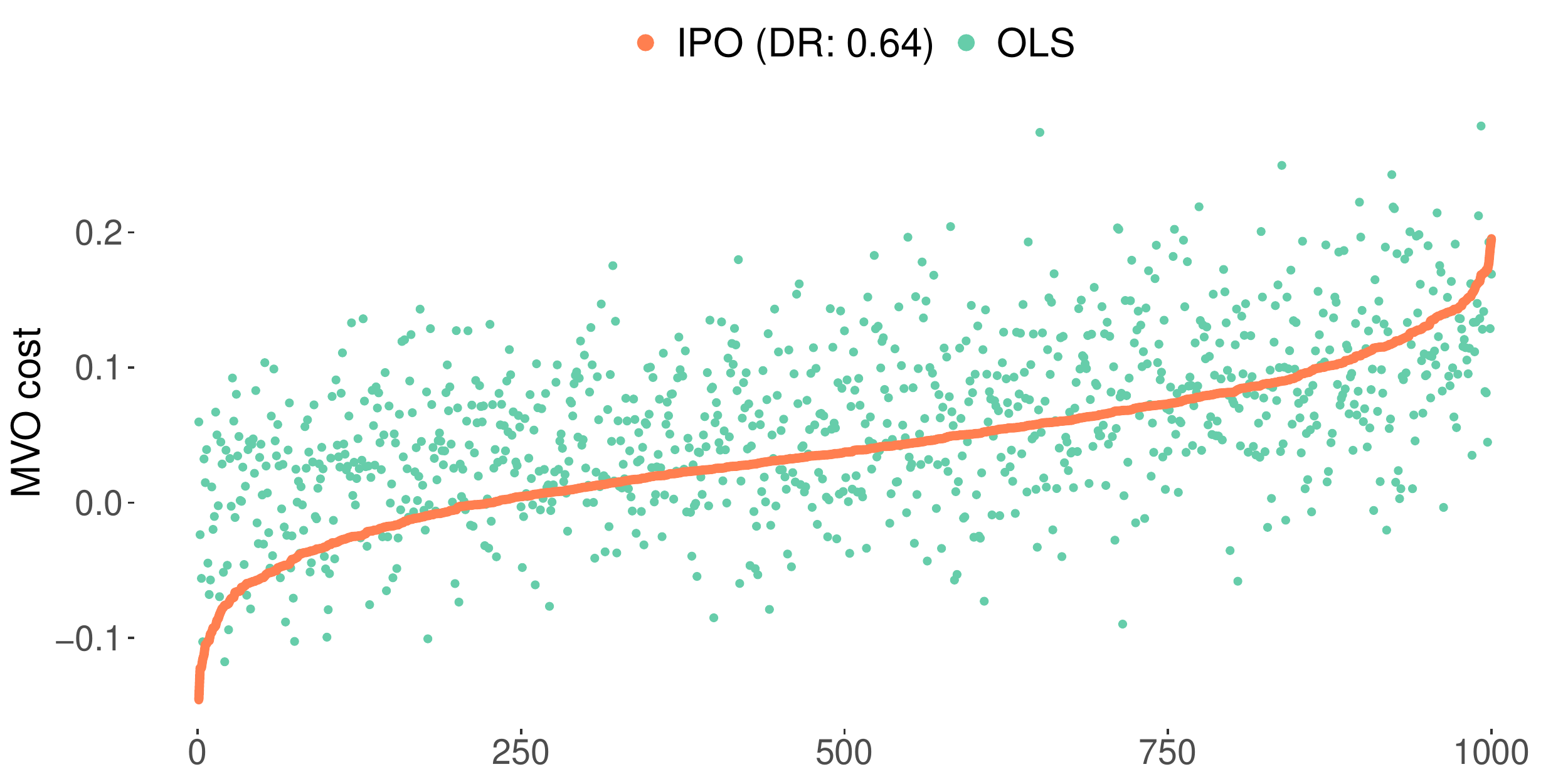}
    \caption{Out-of-sample MVO cost.}
  \end{subfigure}
  \begin{subfigure}[b]{0.40\linewidth}
    \includegraphics[width=\linewidth , trim={0mm 0cm 0cm 0cm},clip]{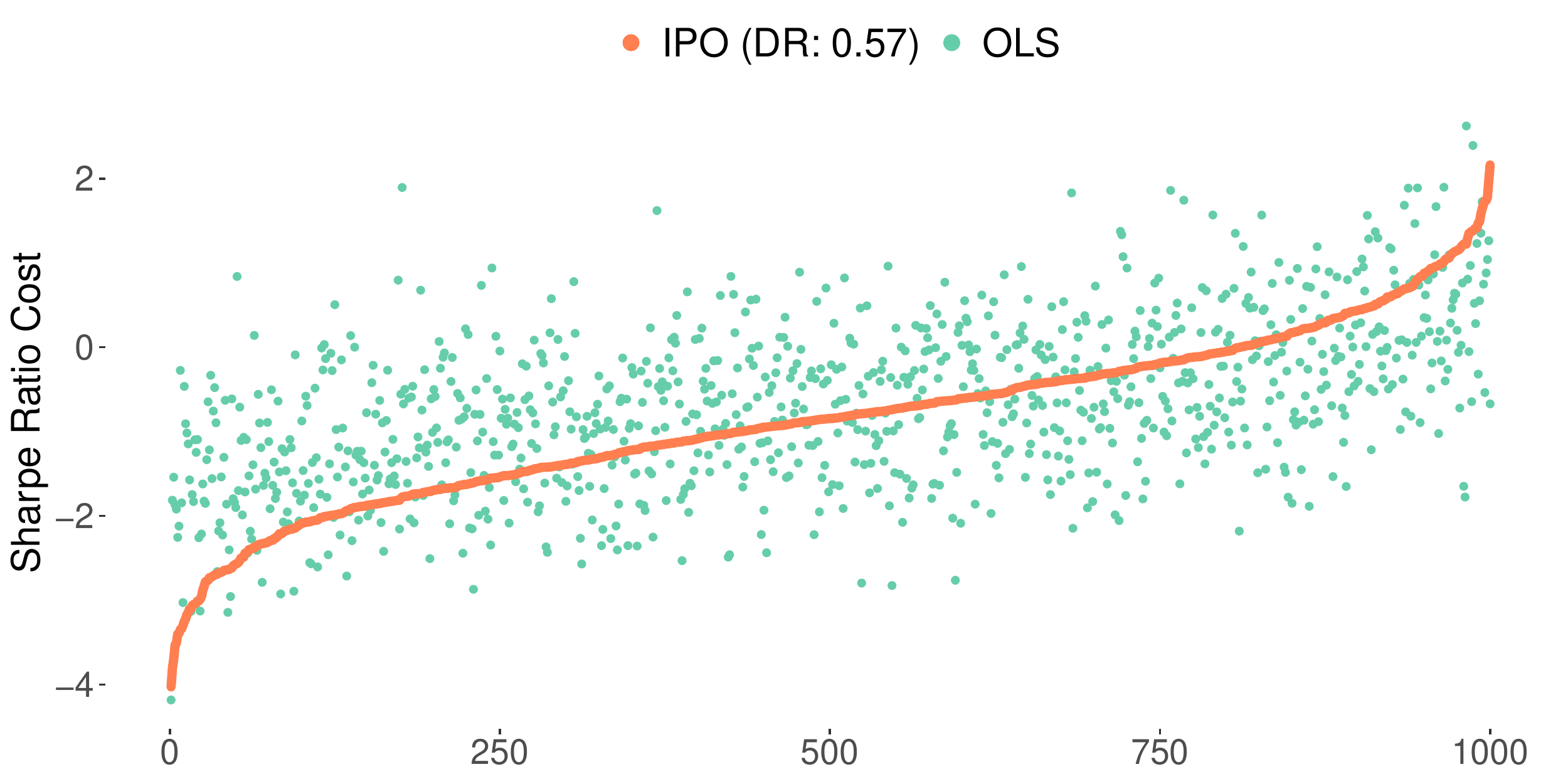}
    \caption{Out-of-sample Sharpe ratio cost.}
  \end{subfigure}
  \caption{Realized out-of-sample MVO and Sharpe ratio costs for the inequality constrained  mean-variance program and multivariate IPO and OLS prediction models.}
  \label{fig:ipo_cost_ineqcon_multi}
\end{figure}

\begin{figure}[h]
\centering
\begin{subfigure}[b]{\linewidth}
\includegraphics[width=\linewidth,height = 3.8cm, trim={0mm 0cm 0mm 0cm},clip]{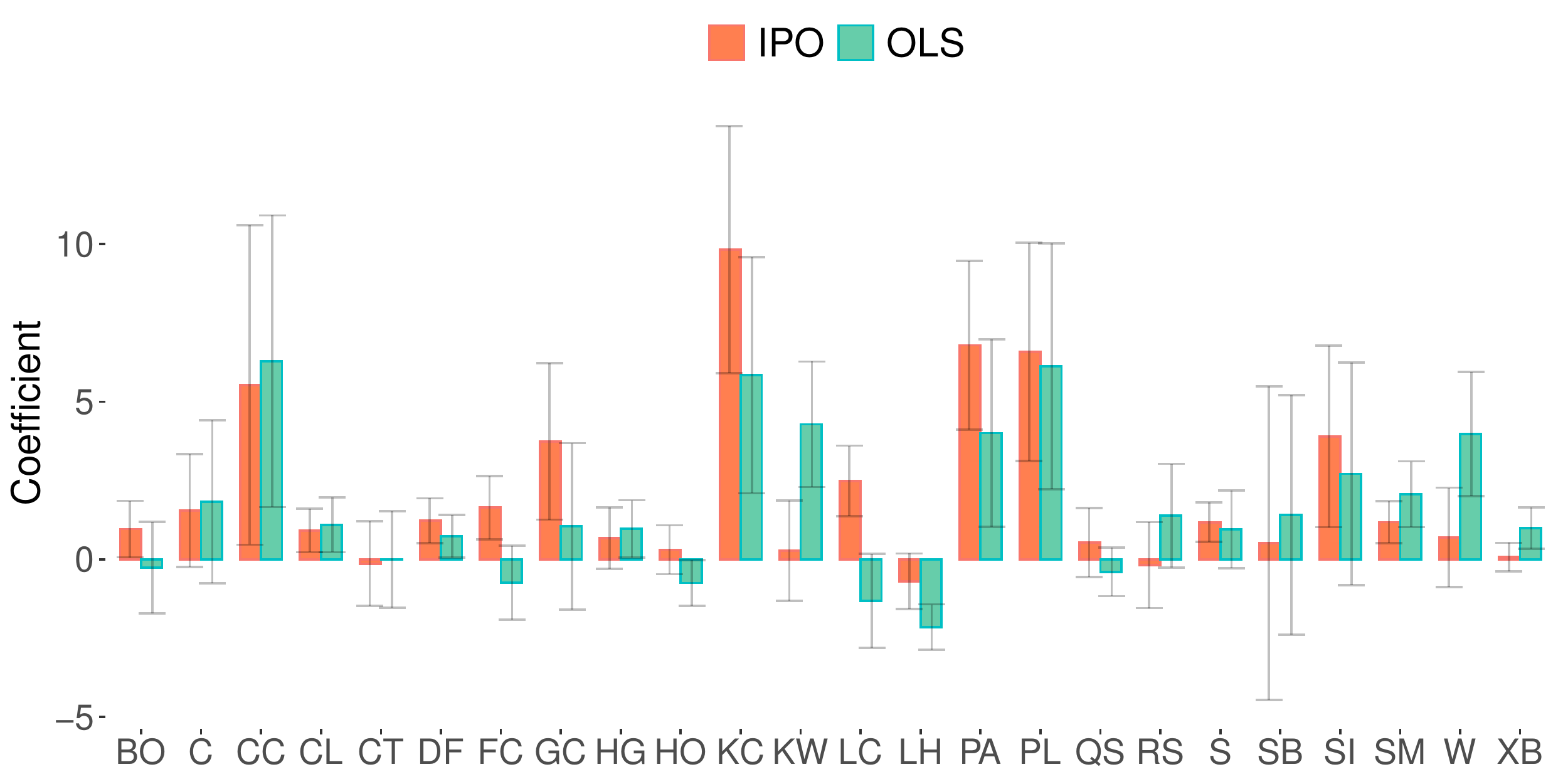}
\caption{Auxiliary feature: Carry.}
\end{subfigure}
\begin{subfigure}[b]{\linewidth}
\includegraphics[width=\linewidth,height = 3.8cm, trim={0mm 0cm 0mm 0cm},clip]{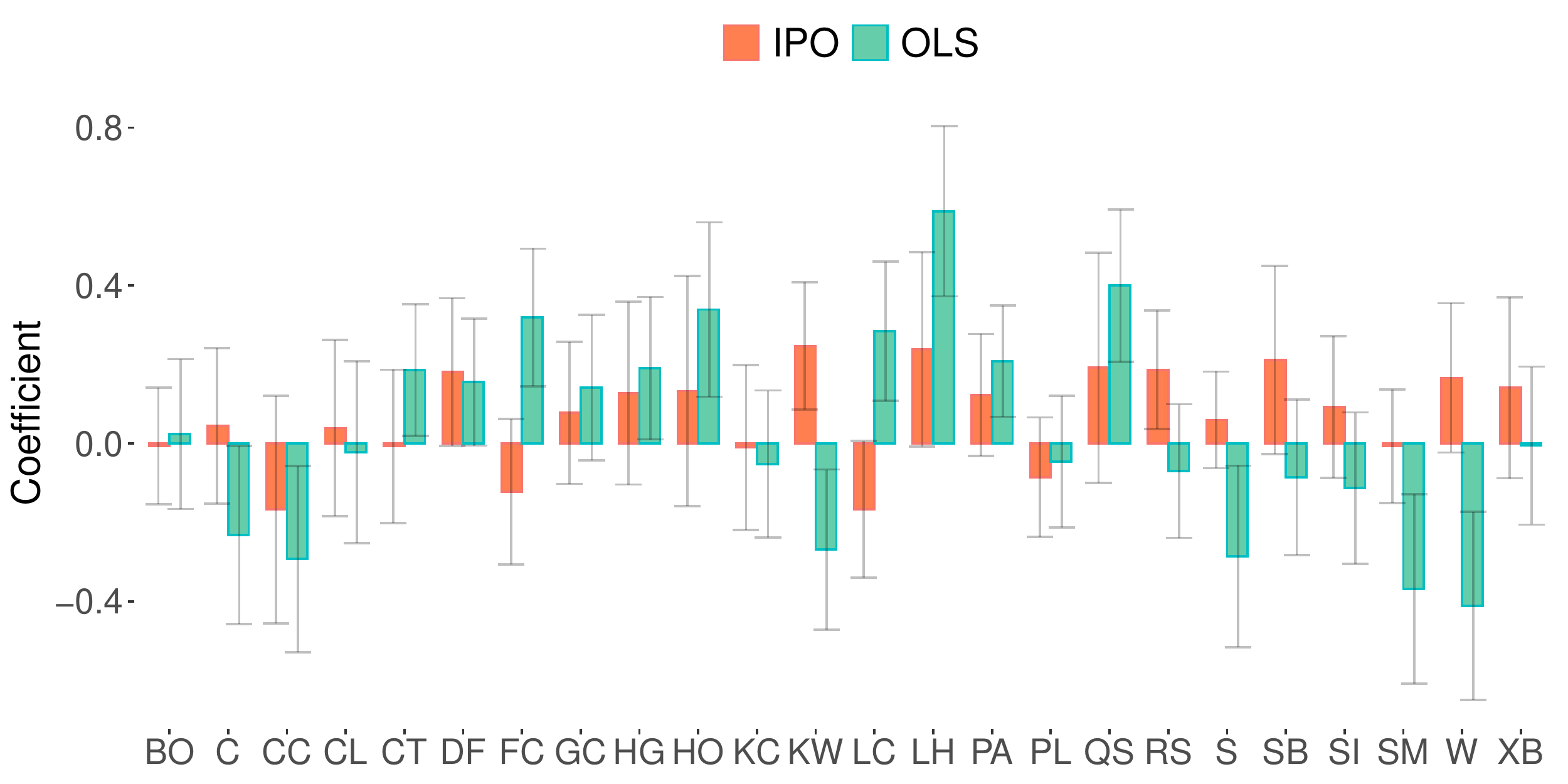}
\caption{Auxiliary feature: Trend.}
\end{subfigure}
\caption{Optimal IPO and OLS regression coefficients for the equality constrained  mean-variance program and multivariate prediction model. }
\label{fig:ipo_coef_eqcon_multi}
\end{figure}

\section{Experiment details}
All experiments were conducted on an Apple Mac Pro computer ($2.7$ GHz $12$-Core Intel Xeon E5,$128$ GB $1066$ MHz DDR3 RAM) running macOS `Catalina'. The software was written using the R programming language (version $4.0.0$) and  torch (version $0.2.0$).

\begin{table}[h]
\centering
\begin{tabular}{ l l l l l}
\hline
\textbf{Asset Class} &  & & \textbf{Market (Symbol)} &  \\
\hline
Energy & & WTI crude (CL)       & Heating oil (HO) & Gasoil (QS)\\
      &  & RBOB gasoline (XB)   &                  &            \\
\\
Grain & & Bean oil (BO)         & Corn (C)         & KC Wheat (KW)\\
      & & Soybean (S)           & Soy meal (SM)    & Wheat (W)\\
\\
Livestock & & Feeder cattle (FC) & Live cattle (LC) &  Lean hogs (LH)\\
\\
Metal   &  & Gold (GC)          & Copper (HG)      & Palladium (PA)\\
        &  & Platinum (PL)      & Silver (SI)      & \\
\\

Soft   &   & Cocoa (CC)         & Cotton (CT)      & Robusta Coffee (DF)\\
       &  & Coffee (KC)        & Canola (RS)       & Sugar (SB)  \\

\hline
\end{tabular}
\caption{Futures market universe. Symbols follow Bloomberg market symbology. Data is provided by Commodity Systems Inc (CSI).}
\label{table:universe}
\end{table}

\end{document}